\documentclass{article}
\usepackage{arxiv_preprint,times}
\usepackage{amsmath,amssymb,amsthm,array,booktabs,longtable,graphicx,microtype,url,xurl,placeins}
\usepackage{algorithm,algpseudocode}
\usepackage{tikz}
\usepackage[hidelinks]{hyperref}
\hypersetup{pdftitle={Learning the Maximum Tolerated Dose for Continuous Toxicity via Monotone Bayesian Trees},pdfauthor={Se Yoon Lee}}
\theoremstyle{definition}
\newtheorem{definition}{Definition}
\theoremstyle{plain}
\newtheorem{proposition}{Proposition}
\newtheorem{theorem}[proposition]{Theorem}
\newcommand{\E}{\mathbb E}

\title{Learning the Maximum Tolerated Dose\\for Continuous Toxicity via\\Monotone Bayesian Trees}
\author{Se Yoon Lee}
\begin{document}
\maketitle
\begin{abstract}
Phase I cancer trials seek the maximum tolerated dose (MTD) while protecting patients from excessive toxicity. Dose assignments must therefore balance patient safety with learning the dose--toxicity relationship as data accrue. We model continuously measured toxicity outcomes using two forms of Bayesian additive regression trees (BART): isotonic BART projects posterior response curves onto nondecreasing functions, whereas monotone BART constrains the model. Joint curve and variance draws induce an MTD posterior that guides dose selection through escalation with overdose control (EWOC). We compare these methods with three parametric procedures across seven dose--toxicity curves in simulation. We assess dose-limiting toxicity (DLT) counts, above-MTD assignments, signed last-dose error, and relative absolute error. The tree procedures attained the lowest mean RAE on four nonlinear curves and jointly minimized mean DLT counts and above-MTD assignments on four curves. A Bayesian reinforcement learning perspective formulates these sequential decisions as a finite-horizon planning problem. Dose restrictions yield a lower bound on last-dose error; under exact posterior-predictive evaluation, an EWOC-based rollout policy has no greater expected weighted loss than its baseline. \textsf{\textbf{Dose Trial Lab}}, a desktop simulator for all five procedures, accompanies the supplementary material.
\end{abstract}
\section{Introduction}
Phase I cancer trials seek to identify the maximum tolerated dose (MTD), the highest dose associated with an acceptable toxicity risk \citep{tighiouart2010}. Traditional phase I cancer dose-escalation trials generally enroll roughly 30--50 patients \citep{dahlberg2014}. Responses from earlier patients or cohorts guide later dose assignments. Because participants often have advanced cancer and limited treatment options, dose escalation must pursue precise MTD estimation while carefully limiting exposure to excessive toxicity \citep{tighiouart2010,bartroff2011}.

In this paper, we model continuously measured toxicity outcomes, such as a composite toxicity score derived from adverse events graded using the Common Terminology Criteria for Adverse Events (CTCAE)\footnote{\url{https://dctd.cancer.gov/research/ctep-trials/for-sites/adverse-events}} or a surrogate biomarker \citep{lee2022,lee2023}. In our framework, a response meeting or exceeding a prespecified threshold defines a dose-limiting toxicity (DLT) event. The probability of this event determines the MTD.

We use escalation with overdose control (EWOC) as our dose-allocation rule. EWOC proposes a conservative posterior quantile of the MTD using a prespecified feasibility level $\alpha$ \citep{babb1998,zacks1998}. For cytotoxic agents, toxicity risk is presumed nondecreasing with dose, although the shape of the dose--toxicity relationship is unknown \citep{tighiouart2010,pantoja2022,lee2023}. Given the small sample sizes of phase I cancer trials, the dose-finding literature has typically relied on parametric models. Our comparators are the one-parameter linear dose finder (1PLD), which estimates a slope; the two-parameter linear dose finder (2PLD), which also estimates residual variance; and the three-parameter nonlinear dose finder (3PND), which estimates scale, variance, and a curvature exponent \citep{eichhorn1973,lee2022,lee2023}. We also study two constructions of Bayesian additive regression trees (BART): isotonic BART (\textbf{iBART}) projects posterior draws onto nondecreasing functions, whereas monotone BART (\textbf{mBART}) constrains the model itself \citep{chipman2010,chipman2022}. Figure~\ref{fig:dose-toxicity-curves} illustrates the mean dose--toxicity curves for all five procedures.

We also study a Bayesian reinforcement learning (RL) formulation that connects these response models to sequential dose decisions under the trial protocol. The state comprises the joint posterior, last administered dose, enrollment, and cumulative DLT count. Under the upward cap and Only Escalation restriction, an excessive dose increase cannot be reversed, while holding the dose can leave the target unreachable within the remaining patient budget. We derive a lower bound on the error of the terminal last-dose estimate and develop a finite-candidate Bellman formulation with an EWOC-based rollout policy. Under exact posterior-predictive evaluation and the stated conditions, rollout has no greater expected weighted loss than the EWOC baseline. Appendix~\ref{app:rl-formulation} gives the proofs and planning algorithm. We evaluate large-budget rollout with 1PLD, whose posterior update is analytic (Appendix~\ref{sec:rl-seven-curve}). This formulation accounts for dose controls, stopping after complete cohorts, and the use of the last administered dose as the terminal estimate.

In simulation experiments, we compare all five procedures under the same EWOC allocation, dose restrictions, and cumulative-DLT stopping, with 45 planned patients and seven dose--toxicity curves. The last dose estimates the MTD; the four criteria are DLT counts, above-MTD assignments, signed last-dose error, and relative absolute error. The accompanying desktop software supports all five dose-finding models, trial playback, posterior plots, and exports. Its bundled runtime permits offline use; Appendix~\ref{app:software-manual} provides an illustrated guide.

\paragraph{Novelty.} Table~\ref{tab:novelty-positioning} compares the dose-finding methods reviewed in Appendix~\ref{app:related-methods}. {Research on dose finding has appeared at ICML, AISTATS, NeurIPS, and ICLR (see Appendix~\ref{app:ml-dose-finding}).} To our knowledge, ours is the first phase I oncology framework combining continuously measured toxicity, monotone Bayesian trees for estimating the MTD without prespecifying a parametric curve, EWOC, and EWOC rollout within a Bayesian reinforcement learning formulation, with MTD-posterior consistency, a protocol-specific reachability bound, and a guarantee that exact rollout has no greater expected loss than EWOC under the stated conditions. Alongside the methodological contribution, we developed a desktop simulator for users.

\section{Bayesian dose finding}
\label{sec:conference-method}
\subsection{Problem setup}
Let $N$ denote the planned enrollment,
$n\leq N$ the number of patients observed
so far, and $\mathcal F_n$ the information contained in the observed
dose--response pairs $\{(X_i,Y_i)\}_{i=1}^n$. The MTD posterior guides the next dose after each patient or completed cohort. Clinicians prespecify a fixed dose interval $\mathcal X=(x_{\min},x_{\max})\subset(0,\infty)$ in original dose units, with width $L_x=x_{\max}-x_{\min}$. Preclinical studies typically inform these limits. {$Y(x)$ denotes the continuous toxicity response of a patient assigned dose $x$.} Larger responses indicate more severe toxicity; a DLT event occurs when $Y(x)\geq\eta$ for prespecified $\eta>0$, with cumulative count $D_n=\sum_{i=1}^n\mathbf1\{Y_i\geq\eta\}$. Sponsors and investigators seek the MTD while limiting excessive DLT risk \citep{bartroff2011}.

\begin{definition}[Maximum tolerated dose; \citealp{lee2023}]
\label{def:clinical-mtd}
For prespecified $\eta>0$ and $\gamma\in(1/2,1)$, a dose $x\in\mathcal X$ is acceptable if
\[
 \Pr\{Y<\eta\mid X=x\}\geq\gamma.
\]
Equivalently, its DLT probability is at most $\theta=1-\gamma$.
The MTD $\xi$ is the supremum of acceptable doses, with $\xi=x_{\min}$
if none are acceptable and $\xi=x_{\max}$ if all are acceptable.
\end{definition}

\begin{figure}[!t]
 \centering
  \begin{minipage}[t]{200bp}
  \vspace{0pt}
   \includegraphics[width=200bp]{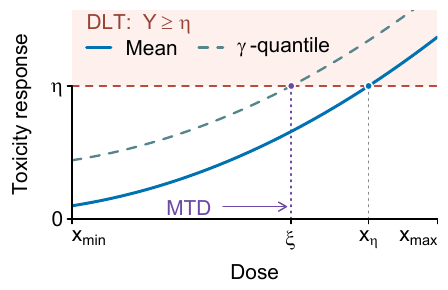}
 \end{minipage}\hfill
  \begin{minipage}[t]{184bp}
  \vspace{0pt}
  \setlength{\abovecaptionskip}{0pt}
  \setlength{\belowcaptionskip}{0pt}
  \caption[A continuously measured toxicity outcome and the MTD.]{A continuously measured toxicity outcome and the MTD.
  For an increasing Gaussian mean, the dashed response quantile
  $f+\sigma\Phi^{-1}(\gamma)$ reaches $\eta$ at $\xi$; the mean reaches it
  at $x_\eta$. DLT probabilities are $1-\gamma$ and $1/2$, respectively.
  The dashed curve describes outcome variability, not posterior uncertainty.
  Increasing $\sigma$ lowers the MTD when both crossings are interior.
  Joint curve--variance draws propagate uncertainty into dose selection.}
  \label{fig:clinical-concept}
 \end{minipage}
\end{figure}

We model the continuously measured outcome with a joint prior on the mean dose--toxicity curve $f$ and residual standard deviation (SD) $\sigma$:
\begin{equation}
 \begin{gathered}
 Y_i=f(X_i)+\epsilon_i,\quad
 \epsilon_i\mid\mathcal F_{i-1},X_i\sim\mathcal N(0,\sigma^2),\\
 d\Pi_n(f,\sigma)\propto
 \left[\prod_{i=1}^n\frac{\phi\{(Y_i-f(X_i))/\sigma\}}{\sigma}\right]
 d\Pi(f,\sigma).
 \end{gathered}
 \label{eq:conference-posterior}
\end{equation}
Here $X_i\in\mathcal X$ is the dose administered to patient $i$. Subsection~\ref{sec:bart-likelihood} derives the posterior under adaptive allocation. {Here $\phi$ and $\Phi$ denote the standard normal density and CDF, and $\Pi$ is the joint prior. Continuous-response BART models the conditional mean nonparametrically and assumes Gaussian errors with common variance \citep[Eqs.~(1)--(2)]{chipman2010}. Following \citet{lee2023}, we use the Gaussian tail to translate mean toxicity and residual variability into DLT risk. This working model is appropriate when residuals are approximately Gaussian with similar spread across doses on the analysis scale.}

\subsection{Monotone Bayesian trees: \textnormal{iBART} and \textnormal{mBART}}

\textbf{iBART.} Fit the unconstrained BART mean $g=b_0+\sum_{j=1}^{m}h_j$, replacing $f$ by $g$ in \eqref{eq:conference-posterior}. {Here $b_0$ is the fixed offset, $m$ the tree count, and $h_j$ the piecewise-constant regression function represented by tree $j$. The prespecified candidate cuts partition the dose interval into ordered cells $B_1,\ldots,B_J$; $g_k$ is the value of a draw $g$ on $B_k$, and $|B_k|$ is its physical length.} Project each draw onto the nondecreasing cone:
\begin{equation*}
 (\mathcal P g)_j=\left[\operatorname*{arg\,min}_{z_1\leq\cdots\leq z_J}
 \sum_{k=1}^{J}|B_k|(g_k-z_k)^2\right]_j .
\end{equation*}
The pool-adjacent-violators algorithm (PAVA) solves this weighted least-squares problem \citep{deleeuw2009}. Cell-length weights give the Lebesgue $L_2$ projection on $[x_{\min},x_{\max}]$. Pair each projected draw with its sampled $\sigma$ for MTD inference.

\textbf{mBART.} Constrain each tree's leaf values to be nondecreasing from left to right \citep{chipman2022}. The sum $f$ is then nondecreasing, so no projection is needed. The constrained-product prior imposes this monotonicity in the fitted model. Appendix~\ref{sec:bart-priors} distinguishes that continuous prior from its numerical leaf-grid approximation.

Both methods target a nondecreasing true mean $f_0$, but the dose-finding estimand is the MTD rather than the curve itself. Thus $f$ (together with $\sigma$) is an intermediate, or nuisance, component for MTD inference, unlike in conventional curve fitting. Monotonicity concerns the mean response and toxicity risk, not individual outcomes. For MTD estimation, iBART uses $f=\mathcal P g$, while mBART uses $f$ directly. Define the toxicity-risk functional $p_{f,\sigma}$ and MTD functional $\xi(f,\sigma)$ by
\begin{equation*}
 p_{f,\sigma}(x)=\Phi\{(f(x)-\eta)/\sigma\},\qquad
 \xi(f,\sigma)=\sup\bigl(\{x_{\min}\}\cup
 \{x\in\mathcal X:p_{f,\sigma}(x)\leq\theta\}\bigr).
\end{equation*}
For mBART, $p_{f,\sigma}$ is conditional toxicity risk; for iBART it is a projected tail functional, with the unprojected curve retained in the likelihood. Definition~\ref{def:clinical-mtd} permits jumps and plateaus without requiring $p_{f,\sigma}(\xi)=\theta$; assigned doses remain in $\mathcal X$. Joint draws of $f$ and $\sigma$ induce the MTD posterior.

\paragraph{Priors and computation.}
We use the BART and mBART tree-depth defaults, 200 trees, and shrinkage $k=2$ \citep{chipman2010,chipman2022}. A common prior on residual variance completes each model. Before enrollment, we fix all priors and 100 equally spaced candidate tree-split points on $(x_{\min},x_{\max})$, the dose thresholds available for splitting tree nodes. Doses and split points retain their physical units. Appendix~\ref{sec:bart-priors} gives the response calibration, variance prior, cut locations, and numerical approximation.
Appendix~\ref{app:method-dags} compares the graphical models of all five procedures: 1PLD, 2PLD, 3PND, iBART, and mBART.

\subsection{EWOC allocation and stopping}
\label{sec:conference-ewoc-allocation}

After observing $\mathcal F_n$, let
$\Pi_n^\xi(x)=\Pr(\xi\leq x\mid\mathcal F_n)$ denote the posterior
cumulative distribution function (CDF) of the MTD.
EWOC \citep{babb1998,zacks1998,tighiouart2010,lee2022,lee2023} proposes the next
dose at its lower $\alpha$-quantile:
\[
D_\alpha(\mathcal F_n)=(\Pi_n^\xi)^{-1}(\alpha).
\]
For a continuous MTD posterior,
$\Pr\{\xi\leq D_\alpha(\mathcal F_n)\mid\mathcal F_n\}=\alpha$.
A smaller $\alpha$ therefore selects a more conservative dose.
The clinical target $\theta=1-\gamma$ specifies acceptable toxicity
risk, whereas $\alpha$ controls allocation under posterior uncertainty.

Tree-based MTD posteriors can place mass at cut points or endpoints.
We inward-adjust the EWOC quantile to obtain a proposal
$\widetilde x_{n+1}$, using the numerical interior floor
$x_{\min}+\delta_{\mathrm{num}}$, and assess the nominal condition $\Pr\{p_{f,\sigma}(\widetilde x_{n+1})>1-\gamma
\mid\mathcal F_n\}\leq\alpha.$ For iBART, this probability concerns the projected tail functional.

The assigned dose also respects the protocol's dose restrictions.
With a prespecified upward cap $\Delta_{\max}>0$ and
Only Escalation enabled, the next cohort receives
\[
X_{n+1}
=
\max\left\{
X_n,\,
\min\left(\widetilde x_{n+1},X_n+\Delta_{\max}\right)
\right\}.
\]
The cap limits the size of an increase, while Only Escalation prevents
a decrease. Thus, the implemented allocation combines an EWOC proposal
with dose restrictions that prevent abrupt escalation: even when the
posterior $\alpha$-quantile suggests a substantially higher dose,
the increase for the next patient or cohort cannot exceed
$\Delta_{\max}$.

Let $K$ denote the prespecified DLT-count limit and $D_n$ the cumulative
number of observed DLTs. Enrollment ends at the first completed cohort
for which $D_n>K$, or when the planned enrollment $N$ is reached.
Clinical thresholds, dose range, $\alpha$, initial dose, patient
budget, cohort size, dose restrictions, stopping limit, and priors
are specified before enrollment.
Appendix~\ref{sec:bart-target} gives the full rule.

This comparison uses $N=45$, cohorts of three,
$\Delta_{\max}=3.5$, and $K=\lfloor0.1N\rfloor=4$, with Only
Escalation enabled. {The generating curve and MTD remain
unknown to the allocation rule; 1PLD uses the known generating residual SD,
while the other methods estimate it.} \textsf{\textbf{Dose Trial Lab}}
allows users to disable either dose restriction and explore
alternative settings.

\begin{samepage}
\paragraph{A common adaptive design for five models.}
Each model supplies an MTD posterior; the common rule $\mathcal D$ applies the resulting EWOC proposal, dose restrictions, and count-based stopping. Algorithm~\ref{alg:main-dose-finding} gives the common sequence. Appendix~\ref{sec:bart-target} specifies MTD-posterior inversion and the common restrictions; Appendix~\ref{app:bart-numerics} gives the model-specific posterior updates.
\par\end{samepage}

\begin{algorithm}[!htb]
\caption{EWOC trial loop for any of the five response models}
\label{alg:main-dose-finding}
\begin{algorithmic}[1]
\Require Prespecified initial dose $x_1$, budget $N$, cohort size $c$, threshold $K_N=\lfloor0.1N\rfloor$, and the chosen model's rule $\mathcal D$, including priors, clinical inputs, $\alpha$, inward margin, floor, upward cap, and Only Escalation setting.
\Ensure Reported MTD estimate $\widehat\xi_{n_{\mathrm{end}}}=X_{n_{\mathrm{end}}}$ and information $\mathcal F_{n_{\mathrm{end}}}$, where $n_{\mathrm{end}}\leq N$.
\State Set $n=0$, $D_0=0$, and $x=x_1$.
\While{$n<N$}
  \State Set $c'=\min(c,N-n)$.
  \State Assign $X_{n+1}=\cdots=X_{n+c'}=x$.
  \State Observe all $c'$ responses, append them to the history, and set $n\gets n+c'$.
  \State Update the posterior using all of $\mathcal F_n$.
  \State Set $D_n=\sum_{i=1}^n\mathbf1\{Y_i\geq\eta\}$.
  \If{$n=N$ or $D_n>K_N$} \State \textbf{break} \EndIf
  \State Set $x=\mathcal D(\mathcal F_n)$ using the inward EWOC proposal, upward cap, and Only Escalation hold.
\EndWhile
\State Set $n_{\mathrm{end}}=n$ and retain the final posterior.
\end{algorithmic}
\smallskip
Each cohort shares one dose. All of its responses enter the posterior before stopping is checked; the next dose is computed only if the trial continues.
\end{algorithm}

\paragraph{MTD posterior consistency.}
Appendix~\ref{sec:doob-consistency} proves consistency of the MTD
posterior under a fixed-cell Gaussian model, prior support at the
truth, persistent cell sampling, and a strict risk-threshold margin.
The proof uses a Schwartz-type likelihood numerator--denominator
argument \citep{schwartz1965}. Under these conditions, posterior mass
concentrates at the true MTD as observations accumulate. This asymptotic
result does not establish finite-sample operating characteristics.
{Persistent sampling of every cell is not provided by the implemented Only Escalation policy: once a lower cell is left, it cannot be revisited. The theorem is therefore conditional on an informative allocation sequence, rather than a consistency guarantee for Algorithm~\ref{alg:main-dose-finding}.}

\section{Bayesian reinforcement learning under the trial protocol}
\label{sec:rl-framework}
\begin{figure}[!t]
\centering
\includegraphics[width=\textwidth]{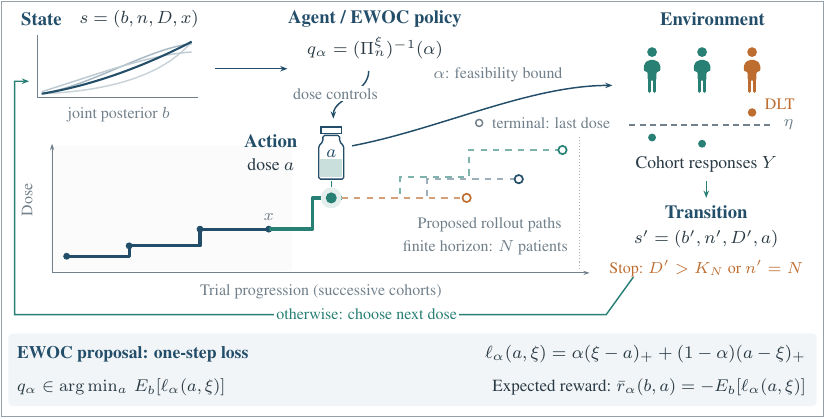}
\caption{\textbf{Dose finding as a Bayesian RL decision process.}
The state combines the joint posterior and trial status. The EWOC policy
selects a dose; complete cohort responses update the state before the
stopping check. Idealized EWOC minimizes posterior expected asymmetric
dose loss \citep{bartroff2011}, whose realized value depends on the unknown
MTD. The implemented policy additionally applies the inward adjustment
and dose controls, including the protocol's floor and hold fallback.
Dashed paths illustrate rollout continuations; the five-model comparison
uses the EWOC rule.}
\label{fig:rl-dose-feedback}
\end{figure}
A Bayesian RL perspective represents sequential dose selection
through a joint posterior state and a finite enrollment horizon
\citep{ross2011bayesian}. Bayesian dynamic programming and EWOC-based
rollout have clinical precedents \citep{bartroff2010dp}; here we incorporate
continuously measured toxicity outcomes, monotone tree inference, and the last-dose estimate.
The upward cap, Only Escalation, and completed-cohort stopping apply to
both actual and hypothetical trajectories. Figure~\ref{fig:rl-dose-feedback}
summarizes the resulting loop. We formulate EWOC and Bayesian planning
within this belief-state decision process.
Section~\ref{sec:conference-experiments} compares five EWOC procedures;
Appendix~\ref{app:rl-formulation} provides the planning proofs, and Appendix~\ref{sec:rl-seven-curve} reports the 1PLD rollout comparison.

\paragraph{State and admissible doses.}
After a completed cohort, use the state
$s=(b,n,D,x)$: the joint parameter posterior, patient count, cumulative
DLT count, and last assigned dose. The joint posterior is needed because
an MTD distribution alone does not determine the next response
distribution. For example, different positive slopes can have the same
Gaussian tail crossing but different responses below that crossing.
The same MTD belief can consequently produce different information
from the next cohort. The last dose $x$ determines the no-decrease
constraint; $n$ and $D$ determine the remaining budget and proximity
to a compulsory stop. A state is terminal exactly when $D>K_N$ or $n=N$.
With Only Escalation and the upward cap active after the prespecified
initial cohort, the protocol interval is
\[
 I(s)=[x,\min\{x+\Delta_{\max},x_{\max}-\delta_{\mathrm{num}}\}].
\]
Set $r_b(a)=b\{p_{f,\sigma}(a)>1-\gamma\}$ and
$F(s)=\{a\in I(s):r_b(a)\leq\alpha\}$. The planning action set is
\begin{equation}
 \mathcal A(s)=
 \begin{cases}
 F(s),&F(s)\ne\varnothing,\\
 \{x\},&F(s)=\varnothing .
 \end{cases}
 \label{eq:rl-actions}
\end{equation}
The second case enforces a hold and records a feasibility exception;
it does not stop enrollment. Monotonicity implies that if the current
dose is infeasible, every allowed higher dose is also infeasible.
With Only Escalation disabled, replace the lower endpoint by
$x_{\min}+\delta_{\mathrm{num}}$ and use that floor as the fallback.
With the cap disabled, replace the upper endpoint by
$x_{\max}-\delta_{\mathrm{num}}$.
Thus the same formulation accommodates all four option combinations.
For the exact-posterior EWOC policy, the inward proposal is below its
lower $\alpha$-quantile unless the lower floor binds. Below that quantile, $r_b(a)\leq b\{\xi\leq a\}<\alpha$;
an upward cap can only lower the proposal. Any infeasible baseline
assignment therefore occurs at the permitted lower endpoint and agrees
with the fallback in \eqref{eq:rl-actions}.

\paragraph{The one-step loss underlying the EWOC proposal.}
The agent applies a dose-selection policy, its action is the next
cohort's dose, and the patient response model defines the environment. EWOC is
a specified policy in this decision process. Its idealized, unadjusted
posterior quantile $q_\alpha=(\Pi_n^\xi)^{-1}(\alpha)$ minimizes the
posterior expected asymmetric dose loss \citep{bartroff2011}:
\[
 \ell_\alpha(a,\xi)=\alpha(\xi-a)_+ +(1-\alpha)(a-\xi)_+,
 \qquad q_\alpha\in\arg\min_a\E_b[\ell_\alpha(a,\xi)].
\]
The two terms penalize underdosing and overdosing, respectively; a small
$\alpha$ gives greater weight to overdosing. Thus the quantile maximizes
the one-step expected reward $-\E_b[\ell_\alpha(a,\xi)]$. The implemented
policy also applies the inward adjustment and dose controls. The
planning extension below evaluates the full trial under a separate
cumulative loss.

\paragraph{Learning transition.}
Write $\vartheta$ for the model parameter and $\mu_\vartheta$ for its
likelihood mean: the raw $g$ for iBART, and $f$ for mBART.
For $m=\min(c,N-n)$ patients at dose $a$, where $c$ is cohort size,
\begin{equation*}
 P_b(d\mathbf y\mid a)=\int\prod_{j=1}^{m}
 \mathcal N(dy_j;\mu_\vartheta(a),\sigma^2)\,b(d\vartheta),\quad
 s'=\bigl(\mathcal U(b,a,\mathbf y),\,n+m,\,D+d(\mathbf y),\,a\bigr),
\end{equation*}
where $d(\mathbf y)=\sum_{j=1}^{m}\mathbf1\{y_j\geq\eta\}$ and
$\mathcal U$ is the joint Bayesian update. Integrating a product of
conditional densities accounts for the shared unknown parameter;
multiplying marginal predictive densities would not. The final
cohort is fully observed and its posterior update is completed before
checking the stopping rule. Under the numerical-study protocol, $N=45$, $c=3$,
and $K_N=4$: five or more cumulative DLTs trigger the count rule after
the cohort, and enrollment also ends at 45 patients. Neither the upward
cap nor Only Escalation is a stopping rule. These variables and the
joint posterior suffice for the conditional law of future responses,
costs, and states (Proposition~\ref{prop:rl-markov}).

\paragraph{Planning with the last dose as the estimate.}
The planning extension uses a different loss from $\ell_\alpha$: it
combines DLT counts, above-MTD assignments, and terminal relative
absolute error with prespecified nonnegative weights $w_D,w_O,w_R$.
EWOC supplies a baseline policy; it is not assumed to minimize this
loss at one step or over the trial. Signed error retains its role as
a reporting criterion. Let $o_b(a)=b\{\xi<a\}$. This probability
defines above-MTD exposure and may differ from $r_b(a)$ at a tree
jump. Write
\[
 R(b,x)=w_R\int\frac{|x-\xi(\vartheta)|}{\xi(\vartheta)}\,b(d\vartheta).
\]
The continuous-action value $V(s)$ is the infimum of posterior expected
continuation loss over admissible policies: future DLT and above-MTD
counts plus terminal last-dose error.
Past count contributions are fixed at the current state. The
following finite-horizon Bellman recursion holds when measurable
$\varepsilon$-optimal selectors exist for the admissible action sets:
\begin{equation}
 \begin{gathered}
 V(s)=R(b,x),\quad D>K_N\ \text{or}\ n=N,\\
 V(s)=\inf_{a\in\mathcal A(s)}\biggl\{w_Om\,o_b(a)
 +\int\bigl[w_Dd(\mathbf y)+V(s')\bigr]
 P_b(d\mathbf y\mid a)\biggr\},\quad\text{otherwise}.
 \end{gathered}
 \label{eq:rl-bellman}
\end{equation}
Its terminal estimate is the dose actually administered to the last
cohort, including after a DLT stop. There is no additional terminal
dose choice. The terminal expectation quantifies loss under current
uncertainty; it does not replace the last dose by a posterior median.
The weights express the relative importance assigned to the existing
clinical criteria; they are not additional reported performance metrics.
In RL terms, maximizing expected undiscounted return is equivalent to
minimizing this loss: each cohort contributes reward
$-w_Dd(\mathbf y)-w_Om\mathbf1\{\xi<a\}$, followed by terminal reward
$-w_R|x-\xi|/\xi$. MTD-dependent rewards are latent and are evaluated
through posterior expectations.
The DLT term in \eqref{eq:rl-bellman} uses the predictive distribution
from the likelihood. In particular, iBART uses its raw mean $g$ to
generate responses, while its projected mean defines the MTD and
posterior feasibility. Substituting the projected tail for the raw
predictive DLT probability would change this planning problem.

The horizon is finite because every action enrolls at least one patient.
With $x_{\min}>0$, bounded doses, and fixed nonnegative weights, all
costs are bounded. For a measurable finite candidate set
$C(s)\subseteq\mathcal A(s)$, the infimum becomes a minimum; backward
induction with a fixed tie rule yields a Markov optimal policy on those
candidates (Proposition~\ref{prop:rl-bellman}). No assertion that a
minimizer exists on the full continuous action set is needed for the
finite-candidate planner.

\paragraph{Reachability and terminal error.}
The allocation restrictions determine the range of terminal doses
reachable from a given state. Let $H(s)=\lceil(N-n)/c\rceil$ at a nonterminal state and
$H(s)=0$ at a terminal state, and put
\[
 u_H=\min\{x+H(s)\Delta_{\max},x_{\max}-\delta_{\mathrm{num}}\}.
\]

\begin{proposition}[Reachability bound]
\label{prop:rl-reachability}
Under Only Escalation and the upward cap, every admissible continuation
from $s$ has terminal dose $X_{n_{\mathrm{end}}}\in[x,u_H]$. If
$\xi(\vartheta)\in[x_{\min},x_{\max}]$ and $x_{\min}>0$, then for
every such policy $\pi$,
\begin{equation*}
 \E_b^\pi\!\left[\frac{|X_{n_{\mathrm{end}}}-\xi(\vartheta)|}{\xi(\vartheta)}\right]
 \geq\int\frac{(x-\xi(\vartheta))_+
               +(\xi(\vartheta)-u_H)_+}{\xi(\vartheta)}\,b(d\vartheta),
\end{equation*}
where $(z)_+=\max(z,0)$ and $n_{\mathrm{end}}$ is terminal enrollment.
\end{proposition}

Each cohort can increase the dose by at most $\Delta_{\max}$ and cannot
decrease it, so the bound follows by measuring the distance from the
latent MTD to the reachable interval. The first term is the error from
overshooting a target to which the trial cannot return. The second is
the error from a target too high to reach in the remaining cohorts.
DLT stopping can shorten the path and cannot invalidate the bound.
From dose 6, the cap of 3.5 permits a last dose of 55 in 45 patients.
{With 30 patients, nine increases permit at most 37.5. At a true MTD of 50, last-dose RAE is therefore at least 25\%. In the 30-patient reanalysis both trees reached 37.5 in every trial (Appendix~\ref{sec:bart-experiment}).} For example, at
$n=33$ and $x=36$, four cohorts remain and
the cap permits a last dose of at most 50. A hold leaves three cohorts,
lowering this upper bound to 46.5; if the true MTD is 50, every continuation then has
relative absolute error at least 7\%. This bound isolates lost
reachability, while the Bellman objective jointly accounts for toxicity
and learning. With four DLTs already observed, any additional
DLT in the next completed cohort ends the trial, potentially making
that cohort's dose the terminal estimate.

\paragraph{Policy improvement over the EWOC baseline.}
{Write $V^\pi(s)$ for the posterior expected continuation loss under policy $\pi$ from state $s$.}
Let $\pi_0$ denote the exact-posterior EWOC policy, including inward
adjustment, dose controls, fallback, and mandatory stopping. Choose a finite candidate
set $C(s)\subseteq\mathcal A(s)$ containing $\pi_0(s)$. Define
\begin{equation*}
 \begin{aligned}
 Q^{\pi_0}(s,a)&=w_Om\,o_b(a)
 +\int\bigl[w_Dd(\mathbf y)+V^{\pi_0}(s')\bigr]
 P_b(d\mathbf y\mid a),\\
 \pi_+(s)&\in\arg\min_{a\in C(s)}Q^{\pi_0}(s,a).
 \end{aligned}
\end{equation*}
Thus a candidate is evaluated using its immediate cohort and the
complete future EWOC continuation, including a possible early DLT
stop. The improved policy recomputes this comparison at every reached
state and therefore governs the full realized trial.

\begin{proposition}[Exact rollout improvement]
\label{prop:rl-rollout}
For fixed nonnegative weights, exact posterior predictive expectations,
and measurable finite candidate sets containing the baseline action,
the policy $\pi_+$ with a fixed tie rule satisfies
$V^{\pi_+}(s)\leq V^{\pi_0}(s)$ at every state. Both policies use the
same response model, dose restrictions, and terminal rules.
\end{proposition}

The baseline candidate gives
$Q^{\pi_0}(s,\pi_+(s))\leq V^{\pi_0}(s)$; induction on the remaining
cohorts extends that comparison to repeated rollout decisions.
For the same candidate sets, $V_C(s)\leq V^{\pi_+}(s)\leq V^{\pi_0}(s)$:
rollout has no greater expected loss than the baseline, while $V_C$
records the finite-candidate optimum. This ordering concerns exact
Bayesian evaluation of the specified weighted loss.

{Finite simulation budgets and numerical posterior updates do not automatically inherit this guarantee. Proposition~\ref{prop:rl-approx-rollout} bounds additional loss by $2H(s)\varepsilon$ when candidate values have uniform error at most $\varepsilon$ on the same belief-state model and legal candidate sets, where $H(s)$ is the maximum remaining cohort count. Such a uniform bound is not established for the implementation.}

Rollout evaluates candidate doses through hypothetical trials, updating the
posterior after every simulated cohort until stopping.
\citet{bartroff2010dp} implemented EWOC rollout in phase I trials and noted
its computational cost. Algorithm~\ref{alg:rl-app-rollout} and
Figure~\ref{fig:rl-repeated-rollout} in Appendix~\ref{sec:rl-mc-rollout}
give the full procedure and schematic. Appendix~\ref{sec:rl-seven-curve} and
Table~\ref{tab:rl-seven-curve} give the 1PLD rollout results on the linear curve using analytic posterior updates. Across
upward caps 3.5, 7, and 10.5, rollout lowered observed mean weighted loss
under all three weight profiles at both $B=1,000$ and $2,000$. Relative to balanced rollout, safety emphasis reduced NPD and NPO at caps 3.5 and 7.
Precision emphasis reduced RAE at all three caps.
{For 2PLD, 3PND, and both trees, each simulated cohort that requires a posterior update incurs a numerical fit. Figure~\ref{fig:rollout-work-bound} in Appendix~\ref{sec:rl-seven-curve} shows how this cost grows with the planning budget.}
 
\paragraph{The final-cohort decision.}
When only one cohort remains, every candidate dose becomes the terminal
estimate. Write $\rho_b(a)$ for the posterior predictive DLT probability,
computed from the likelihood mean $\mu_\vartheta$. The candidate's expected
loss then reduces to
\[
 Q(s,a)=m w_D\rho_b(a)+m w_Oo_b(a)+R(b,a).
\]
This follows by averaging the terminal posterior loss over the cohort's
possible responses: the law of iterated expectation gives
$\int R\{\mathcal U(b,a,\mathbf y),a\}P_b(d\mathbf y\mid a)=R(b,a)$.
The final observations can update the posterior but cannot change the
already assigned dose. Earlier in the trial, information can change
subsequent assignments, and its value depends on which doses remain
reachable. The Bellman recursion accounts for this difference without
adding a separate reward for reducing posterior uncertainty.

\section{Experiments}
\label{sec:conference-experiments}
\subsection{Design and performance criteria}
{The five-model comparison evaluates the outcome-model choices of Section~\ref{sec:conference-method} under a common EWOC protocol. Policy comparisons within a fixed model separately evaluate the planning question of Section~\ref{sec:rl-framework}.}
The study compares 1PLD, 2PLD, 3PND, iBART, and mBART across
seven dose--toxicity curves: linear, piecewise linear, quadratic, square
root, sigmoid, exponential, and logarithmic. Figure~\ref{fig:bart-truth-scenarios}
in Appendix~\ref{sec:bart-seven-curves} shows their shapes and shared MTD
calibration. The dose interval is $(5,80)$,
initial dose 6, toxicity threshold $\eta=3$, true residual SD $0.1$, and
true MTD $\xi_0=50$. Each trial plans $N=45$ patients in cohorts of three,
with $\gamma=0.9$ and $\alpha=0.05$. Both controls are enabled:
the upward cap $3.5$ and Only Escalation
(Section~\ref{sec:conference-ewoc-allocation}). They are required
components of the evaluated protocol and the software defaults. The posterior uses every
response in each completed cohort. The stopping rule is checked after
every cohort: enrollment ends when cumulative DLTs exceed $K_N=4$,
or when 45 patients have been observed.

The study uses 100 paired replications per curve; each replication evaluates
all five methods with shared presampled patient innovations, for 3,500 trials.
Posterior sampling uses separate streams. Priors are specified in
Appendices~\ref{sec:bart-priors} and~\ref{app:bayesian-comparators}; 1PLD
knows the residual SD, whereas the other models estimate it. {On fixed linear-curve histories, halving mBART's prior split parameter raised its mean EWOC proposal by 3.53 dose units but its mean assigned dose by only 0.41 (Appendix~\ref{app:sqkv-prior-simulation}).}

The reported MTD estimate is the last administered dose
$\widehat\xi_{n_{\mathrm{end}}}=X_{n_{\mathrm{end}}}$. The four performance criteria are
\begin{equation*}
 \mathrm{NPD}=\sum_{i=1}^{n_{\mathrm{end}}}\mathbf1\{Y_i\geq\eta\},\enspace
 \mathrm{NPO}=\sum_{i=1}^{n_{\mathrm{end}}}\mathbf1\{X_i>\xi_0\},\enspace
 \mathrm{BTM}=X_{n_{\mathrm{end}}}-\xi_0,\enspace
 \mathrm{RAE}=|\mathrm{BTM}|/\xi_0.
\end{equation*}
NPD and NPO assess safety performance. The precision criteria BTM and
RAE describe MTD estimation error: BTM retains its direction, whereas
RAE measures its relative magnitude without cancellation.
All trials remain in the summaries with their actual enrollment and
early stopping status.

\subsection{Seven-curve comparison}
{Table~\ref{tab:primary-learning} compares mean DLT counts and mean RAE across all seven truths and five methods. Table~\ref{tab:primary-learning-full} in Appendix~\ref{app:paired-reanalysis} gives means, medians, and Monte Carlo standard errors (MCSEs) for all four criteria.} No trial ended
before $N=45$, so the early-stopping percentage was zero.
Appendices~\ref{sec:bart-experiment} and~\ref{app:bart-numerics} give
the simulation design, enrollment and stopping summaries, and posterior
computation.

\begin{table}[!htb]
\centering
\normalsize
\setlength{\tabcolsep}{2.5pt}
\renewcommand{\arraystretch}{1.08}
\newcommand{\CompactPair}[2]{\makebox[2.1em][r]{#1}\,\textbar\,\makebox[2.7em][r]{#2}}
\caption{{Original experiment (100 paired trials per curve): mean DLT count \textbar\ mean RAE (\%). Bold, underlined values mark the lowest mean for each measure within a curve, including ties. Table~\ref{tab:primary-learning-full} reports all four criteria.}}
\label{tab:primary-learning}

\begin{tabular*}{\textwidth}{@{}l@{\extracolsep{\fill}}rrrrr@{}}
\toprule
True curve & \multicolumn{1}{c}{1PLD} & \multicolumn{1}{c}{2PLD} & \multicolumn{1}{c}{3PND} & \multicolumn{1}{c}{iBART} & \multicolumn{1}{c}{mBART} \\
\midrule
Linear & \CompactPair{0.34}{\textbf{\underline{1.20}}} & \CompactPair{0.26}{1.53} & \CompactPair{0.20}{1.83} & \CompactPair{\textbf{\underline{0.03}}}{3.98} & \CompactPair{0.06}{3.24} \\
Piecewise linear & \CompactPair{3.86}{10.00} & \CompactPair{3.86}{10.00} & \CompactPair{3.86}{9.97} & \CompactPair{\textbf{\underline{2.02}}}{3.02} & \CompactPair{\textbf{\underline{2.02}}}{\textbf{\underline{3.00}}} \\
Quadratic & \CompactPair{5.21}{10.00} & \CompactPair{4.36}{3.02} & \CompactPair{0.93}{\textbf{\underline{0.55}}} & \CompactPair{\textbf{\underline{0.39}}}{1.19} & \CompactPair{1.05}{1.20} \\
Square root & \CompactPair{\textbf{\underline{0.00}}}{20.05} & \CompactPair{\textbf{\underline{0.00}}}{35.56} & \CompactPair{0.11}{\textbf{\underline{6.90}}} & \CompactPair{0.03}{9.49} & \CompactPair{0.06}{8.73} \\
Sigmoid & \CompactPair{4.73}{10.00} & \CompactPair{3.52}{2.99} & \CompactPair{\textbf{\underline{0.12}}}{1.55} & \CompactPair{0.14}{1.74} & \CompactPair{0.53}{\textbf{\underline{1.43}}} \\
Exponential & \CompactPair{5.67}{10.00} & \CompactPair{5.28}{3.18} & \CompactPair{3.97}{2.23} & \CompactPair{\textbf{\underline{0.72}}}{\textbf{\underline{1.08}}} & \CompactPair{2.24}{1.49} \\
Logarithmic & \CompactPair{\textbf{\underline{0.00}}}{16.49} & \CompactPair{\textbf{\underline{0.00}}}{25.26} & \CompactPair{0.03}{9.44} & \CompactPair{0.02}{8.93} & \CompactPair{0.04}{\textbf{\underline{7.90}}} \\
\bottomrule
\end{tabular*}
\end{table}

For safety, Table~\ref{tab:primary-learning-full} shows that iBART jointly
minimized the observed mean NPD and NPO on four of seven curves, including
a tie with mBART for the piecewise linear curve. On the square-root and
logarithmic curves, 1PLD and 2PLD recorded zero NPD and NPO but
substantially underestimated the MTD; iBART recorded zero mean NPO and
mean NPD of only 0.03 and 0.02, respectively, with lower RAE than those
zero-count procedures. The sigmoid curve instead favored 3PND on safety.

For precision, the lowest observed mean RAE was attained by 1PLD on the
linear curve, 3PND on the quadratic and square-root curves, iBART on the
exponential curve, and mBART on the piecewise linear, sigmoid, and
logarithmic curves. mBART on the piecewise linear curve and iBART on the
exponential curve attained both safety and precision minima. No method attained both the safety and precision minima across all seven curves.
\subsection{\texorpdfstring{{MTD estimation under higher outcome variability}}{MTD estimation under higher outcome variability}}
\label{sec:focused-results}
{We doubled the residual SD from 0.1 to 0.2 and recalibrated each curve to retain the MTD of 50. All five methods used the same EWOC rule and 45-patient budget (50 paired replications per curve, 1,750 trials).} {On the piecewise-linear curve, mean RAE was $9.30\%$ with 3PND, $2.96\%$ with iBART, and $3.05\%$ with mBART. Both trees also reduced DLT counts there.} On the exponential curve, both trees had lower mean RAE, fewer DLTs, and fewer above-MTD assignments than 3PND (Figure~\ref{fig:higher-noise-tree-outcomes}). {On the square-root and logarithmic curves, greater precision came with more DLTs. On the linear curve, 1PLD had the lowest mean RAE ($2.28\%$). Appendix~\ref{app:focused-noise} reports all four outcomes.}

\paragraph{Gains of Bayesian trees over parametric dose-finding models.} {In the original experiment, both trees reduced mean RAE and DLT counts relative to the power-law 3PND model of \citet{lee2023} on the piecewise-linear and exponential curves. The paired mBART-minus-3PND RAE differences were $-6.97$ percentage points (95\% Monte Carlo interval $[-7.00,-6.94]$) and $-0.74$ ($[-0.94,-0.53]$), respectively. At the higher residual SD of 0.2, both trees had lower mean RAE than 1PLD and 2PLD on all six nonlinear curves (Figure~\ref{fig:focused-comparisons}(a)--(b)). On the square-root curve, for example, mean RAE was $21.59\%$ for 1PLD and $34.70\%$ for 2PLD, compared with $9.99\%$ for iBART and $9.49\%$ for mBART. Against 3PND, the paired 95\% intervals favored both trees on the piecewise-linear, square-root, exponential, and logarithmic curves (Figure~\ref{fig:focused-comparisons}(c)).}

\begin{figure}[!ht]
\centering
\includegraphics[width=\textwidth]{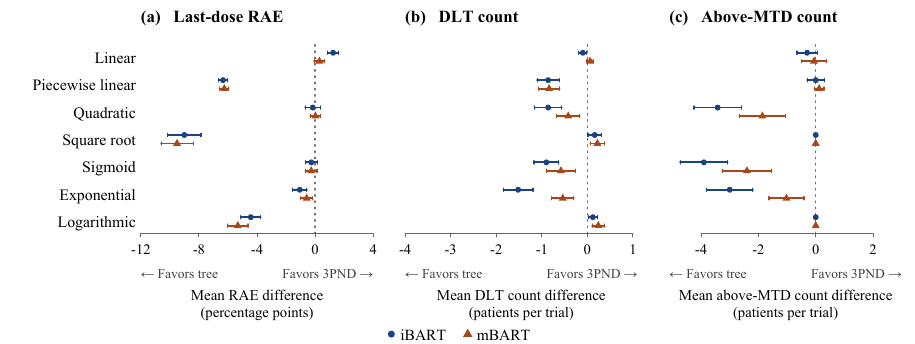}
\caption{{Paired tree-minus-3PND differences in mean last-dose RAE (percentage points), DLT count, and above-MTD count (patients per trial) at residual SD 0.2. Bars give pointwise 95\% Monte Carlo intervals for 50 paired trials per curve.}}
\label{fig:higher-noise-tree-outcomes}
\end{figure}

\begin{figure}[!ht]
\centering
\includegraphics[width=\textwidth]{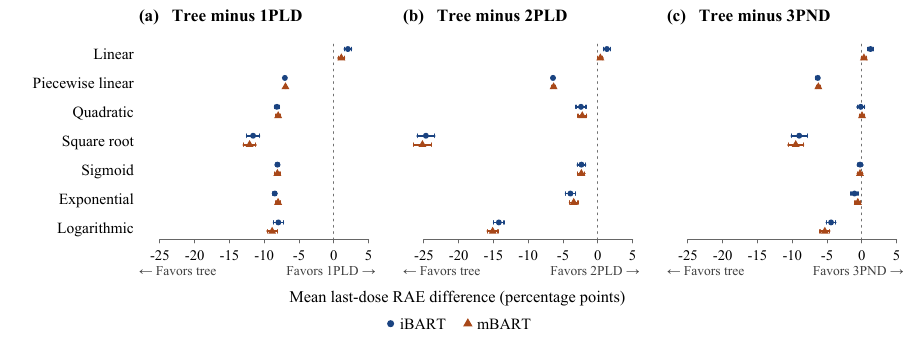}
\caption{{Paired differences in mean last-dose relative absolute error (RAE, percentage points) at residual SD 0.2. Each panel compares iBART and mBART with one parametric model under EWOC. Intervals are pointwise 95\% Monte Carlo intervals based on 50 paired trials per curve.}}
\label{fig:focused-comparisons}
\end{figure}
\FloatBarrier
\section{Discussion}
Flexible nonparametric and machine learning methods often need more data than correctly specified parametric models, a particular obstacle in phase I cancer trials with limited enrollment. Our design uses the available structure: dose is the sole predictor, {both Bayesian tree procedures enforce monotonicity for MTD inference}, and controlled EWOC limits escalation. Across seven dose--toxicity curves with 45 patients, these models balance safety and MTD precision without prespecifying curve shape, showing that flexible MTD estimation remains practical in this small-sample setting. We also applied all five methods to the $O^6$-benzylguanine example to estimate an AGT-based dose target. Their final doses of 80--100\,mg/m$^2$ are broadly consistent with the 100\,mg/m$^2$ biochemical modulatory dose reported by \citet{friedman1998}; see Appendix~\ref{app:bart-case}.

{The MTD targets toxicity. If efficacy plateaus or declines, a lower dose may offer greater benefit relative to risk. A phase I/II extension could model efficacy without a monotonicity constraint and recommend the dose with the highest posterior expected benefit--risk utility among those meeting the toxicity criterion.}

Future work will compare deep reinforcement learning \citep{matsuura2022dose,matsuura2023phasei}, Super Learner ensembles \citep{vanderlaan2007superlearner}, and Gaussian processes \citep{gotovos2013lse,losalka2023monotone,chien2024safet} under state-dependent dose controls, and quantify the effects of finite rollout budgets and approximate posterior updates.

\FloatBarrier
\label{page:main-end}
\clearpage

\section*{Ethics statement}
\label{page:statements-start}
The study uses simulated trials and previously reported pharmacodynamic measurements. It enrolls no participants.

\section*{Reproducibility statement}
The supplement specifies the models, priors, allocation and stopping rules,
simulation scenarios, posterior computation, and evaluation criteria. It
provides the RL formulation, proofs, rollout algorithm, and policy comparison. The accompanying
software package provides a Windows installer for \textsf{\textbf{Dose Trial Lab}},
which implements the five dose-finding models and the rollout algorithm.
Appendix~\ref{app:software-manual} gives an illustrated usage guide.

\section*{AI Use Statement}
Large language models (LLMs) assisted with grammatical checking and language polishing of the manuscript, and with software development for \textsf{\textbf{Dose Trial Lab}}. The authors are responsible for the mathematical
claims, empirical results, citations, and final text.
\label{page:references-start}
\bibliography{bart_references}
\bibliographystyle{arxiv_preprint}
\clearpage
\appendix
\pdfbookmark[0]{Supplementary material}{bart.supplement}
\section*{Guide to the supplementary material}
\label{app:bart-supplement}

\begin{center}
\small
\setlength{\tabcolsep}{3pt}
\renewcommand{\arraystretch}{1.0}
\begin{tabular*}{\textwidth}{@{}@{\extracolsep{\fill}}p{0.10\textwidth}p{0.78\textwidth}r@{}}
\toprule
Section & Contents & Page\\
\midrule
\multicolumn{3}{@{}l}{\textbf{A. Related methods}}\\
\ref{app:related-methods} & Clinical outcomes and decision targets & \pageref{app:related-methods}\\
\addlinespace
\multicolumn{3}{@{}l}{\textbf{B--E. Models and allocation}}\\
\ref{sec:bart-target} & MTD definition, posterior inversion, and EWOC allocation & \pageref{sec:bart-target}\\
\ref{sec:bart-priors} & Tree priors, isotonic projection, and prior sensitivity & \pageref{sec:bart-priors}\\
\ref{app:bayesian-comparators} & 1PLD, 2PLD, and 3PND likelihoods and priors & \pageref{app:bayesian-comparators}\\
\ref{app:method-dags} & Mean dose--toxicity curves and graphical models & \pageref{app:method-dags}\\
\addlinespace
\multicolumn{3}{@{}l}{\textbf{F--G. Posterior theory and computation}}\\
\ref{sec:doob-consistency} & MTD posterior consistency & \pageref{sec:doob-consistency}\\
\ref{app:bart-numerics} & Posterior computation & \pageref{app:bart-numerics}\\
\addlinespace
\multicolumn{3}{@{}l}{\textbf{H--I. Reinforcement learning and rollout}}\\
\ref{app:rl-formulation} & Belief-state planning, reachability, and rollout proofs & \pageref{app:rl-formulation}\\
\ref{sec:rl-seven-curve} & Linear 1PLD rollout illustration and budget sensitivity & \pageref{sec:rl-seven-curve}\\
\addlinespace
\multicolumn{3}{@{}l}{\textbf{J. Numerical study}}\\
\ref{sec:bart-experiment} & Dose--toxicity scenarios, enrollment, and performance metrics & \pageref{sec:bart-experiment}\\
\addlinespace
\multicolumn{3}{@{}l}{\textbf{\ref{app:paired-reanalysis}--\ref{app:focused-noise}. Seven-curve model comparisons}}\\
\ref{app:paired-reanalysis} & Complete original seven-curve results & \pageref{app:paired-reanalysis}\\
\ref{app:focused-noise} & Higher-noise results across all five models and seven curves & \pageref{app:focused-noise}\\
\addlinespace
\multicolumn{3}{@{}l}{\textbf{\ref{app:bart-case}. Pharmacodynamic case study}}\\
\ref{app:bart-case} & Sequential pharmacodynamic example & \pageref{app:bart-case}\\
\addlinespace
\multicolumn{3}{@{}l}{\textbf{\ref{app:software-manual}. Software manual}}\\
\ref{app:software-manual} & Installation, settings, trial playback, visualization, and exports & \pageref{app:software-manual}\\
\bottomrule
\end{tabular*}
\end{center}

%\paragraph{Allocation policies and algorithms.}
\section{Related methods}
\label{app:related-methods}

The literature on dose finding in phase I cancer trials is extensive and
diverse. We review the work most closely related to our methods.

\subsection{Binary toxicity responses}

\paragraph{Parametric models.}
Several widely used designs learn from the same binary DLT indicator but
translate it into a dose decision differently. The continual reassessment
method (CRM) updates a one-parameter monotone working model; common
implementations encode prior DLT probabilities across dose levels in a
skeleton and choose the dose whose posterior toxicity estimate is closest
to the target \citep{oquigley1990,oquigley2010crm}. Classical EWOC instead
chooses a lower posterior quantile of the MTD, with a feasibility bound
that directly limits the posterior probability of assigning a dose above
the MTD \citep{babb1998,zacks1998,tighiouart2010}. The Bayesian logistic
regression model (BLRM) uses a two-parameter logistic curve in log dose,
a posterior target interval, and an overdose-probability screen
\citep{neuenschwander2008}. The Bayesian optimal interval (BOIN) design
compares the observed DLT rate at the current dose with prespecified
escalation and de-escalation boundaries \citep{liu2015boin}. Thus, CRM,
EWOC, BLRM, and BOIN all use binary DLT feedback, but their dose-selection
rules are distinct.

\paragraph{Nonparametric models.}
Gaussian-process (GP) models allow more flexible dose--toxicity
relationships. SAFE-T uses multiple-output GPs with a probit link
to model binary toxicity and efficacy in heterogeneous participants
\citep{chien2024safet}. The latent function is continuous-valued,
but the observed toxicity response is binary.
Thus, flexibility in the probability curve does not change the
type of toxicity information entering the likelihood.

\subsection{Continuously measured toxicity outcomes}

\noindent\textbf{Parametric models.}\enspace
Earlier designs retained toxicity severity in forms other than a binary
DLT indicator. \citet{ivanova2009} give a unified local allocation rule
for binary, ordinal, or continuous monotone objectives. EWOC-NETS places
a normalized equivalent toxicity score in a quasi-Bernoulli likelihood
and applies EWOC to the resulting quasi-continuous outcome
\citep{chen2012ewocnets}. These approaches establish that graded toxicity
information can guide phase I allocation, while using a local monotone
rule or a parametric score model rather than a flexible tree-based mean.

Methods using the measured toxicity outcome retain more information
than its binary DLT indicator. The one-parameter linear dose
finder (1PLD) estimates a slope with residual variance treated as
known, while the two-parameter linear dose finder (2PLD) estimates
both slope and residual variance
\citep{eichhorn1973,lee2022}. The three-parameter nonlinear dose
finder (3PND) additionally allows curvature through a power exponent
\citep{lee2023}. Although 3PND accommodates nonlinear mean curves,
it remains parametric because the curve belongs to a specified
finite-dimensional family. In these models, a response threshold
defines DLT events, and the response distribution determines their
dose-dependent probability.

\noindent\textbf{Nonparametric models.}\enspace
Flexible modeling is central when toxicity is continuously measured
because the MTD depends on a tail probability of the response
distribution; misspecifying the mean curve can therefore shift the
clinical target. BART represents the mean as a sum of regression trees
and can recover nonlinear shapes without committing to a linear or
power-law curve \citep{chipman2010}. We develop two monotone Bayesian
tree procedures for this setting: iBART projects each posterior draw
onto nondecreasing functions, whereas mBART imposes monotonicity within
the tree ensemble \citep{chipman2022}. Both retain flexible nonlinear
learning while enforcing the clinical requirement that toxicity not
decrease with dose.

The core contribution is to turn these monotone tree models into a
complete phase I decision framework. Each joint posterior draw of the
mean curve and residual standard deviation induces an MTD draw under
the specified DLT threshold and acceptable toxicity probability. The
resulting MTD posterior determines the EWOC proposal, so
uncertainty in curve shape and outcome variability propagates into each
dose decision. This framework uses the full continuously measured
toxicity outcome, accommodates nonlinear relationships beyond the 1PLD,
2PLD, and 3PND forms, and supports MTD-posterior consistency, a
reachability bound, and an expected-loss bound for exact rollout under
the stated conditions.

Table~\ref{tab:novelty-positioning} shows how representative phase I
dose-finding methods differ in toxicity outcome type, dose--toxicity
model, dose-selection approach, and theoretical scope. The final row
highlights the combination developed here: continuously measured
toxicity, monotone Bayesian trees, EWOC allocation, a separate 1PLD rollout evaluation, and formal
guarantees for MTD learning and policy improvement under the stated
conditions.

\begingroup
\setlength{\intextsep}{5pt}
\begin{table}[H]
\setlength{\abovecaptionskip}{3pt}
\centering
\caption{\textbf{Representative phase I dose-finding methods and their relation to the present study.}}
\label{tab:novelty-positioning}
\scriptsize
\setlength{\tabcolsep}{1.8pt}
\renewcommand{\arraystretch}{1.00}
\begin{tabular*}{\textwidth}{@{\extracolsep{\fill}}
>{\raggedright\arraybackslash}p{.245\textwidth}
>{\raggedright\arraybackslash}p{.100\textwidth}
>{\raggedright\arraybackslash}p{.215\textwidth}
>{\raggedright\arraybackslash}p{.110\textwidth}
>{\raggedright\arraybackslash}p{.288\textwidth}@{}}
\toprule
Selected method & Toxicity outcome type & Dose--toxicity model & Dose-selection approach & Main scope or result \\
\midrule
{Traditional 3+3 \citep{storer1989}}
& \leavevmode{Binary}
& \leavevmode{No fitted curve}
& \leavevmode{DLT count rule}
& \leavevmode{Escalate, expand the cohort, or stop} \\
\addlinespace[.7pt]
CRM \citep{oquigley1990}
& Binary
& One-parameter monotone working model; calibrated skeleton
& Skeleton
& Foundational continual reassessment \\
\addlinespace[.7pt]
Classical EWOC \citep{babb1998}
& Binary
& Parametric DLT-risk and MTD model
& EWOC
& Direct posterior overdose control \\
\addlinespace[.7pt]
BLRM \citep{neuenschwander2008}
& Binary
& Two-parameter logistic model in log standardized dose
& Target interval
& Prior-calibrated, interval-based decisions \\
\addlinespace[.7pt]
BOIN \citep{liu2015boin}
& Binary
& Model-assisted interval construction
& Interval rule
& Simple decision rule with finite- and large-sample properties \\
\addlinespace[.7pt]
Graded-outcome designs \citep{ivanova2009,chen2012ewocnets}
& Ordinal / continuous
& Monotone objective or quasi-Bernoulli NETS model
& Up/down or EWOC
& Retains severity beyond one DLT indicator \\
\addlinespace[.7pt]
Bartroff--Lai rollout \citep{bartroff2010dp}
& Binary
& Two-parameter logistic model
& EWOC rollout
& Combines cumulative patient loss and terminal estimation loss \\
\addlinespace[.7pt]
1PLD, 2PLD, and 3PND \citep{eichhorn1973,lee2022,lee2023}
& Continuous
& Linear (1PLD/2PLD) or nonlinear power (3PND)
& EWOC
& Parametric mean; residual variation learned in 2PLD and 3PND \\
\midrule
\textbf{This paper}
& \textbf{Continuous}
& \textbf{Monotone Bayesian Trees (iBART/mBART) with 1PLD, 2PLD, and 3PND as comparators}
& \textbf{EWOC for all five models, rollout evaluated with 1PLD}
& \textbf{MTD-posterior consistency, a reachability bound, and an expected-loss bound for exact rollout under stated conditions} \\
\bottomrule
\end{tabular*}
\end{table}
\endgroup

\subsection{\texorpdfstring{{Machine learning methods for dose finding}}{Machine learning methods for dose finding}}
\label{app:ml-dose-finding}

{At ICML, dose-finding studies have examined both allocation and final dose selection. In a simulation based on a rheumatoid arthritis trial, \citet{katzsamuels2018feasible} studied how to identify doses meeting efficacy and adverse-event criteria under a fixed budget. \citet{shen2020seeda} used a bandit rule to maximize cumulative efficacy during a trial subject to a toxicity constraint. At AISTATS, \citet{lee2020c3t} allocated a limited enrollment budget across subgroups while learning a dose for each. \citet{lee2021sdf} used a Bayesian design to search for the MTD on a grid of permissible drug combinations.}

{The dosing objective also depends on the outcome. At ICML, \citet{chien2024safet} used separate multi-output Gaussian processes for toxicity and efficacy to recommend doses for patient subgroups. In personalized dosing, \citet{demirel2022escada} evaluated a safety-constrained bandit for mealtime insulin dosing in simulated patients (NeurIPS). \citet{yao2026leon} evaluated a language-model optimizer on a patient-specific warfarin dosing benchmark (ICLR). In our phase I setting, the target is the highest dose at which the probability of reaching or exceeding the toxicity threshold remains acceptable. Each patient's continuous measurement informs that probability, and EWOC uses the resulting MTD posterior to propose a dose for the next cohort.}

\subsection{Clinical objectives and sequential planning}

Beyond the response type and curve model, dose-finding designs
differ in their clinical objectives. Related BART-based approaches
address broader phase I/II dose optimization. DOD-BART combines
prognostic factors and accruing outcomes in seamless phase I/II
allocation \citep{zhao2024dodbart}, while DOD-PRO-BART incorporates
patient-reported outcomes alongside clinician-reported toxicity
and efficacy \citep{chung2025dodprobart}. Our study focuses on
phase I dose finding with a single continuously measured toxicity outcome,
where the primary objective is to estimate the MTD and guide dose
allocation under a tightly limited patient budget.

Bayesian planning represents accumulated information through a
posterior belief state, allowing decisions to account for immediate
outcomes and subsequent learning. \citet{ross2011bayesian} develop
this approach for partially observed systems with unknown models,
while \citet{guez2012bayes} study sample-based Bayes-adaptive search.
In phase I design, \citet{bartroff2010dp} formulate Bayesian dynamic
programming with cumulative patient losses and terminal estimation
loss, and study rollout with EWOC as the base policy. Deep
reinforcement learning has also been used for adaptive allocation
in dose--response studies \citep{matsuura2022dose} and dose
escalation in phase I oncology trials \citep{matsuura2023phasei}.

Safe optimization and level-set learning provide related approaches
to sequential exploration using GP confidence bounds
\citep{sui2015safeopt,gotovos2013lse}. Monotone variants address
safety, regret, and safe-boundary learning
\citep{losalka2023monotone,losalka2024noregret}, while robust
level-set methods consider environmental uncertainty
\citep{inatsu2021robustlse}. Their function-value feedback and
function-level safety objectives differ from the patient-level
toxicity probability induced by a continuous response distribution.

Our planning state contains the joint posterior, cumulative DLT count,
enrollment, and last administered dose; these variables determine the next
transition and permitted actions. The
upward cap and no-decrease restriction yield the terminal-error
bound in Proposition~\ref{prop:rl-reachability}, and the final loss
evaluates the last administered dose. The rollout result specializes
exact expected-loss policy improvement to these dose controls,
cohort stopping rules, and terminal estimation convention.

\subsection{Clinical risk and posterior feasibility}
\label{app:clinical-goal}

The clinical target is
$p_0(x)=\Pr_0(Y\geq\eta\mid X=x)\leq1-\gamma$.
For a continuously measured outcome, this probability depends on both the
mean response and residual variability. The EWOC feasibility
criterion concerns posterior uncertainty about the modeled risk,
with nominal bound $\alpha$. The clinical risk limit $1-\gamma$
and the posterior feasibility bound $\alpha$ therefore have
different roles. For iBART, the modeled risk is a projected tail
functional. A no-decrease restriction can require a dose exceeding
the nominal feasibility bound, which is recorded as an allocation
exception. The observed DLT-count stopping threshold is separate
from both probability bounds.
Figure~\ref{fig:supp-clinical-goal} illustrates the measured-outcome
target and the effect of residual variability on its dose location.

\begin{figure}[!htbp]
	\centering
	\includegraphics[width=\textwidth]{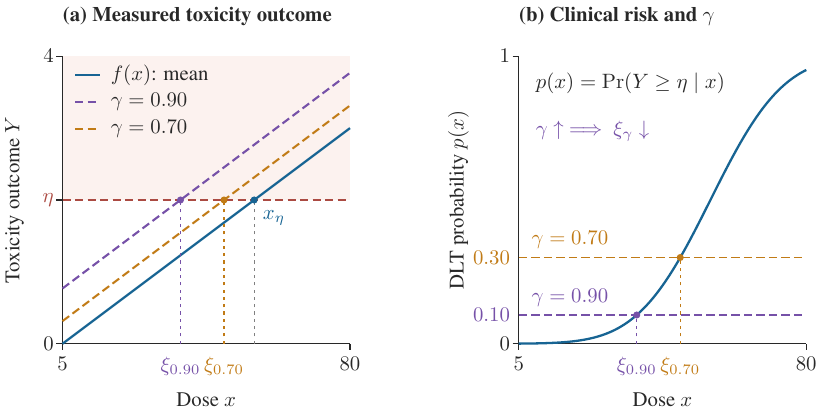}
	\caption{\textbf{A continuously measured toxicity outcome, the clinical target, and the MTD.}
		Gaussian schematic with $f(x)=0.04(x-5)$, $\sigma=0.6$, and $\eta=2$.
		The subscript in $\xi_\gamma$ identifies the required non-DLT probability.
		(a) The response quantile $f(x)+\sigma\Phi^{-1}(\gamma)$ reaches $\eta$ at the MTD;
		the mean reaches it at $x_\eta$, where DLT probability is $1/2$.
		(b) Increasing $\gamma$ lowers the acceptable risk $1-\gamma$ and the MTD;
		the underlying risk curve is unchanged.
		Dashed response curves describe outcome variability, not posterior uncertainty.
		The schematic follows the interpretation in \citet[Figure~3]{lee2023}.}
	\label{fig:supp-clinical-goal}
\end{figure}

\FloatBarrier
\subsection{\texorpdfstring{{Comparison with other dose-finding methods}}{Comparison with other dose-finding methods}}
\label{app:comparison-dose-finding}
{The conventional 3+3 design and BOIN \citep{liu2015boin} use binary toxicity outcomes (DLT or no DLT). Our methods use the continuous measurements directly. In Figure~\ref{fig:continuous-binary-comparison}, the responses at dose 30 lie closer to the DLT threshold than those at dose 15, yet both cohorts yield zero DLTs. A head-to-head comparison would change both the information used for allocation and the dose-selection rule. A benchmark could use the same dose grid, DLT threshold, starting dose, and enrollment limit for all designs, allow each design its own allocation and stopping rules, and score its final recommendation against the same target. BOIN can use the prespecified DLT target, whereas the conventional 3+3 rule uses fixed count boundaries. Table~\ref{tab:novelty-positioning} places our five continuous-outcome procedures in context.}

% Original schematic. Both panels use the Gaussian example in Appendix A.4.
\begin{figure}[!htbp]
\centering
\resizebox{\linewidth}{!}{%
\begin{tikzpicture}[
  font=\small,
  axis/.style={draw=black!65,line width=.4pt},
  curve/.style={draw=blue!55!black,line width=1.1pt},
  boundary/.style={draw=black!65,dashed,line width=.55pt},
  drop/.style={draw=black!55,densely dotted,line width=.65pt},
  quiet/.style={font=\footnotesize,inner sep=1.4pt},
  ticklabel/.style={font=\footnotesize,inner sep=1.8pt}]
% The measured response triples are (dose, Y, horizontal offset in panel (b)).
% Offsets separate coincident binary outcomes; both panels use the same doses.
\def\cohorts{15/.40/-2,15/.68/0,15/.95/2,
             30/1.65/-2,30/1.79/0,30/1.92/2,
             50/1.41/-1.05,50/1.84/1.05,50/2.23/0,
             65/1.83/0,65/2.40/-1.05,65/2.86/1.05}
% For f(x)=.04(x-5), sigma=.6, eta=2, gamma=.9.
\def\targetdose{35.77673}
\begin{scope}[x=.081cm,y=.94cm]
  \node[font=\bfseries] at (41,3.89) {(a)\enspace Continuous toxicity (this study)};
  \draw[axis] (10,3.58) -- (10,0) -- (72,0);
  \foreach \y in {0,1,2,3}
    \draw[axis] (10,\y) -- ++(-.7,0) node[left,ticklabel]{\y};
  \foreach \x in {15,30,50,65}
    \draw[axis] (\x,0) -- ++(0,-.07) node[below,ticklabel]{\x};
  \node[rotate=90] at (2,1.80) {Measured response $Y$};
  \node at (41,-.64) {Dose $x$};
  \draw[boundary] (10,2) -- (72,2);
  \node[quiet,anchor=south west] at (10.7,2.03) {DLT threshold $\eta$};
  \draw[drop] (\targetdose,0) -- (\targetdose,2);
  \node[quiet,anchor=north] at (\targetdose,-.03) {$\xi$};
  \draw[curve] (10,.96893) -- (72,3.44893);
  \node[quiet,text=blue!55!black,anchor=west] at (15.5,2.94) {90th percentile of $Y$};
  \foreach \x/\v/\offset in \cohorts {
    \pgfmathtruncatemacro{\dlt}{\v>=2}
    \ifnum\dlt=1
      \fill[black!80] (\x,\v) circle[radius=1.85pt];
    \else
      \draw[draw=black!80,fill=white,line width=.65pt] (\x,\v) circle[radius=1.85pt];
    \fi
  }
\end{scope}
\begin{scope}[xshift=6.65cm,yshift=.1749cm,x=.081cm,y=3.18cm]
  \node[font=\bfseries] at (41,1.095) {(b)\enspace Binary toxicity (3+3 and BOIN)};
  \draw[axis] (10,1.06) -- (10,-.055) -- (72,-.055);
  \foreach \y in {0,1}
    \draw[axis] (10,\y) -- ++(-.7,0) node[left,ticklabel]{\y};
  \foreach \x in {15,30,50,65}
    \draw[axis] (\x,-.055) -- ++(0,-.020) node[below,ticklabel]{\x};
  \node[rotate=90] at (2,.47) {DLT indicator or probability};
  \node at (41,-.244) {Dose $x$};
  \draw[boundary] (10,.10) -- (72,.10);
  \node[quiet,anchor=south west] at (10.6,.111) {Target $\theta=0.10$};
  \draw[drop] (\targetdose,-.055) -- (\targetdose,.10);
  \node[quiet,anchor=north] at (\targetdose,-.063) {$\xi$};
  % p(x)=Phi((.04(x-5)-2)/.6), including the crossing at (xi,.10).
  \draw[curve] plot[smooth] coordinates {(10,0.001350) (12,0.002074) (14,0.003135) (16,0.004661) (18,0.006819) (20,0.009815) (22,0.013903) (24,0.019383) (26,0.026598) (28,0.035930) (30,0.047790) (32,0.062597) (34,0.080757) (35.77673,0.100000) (36,0.102637) (38,0.128537) (40,0.158655) (42,0.193062) (44,0.231678) (46,0.274253) (48,0.320369) (50,0.369441) (52,0.420740) (54,0.473424) (56,0.526576) (58,0.579260) (60,0.630559) (62,0.679631) (64,0.725747) (66,0.768322) (68,0.806938) (70,0.841345) (72,0.871463)};
  \node[quiet,text=blue!55!black,anchor=south east] at (69,.82) {DLT probability};
  \foreach \x/\v/\offset in \cohorts {
    \pgfmathtruncatemacro{\dlt}{\v>=2}
    \ifnum\dlt=1
      \fill[black!80] ({\x+\offset},1) circle[radius=1.85pt];
    \else
      \draw[draw=black!80,fill=white,line width=.65pt] ({\x+\offset},0) circle[radius=1.85pt];
    \fi
  }
\end{scope}
\end{tikzpicture}}
\caption[Continuous toxicity and binary DLT outcomes.]{{The same 12 illustrative patients on two outcome scales. Open circles indicate no DLT and filled circles indicate a DLT. Panel (b) records $\mathbf{1}\{Y\geq 2\}$ for each patient and spreads the points within dose for visibility. The 90th percentile of $Y$ reaches $\eta=2$ at the same dose $\xi$ where the underlying DLT probability reaches $0.10$. The risk curve illustrates the target. The binary rules act on observed DLT counts. The two perspectives follow \citet[Figure~1]{lee2023}.}}
\label{fig:continuous-binary-comparison}
\end{figure}
\FloatBarrier

\section{The MTD posterior and allocation rule}
\label{sec:bart-target}
\subsection{Clinical specification}
Let $x$ denote dose in its original units and $Y(x)$ a continuously measured toxicity outcome. Before enrollment, clinicians specify the interval
$\mathcal X=(x_{\min},x_{\max})\subset(0,\infty)$, the maximum toxicity level $\eta>0$, and the homogeneity constant $\gamma\in(1/2,1)$ \citep{lee2022,lee2023}. Set $\theta=1-\gamma$. The criterion $\Pr\{Y(x)<\eta\}\geq\gamma$ defines an acceptable dose. These quantities determine the target; the feasibility bound $\alpha$ determines how doses are selected under posterior uncertainty.

\subsection{Likelihood}
\label{sec:bart-likelihood}
Write $\mathcal F_n$ for the history of the first $n$ dose--response pairs $\{(X_i,Y_i)\}_{i=1}^n$, where $i$ indexes patients and $N$ is planned enrollment. For a cohort, doses are selected using the history at the preceding completed-cohort boundary and remain fixed until all cohort responses are observed. For iBART, let $g$ denote the unconstrained mean. The sequential working model is
\begin{equation}
 Y_i\mid\mathcal F_{i-1},X_i,g,\sigma\sim \mathcal N\{g(X_i),\sigma^2\},\qquad \sigma>0.
 \label{eq:bart-observation}
\end{equation}
The allocation kernel is the same function of past observations under every candidate parameter. Prespecified dose restrictions are also functions of that history and do not change this property. The allocation factors therefore cancel from the posterior:
\begin{equation}
 d\Pi(g,\sigma\mid\mathcal F_n)
 \propto\prod_{i=1}^{n}\sigma^{-1}
 \phi\{(Y_i-g(X_i))/\sigma\}\,d\Pi(g,\sigma).
 \label{eq:bart-posterior}
\end{equation}
For mBART, replace $g$ by its monotone mean $f$ in \eqref{eq:bart-observation} and \eqref{eq:bart-posterior}. In iBART, the unconstrained draw $g$ yields $f=\mathcal P g$ by projection onto the nondecreasing cone in $L_2([x_{\min},x_{\max}])$. The pool-adjacent-violators algorithm (PAVA) uses the lengths of the constant dose intervals as weights. In mBART, each tree is nondecreasing and no projection is needed. In both cases the mean draw retains its corresponding residual standard deviation (SD).

\subsection{MTD posterior}
For the projected curve $f=\mathcal P g$ in iBART or the directly sampled monotone curve $f$ in mBART, let
\begin{equation}
 \begin{split}
 p_{f,\sigma}(x)&=\Phi\{(f(x)-\eta)/\sigma\},\\
 \xi(f,\sigma)&=\sup\bigl(\{x_{\min}\}\cup
             \{x\in(x_{\min},x_{\max}):p_{f,\sigma}(x)\leq1-\gamma\}\bigr).
 \end{split}
 \label{eq:bart-root}
\end{equation}
The lower endpoint $\xi=x_{\min}$ represents the absence of an acceptable admissible dose; $\xi=x_{\max}$ represents an entirely acceptable interval. The posterior of $\xi$ is supported on the ordinary dose interval $[x_{\min},x_{\max}]$. Neither endpoint is administered.

Applying \eqref{eq:bart-root} to every paired mean and variance draw propagates their joint uncertainty to the MTD. Inverting a posterior mean curve generally produces a different quantity.

\subsection{The feasibility bound \texorpdfstring{$\alpha$}{alpha}}
Write $\Pi_n^\xi(x)=\Pr(\xi\leq x\mid\mathcal F_n)$. The common EWOC proposal is the posterior $\alpha$-quantile:
\begin{equation*}
 D_\alpha(\mathcal F_n)
   =(\Pi_n^\xi)^{-1}(\alpha)
   =\inf\{x\in[x_{\min},x_{\max}]:\Pi_n^\xi(x)\geq\alpha\}.
\end{equation*}
When the posterior is continuous, $\Pi_n^\xi\{D_\alpha(\mathcal F_n)\}=\alpha$. With atoms, the generalized quantile satisfies
$\Pi_n^\xi\{D_\alpha(\mathcal F_n)-\}\leq\alpha
 \leq\Pi_n^\xi\{D_\alpha(\mathcal F_n)\}$.
A smaller $\alpha$ selects a lower posterior quantile. It changes the allocation rule while leaving the MTD definition and the response likelihood unchanged.

The feasibility condition for the next dose is
\begin{equation}
 \Pr\{p_{f,\sigma}(x_{n+1})>1-\gamma\mid\mathcal F_n\}\leq\alpha.
 \label{eq:bart-feasible}
\end{equation}
Right-continuous tree curves can jump at the MTD. We move inward from
the quantile $q_n=D_\alpha(\mathcal F_n)$, using a prespecified numerical
margin $\delta_{\mathrm{num}}>0$:
\[
 \widetilde x_{n+1}
 =\max\{x_{\min}+\delta_{\mathrm{num}},
             \min(x_{\max},q_n)-\delta_{\mathrm{num}}\}.
\]
The numerical study uses $\delta_{\mathrm{num}}=7.5\times10^{-7}$
dose units. Apply an optional prespecified upward cap first,
then the optional \emph{Only Escalation} restriction:
\begin{equation}
 \begin{split}
 x^{\mathrm{cap}}_{n+1}
   &=\min\{\widetilde x_{n+1},X_n+\Delta_{\max}\},
       \qquad 0\leq\Delta_{\max}\leq\infty,\\
 X_{n+1}&=
 \begin{cases}
  \max\{X_n,x^{\mathrm{cap}}_{n+1}\},&\text{Only Escalation enabled},\\
  x^{\mathrm{cap}}_{n+1},&\text{otherwise}.
 \end{cases}
 \end{split}
 \label{eq:bart-restricted-dose}
\end{equation}
The cap limits increases only; Only Escalation permits a hold but forbids
a decrease. These options can be used separately or together. The
primary five-model comparison uses $\Delta_{\max}=3.5$ with Only
Escalation enabled.

We evaluate \eqref{eq:bart-feasible} at the actual assigned dose.
An inward proposal strictly below $q_n$ satisfies the nominal bound;
the upward cap alone can only lower this proposal and preserves that
property. A binding lower floor can violate the bound. Only Escalation
can also violate it by holding the current dose above a newly lowered
proposal. We record the assigned-dose posterior exceedance probability
and flag these exceptions to nominal $\alpha$; they do not trigger stopping.
If the lower floor binds, it is assigned unless Only Escalation requires
holding the current dose. In particular, posterior mass above $\alpha$
at $\xi=x_{\min}$ implies that no interior dose satisfies the bound,
but does not terminate this protocol.
For iBART the feasibility calculation concerns the projected tail
functional, not the unprojected model's predictive risk.

The first dose $x_1$ is prespecified in $(x_{\min},x_{\max})$, and
$0<\delta_{\mathrm{num}}<
\min\{x_1-x_{\min},x_{\max}-x_1\}$.
The numerical study uses $x_{\min}=5$ and $x_1=6$.
Initialization is included in the patient
budget and remains an exception to the posterior allocation rule.
Optional dose restrictions can accompany the quantile proposal
\citep[Section~4.2]{lee2022}.

\subsection{Cumulative DLT stopping}
Set $D_n=\sum_{i=1}^n\mathbf1\{Y_i\geq\eta\}$ and
$K_N=\lfloor0.1N\rfloor$ before enrollment. After observing every response
in a completed cohort and updating the posterior, stop if $D_n>K_N$
or $n=N$. Thus $N=45$ gives $K_N=4$, and the fifth DLT triggers the
count rule. A cohort can cross the threshold by more than one DLT.
Neither the posterior endpoint mass nor the numerical distance of the
EWOC proposal from $x_{\min}$ is a stopping criterion. The terminal
posterior includes all responses from the final cohort.
Figure~\ref{fig:sequential-workflow} shows the history updates used for
subsequent decisions.

Algorithm~\ref{alg:bart-dose-finding} gives the complete cohort sequence
for the two tree procedures, including the final posterior update.

\begin{algorithm}[tb]
\caption{Tree-model implementation of EWOC: iBART and mBART}
\label{alg:bart-dose-finding}
\begin{algorithmic}[1]
\Require Dose interval $(x_{\min},x_{\max})$, initial dose $x_1$, prior, cuts, thresholds $\eta,\theta$, allocation parameter $\alpha$, cohort size $c$, budget $N$, optional upward cap $\Delta_{\max}$ and Only Escalation setting, numerical margin $\delta_{\mathrm{num}}$, and posterior sample size $S$.
\State Set $\mathcal H=\varnothing$, $n=0$, $D_0=0$, $K_N=\lfloor0.1N\rfloor$, and $x=x_1$.
\While{$n<N$}
  \State Assign $c'=\min(c,N-n)$ patients to $x$ and observe all their responses.
  \State Append the observations to $\mathcal H$; set $n\gets n+c'$ and $D_n=\sum_{i=1}^n\mathbf1\{Y_i\geq\eta\}$.
  \If{using iBART}
    \State Draw $(g^{(s)},\sigma^{(s)})_{s=1}^S$; set $f^{(s)}=\mathcal P g^{(s)}$ for each $s$.
  \Else
    \State Draw $(f^{(s)},\sigma^{(s)})_{s=1}^S$ directly from the mBART posterior.
  \EndIf
  \For{$s=1,\ldots,S$}
    \State Compute $\xi^{(s)}\in[x_{\min},x_{\max}]$ by \eqref{eq:bart-root}.
  \EndFor
  \State Record posterior summaries and last-assigned-dose metrics.
  \If{$n=N$ or $D_n>K_N$} \State \textbf{break} \EndIf
  \State Let $q_\alpha$ be the lower empirical $\alpha$-quantile of the MTD draws.
  \State Set $\widetilde x=\max\{x_{\min}+\delta_{\mathrm{num}},\min(x_{\max},q_\alpha)-\delta_{\mathrm{num}}\}$.
  \State Set $x'=\min\{\widetilde x,x+\Delta_{\max}\}$.
  \State If Only Escalation is enabled, set $x'\gets\max\{x,x'\}$.
  \State Report $\widehat u=S^{-1}\sum_s\mathbf1\{p_{f^{(s)},\sigma^{(s)}}(x')>\theta\}$ and flag any floor or held-dose exception $\widehat u>\alpha$.
  \State Set $x\gets x'$; the nominal $\alpha$ check does not determine stopping.
\EndWhile
\State Set $n_{\mathrm{end}}=n$ and retain the final posterior.
\end{algorithmic}
\end{algorithm}
The primary five-model comparison uses $N=45$, cohorts of three, $x_1=6$,
$\delta_{\mathrm{num}}=7.5\times10^{-7}$, and $\Delta_{\max}=3.5$,
with Only Escalation enabled. All assignments remain in the open interval.
The first cohort is prespecified; subsequent exceedance probabilities
are estimated from posterior draws. The terminal posterior includes the
final cohort even when its DLT count triggers stopping.

\begin{figure}[!htbp]
 \centering
 \includegraphics[width=\textwidth]{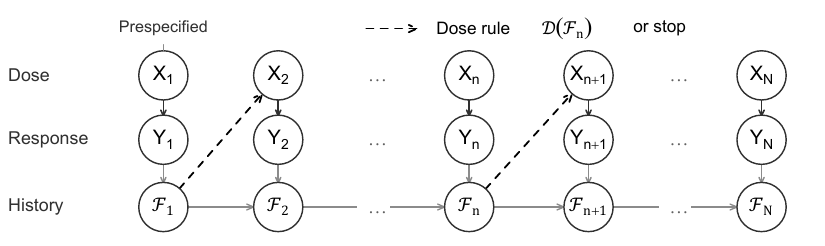}
 \caption{\textbf{Sequential learning during a dose-finding trial.}
 The initial dose is prespecified. After $n$ patients, the cumulative history
 $\mathcal F_n$ contains the first $n$ dose--response pairs and determines the
 updated posterior belief. The selected rule $\mathcal D(\mathcal F_n)$
 chooses the next dose $X_{n+1}$ or stops enrollment;
 it may use 1PLD, 2PLD, 3PND, iBART, or mBART with feasibility bound
 $\alpha$. Solid black arrows denote observation, gray arrows information
 updates, and dashed arrows dose decisions. Ellipses omit intermediate
 updates; the displayed path ends at the budget $N$ and may stop earlier.
 With cohorts, doses are assigned before that cohort's responses are
 observed, and the rule updates after the cohort is complete.}
 \label{fig:sequential-workflow}
\end{figure}
\section{Tree priors and isotonic projection}
\label{sec:bart-priors}
\subsection{Prior specification}
Both procedures use $m=200$ trees and shrinkage $k=2$.
At depth $d$, a node splits with probability $0.95(1+d)^{-2}$ in
BART and $0.25(1+d)^{-0.8}$ in mBART
\citep{chipman2010,chipman2022}. The response reference range $[0,4]$
gives mean offset $b_0=2$ and Gaussian base leaf standard deviations
$\tau_{\mathrm{iBART}}=1/\sqrt{200}$ and
$\tau_{\mathrm{mBART}}=\sqrt{1.467}/\sqrt{200}$.
This range calibrates the prior; Gaussian responses are not restricted
to it. Our implementation applies the inflated mBART base scale to all
leaves, including single-leaf trees. Section~3.3 of \citet{chipman2022}
specifies inflation only for means subject to monotonicity constraints.
Independently of the mean, the residual variance has prior
$\sigma^2\sim\nu_\sigma\lambda/\chi^2_{\nu_\sigma}$, with
$\nu_\sigma=3$ and $\lambda=0.20^2F^{-1}_{\chi^2_3}(0.10)/3$,
so $\Pr(\sigma<0.20)=0.90$.

{The split probabilities govern tree depth and thus how many steps a curve can have; leaf shrinkage governs how far fitted levels can depart from the prior offset. Residual variation enters the MTD through the condition $f(x)\leq\eta-\sigma\Phi^{-1}(\gamma)$. For $\eta=3$ and $\gamma=0.9$, raising $\sigma$ from 0.1 to 0.2 lowers the largest acceptable mean from 2.872 to 2.744. For a fixed curve this can move the boundary to a lower dose. Across posterior draws, the lower MTD quantile determines the EWOC proposal; two different proposals can still yield the same capped dose. The full-trial comparisons use fixed priors. Appendix~\ref{app:sqkv-prior-simulation} varies each prior setting on the same observed histories.}

The offset, leaf scales, variance prior, and cut set are fixed before
enrollment. Each tree node may split only at one of the 100 candidate dose thresholds
$x_{\min}+j(x_{\max}-x_{\min})/101$, $j=1,\ldots,100$;
in the simulation these are $5+75j/101$. Both samplers use doses and
cuts in their original units. A split is eligible whenever its
prespecified cut lies inside the node, including when a resulting
child has no observations.

\subsection{Isotonic projection of BART draws}
iBART fits the unconstrained Gaussian BART likelihood. For each
posterior draw, let $g_j$ be its value on dose cell $B_j$. Its
Lebesgue $L_2([x_{\min},x_{\max}])$ projection is obtained by
weighted isotonic regression:
\begin{equation}
 (z_1,\ldots,z_J)
 =\arg\min_{z_1\leq\cdots\leq z_J}
             \sum_{j=1}^{J}|B_j|(g_j-z_j)^2,\qquad
 (\mathcal P g)|_{B_j}=z_j.
 \label{eq:bart-pava}
\end{equation}
PAVA solves this problem with physical cell lengths as weights \citep{deleeuw2009}.
The projected mean retains its paired residual-SD draw when computing
risk and the MTD. The projection is a posterior functional; it does
not replace $g$ in the fitted likelihood.

\subsection{Monotone tree priors}
For reference, the two continuous normalizations are
\begin{equation}
 \underbrace{\rho(\mathcal T)K!\prod_{j=1}^{K}\phi_\tau(\mu_j)\mathbf1_{\rm ordered}}_{\text{conditional ordered-leaf prior}}
 \quad\text{and}\quad
 \underbrace{\widetilde\pi(\mathcal T,\mu)\propto\rho(\mathcal T)\prod_{j=1}^{K}\phi_\tau(\mu_j)\mathbf1_{\rm ordered}}_{\text{constrained-product prior}}.
 \label{eq:conference-priors}
\end{equation}

mBART orders the leaves of each univariate tree, making their sum
nondecreasing. For a tree with probability $\rho(\mathcal T)$ and
$K$ independent and identically distributed (iid) continuous Gaussian base heights, conditioning on their
order multiplies the product density by $K!$ and retains tree
marginal $\rho(\mathcal T)$. The constrained-product formulation is
instead
\begin{equation}
 \widetilde\pi(\mathcal T,\mu)\propto
 \rho(\mathcal T)\prod_{j=1}^{K}\phi_\tau(\mu_j)
       \mathbf1\{\mu_1\leq\cdots\leq\mu_K\}.
 \label{eq:bart-main-working-prior}
\end{equation}
Global normalization gives tree marginal proportional to
$\rho(\mathcal T)/K!$. The computation uses this formulation,
approximating leaf integration with 50 points in $[-3\tau,3\tau]$.
Discrete ties change the ordering probability. The implemented grid is
a numerical approximation; its posterior sampling is described in
Appendix~\ref{app:bart-numerics}.
\clearpage
\subsection{\texorpdfstring{{Prior sensitivity of the next EWOC dose}}{Prior sensitivity of the next EWOC dose}}
\label{app:sqkv-prior-simulation}
{How much does a change in the prior affect the next cohort? We held the observed doses and responses fixed across fits to isolate this question. We generated 20 independent histories for each of the seven curves, all with true MTD 50 and residual SD 0.1. Each history had 12 cohorts of three, at doses $6,9.5,\ldots,44.5$, conditional on the trial still being open (at most four DLTs). The next EWOC proposal was the fifth percentile of the MTD posterior. The cap and no-decrease rule confined the assigned dose to $[44.5,48]$.}

{We fit iBART and mBART four times to each history (1,120 fits). Starting from the prior in Appendix~\ref{sec:bart-priors}, we changed one setting at a time: halved the base of the depth-dependent split rule, increased leaf shrinkage from $k=2$ to $3$, or raised the prior 90th percentile of the residual SD from 0.20 to 0.40. The generating SD remained 0.1. Each fit used the original two-chain sampler settings; the prior comparisons used common chain seeds within a history and model. Figure~\ref{fig:sqkv-prior-sensitivity} compares each altered prior with its baseline on the same histories.}

\begin{figure}[!htbp]
\centering
\includegraphics[width=\textwidth]{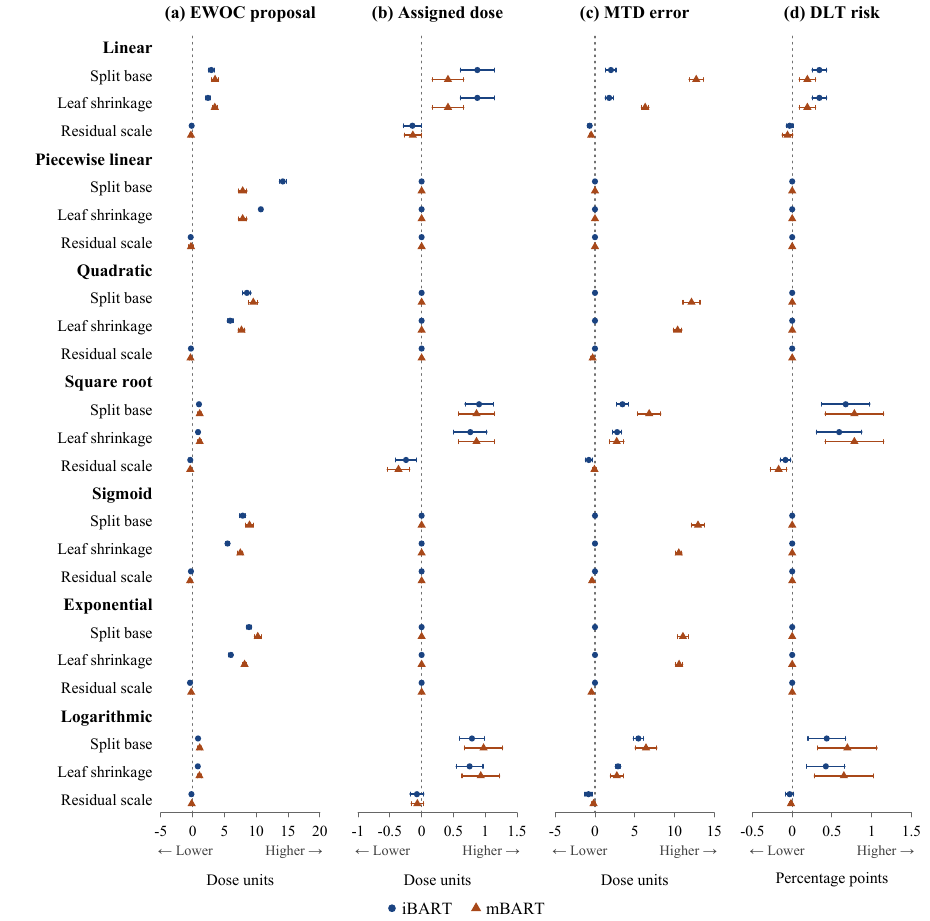}
\caption{{Prior sensitivity after 36 patients. Each point shows the mean paired difference (altered prior minus baseline) across 20 histories per curve. Bars show pointwise 95\% Monte Carlo intervals. Panels (a) and (b) show the EWOC proposal and the assigned dose in dose units, with a narrower axis for the assigned dose. Panel (c) shows the change in absolute error of the posterior median MTD relative to the true MTD of 50. Panel (d) shows the change in the next patient's DLT probability at the assigned dose under the generating curve.}}
\label{fig:sqkv-prior-sensitivity}
\end{figure}

{On the linear curve, halving mBART's split base raised its EWOC proposal in all 20 histories (mean increase 3.53 dose units), but changed the assigned dose in only nine (mean increase 0.41 across all 20). Its mean absolute error for the posterior median MTD increased by 12.74 dose units. In these histories the highest observed dose was 44.5, below the true MTD of 50: similar allocations can conceal substantial differences in extrapolation.}

{Across the seven curves, changing one prior setting changed the EWOC proposal in 653 of 840 paired comparisons and the assigned dose in 213. On the piecewise-linear, quadratic, sigmoid, and exponential curves, every fit assigned the largest permitted dose of 48. Across the 42 model, curve, and prior comparisons, the mean change in next-patient DLT probability ranged from $-0.17$ to $+0.78$ percentage points. These results describe the next cohort after the same 36 observations within each history.}
\FloatBarrier
\section{Parametric dose finders}
\label{app:bayesian-comparators}
All three parametric procedures observe the measured outcome and use the same dose interval and feasibility bound as the tree procedures. Write $q_\gamma=\Phi^{-1}(\gamma)$ and $L_x=x_{\max}-x_{\min}$.

\paragraph{1PLD.}
The Bayesian search using a continuously measured toxicity outcome from \citet{eichhorn1973}, reproduced as the one-parameter linear dose finder in \citet[Eq.~(2)]{lee2022}, uses
\[
 Y_i\mid X_i,\beta\sim \mathcal N\{\beta(X_i-x_{\min}),s^2\},\qquad
 \beta\sim \mathcal N(a,b^2),
\]
with known $s$. The Gaussian slope prior is untruncated. For $(x_{\min},x_{\max})=(5,80)$, we set $a=0.04$, $b=0.02$, and $s=0.1$ before simulation. These are study calibrations.

Let $A=\eta-s\Phi^{-1}(\gamma)>0$ and $L_x=x_{\max}-x_{\min}$. The MTD is
\begin{equation}
 \xi(\beta)=
 \begin{cases}
 x_{\max},&\beta\leq A/L_x,\\
 x_{\min}+A/\beta,&\beta>A/L_x.
 \end{cases}
 \label{eq:onepld-root}
\end{equation}
Negative slopes contribute to the mass at $x_{\max}$. The risk near $x_{\min}$ is below the target for every slope, so this calibration has no posterior mass at the lower endpoint.

\paragraph{Posterior calculation.}
Write $z_i=X_i-x_{\min}$,
$P_n=b^{-2}+s^{-2}\sum_i z_i^2$, and
$S_n=a/b^2+s^{-2}\sum_i z_iY_i$. Completing the square gives
\begin{equation}
 \beta\mid\mathcal F_n\sim \mathcal N(\mu_n,V_n),\qquad
 \mu_n=S_n/P_n,\quad V_n=P_n^{-1}.
 \label{eq:onepld-posterior}
\end{equation}
For $x\in(x_{\min},x_{\max})$,
\begin{align*}
 \Pr\{p_\beta(x)>\theta\mid\mathcal F_n\}
 &=1-\Phi\!\left\{\frac{A/(x-x_{\min})-\mu_n}{\sqrt{V_n}}\right\},\\
 \Pr(\xi=x_{\max}\mid\mathcal F_n)
 &=\Phi\!\left\{\frac{A/L_x-\mu_n}{\sqrt{V_n}}\right\}.
\end{align*}
For $0<u<1$, put $B_u=\mu_n+\sqrt{V_n}\Phi^{-1}(1-u)$. The lower posterior $u$-quantile is $x_{\max}$ if $B_u\leq A/L_x$ and $x_{\min}+A/B_u$ otherwise. The allocation uses $u=\alpha$ and the common dose restrictions.

The known noise SD equals the simulation value $0.1$. This supplies 1PLD with information that the remaining methods estimate.

\subsection{2PLD and 3PND}
The likelihoods of \citet{lee2022}, Eq.~(3), and \citet{lee2023}, Eq.~(4), are
\[
 Y_i\mid X_i,\beta,\nu,\sigma\sim
 \mathcal N\{\beta(X_i-x_{\min})^\nu,\sigma^2\}.
\]
2PLD fixes $\nu=1$; 3PND estimates $\nu>0$. Both fix the mean at zero at $x_{\min}$. Their toxicity risk and MTD are
\begin{equation}
 p_{\beta,\nu,\sigma}(x)=
 \Phi\left\{\frac{\beta(x-x_{\min})^\nu-\eta}{\sigma}\right\},\qquad
 \xi=x_{\min}+
 \left\{\frac{\eta-\sigma q_\gamma}{\beta}\right\}^{1/\nu}.
 \label{eq:comparator-root}
\end{equation}

The priors in \citet{lee2022}, Eqs.~(4)--(5), and \citet{lee2023}, Eqs.~(7)--(9), are
\begin{equation}
 \begin{gathered}
 \pi_\sigma(s)=\frac{\mathbf1\{0<s<\eta/q_\gamma\}}
 {\arctan(\eta/q_\gamma)(1+s^2)},\quad
 \beta\mid\sigma,\nu\sim\operatorname{Unif}\{l(\sigma,\nu),u(\sigma,\nu)\},\\
 l(\sigma,\nu)=\frac{\eta-\sigma q_\gamma}{L_x^\nu},\qquad
 u(\sigma,\nu)=\frac{\eta}{L_x^\nu}+\sigma q_\gamma .
 \end{gathered}
 \label{eq:comparator-prior}
\end{equation}
For 3PND, $\log\nu\sim \mathcal N(0,\delta^2)$ independently of $\sigma$, with $\delta=0.1$ in the simulation. The scale-one half-Cauchy prior is truncated on the standard deviation. The uniform density includes the normalizer
\[
 \{u-l\}^{-1}=\{\sigma q_\gamma(1+L_x^{-\nu})\}^{-1}.
\]
In particular, the $\sigma q_\gamma$ term in the upper endpoint is not divided by $L_x^\nu$.

Independent $U,V\sim\operatorname{Unif}(0,1)$ and $Z\sim \mathcal N(0,1)$ give the exact prior coordinates
\begin{equation}
 \sigma=\tan\{V\arctan(\eta/q_\gamma)\},\qquad
 \nu=\exp(\delta Z),\qquad
 \beta=l(\sigma,\nu)+U\{u(\sigma,\nu)-l(\sigma,\nu)\}.
 \label{eq:comparator-latent-prior}
\end{equation}
For 2PLD, omit $Z$ and set $\nu=1$. These constraints imply $x_{\min}<\xi<x_{\max}$ almost surely. The priors exclude both endpoint states, whereas tree priors permit them.

\subsection{Allocation and evaluation}
The posterior dose rule is $D_\alpha(\mathcal F_n)=F^{-1}_{\xi\mid\mathcal F_n}(\alpha)$. For the increasing parametric curves, the event that dose $x$ is excessive equals $\{\xi<x\}$. Under the 1PLD convention in Eq.~\eqref{eq:onepld-root}, the same equivalence holds for its allocated positive doses; negative slopes contribute to the all-acceptable endpoint. The tree procedures assess the toxicity-risk event directly to account for step boundaries. Posterior feasibility is a statement under each method's own posterior and is not a common fixed-truth frequentist bound.

In Eq.~\eqref{eq:bart-metrics}, BTM and RAE use the last assigned dose; NPD and NPO use individual patient records. Posterior summaries incorporate all observed responses, including the final cohort, and are distinct from the last-dose estimator.
\section{Dose--toxicity curves and graphical models}
\label{app:method-dags}

Figure~\ref{fig:method-dags}
shows what each method learns from the measured outcome and how those
unknowns determine the MTD. Shaded circles denote observed responses
and assigned doses; unshaded circles denote unknown quantities.
Rectangles contain fixed inputs, and diamonds are deterministic
functions of their parents. Solid arrows specify stochastic
dependencies; dashed arrows lead to deterministic functions. The
patient plate repeats the response factor for $i=1,\ldots,n$.

The clinical inputs $C=(x_{\min},x_{\max},\eta,\gamma)$ are fixed
throughout the trial. To keep the diagrams readable, arrows from $C$
are suppressed, including its contribution to the priors and the MTD.
The dose nodes show the values used in the sequential likelihood
$p(Y_i\mid X_i,\mathcal F_{i-1},\vartheta)$, where $\vartheta$ denotes
a method's unknown parameters. The graphs omit the allocation history
and do not assert conditional independence after conditioning on the
entire random sequence of adaptive doses. The allocation kernels
cancel from the posterior for the reason given in
Section~\ref{sec:bart-target}.

\begin{figure}[!ht]
 \centering
 \includegraphics[width=0.90\textwidth]{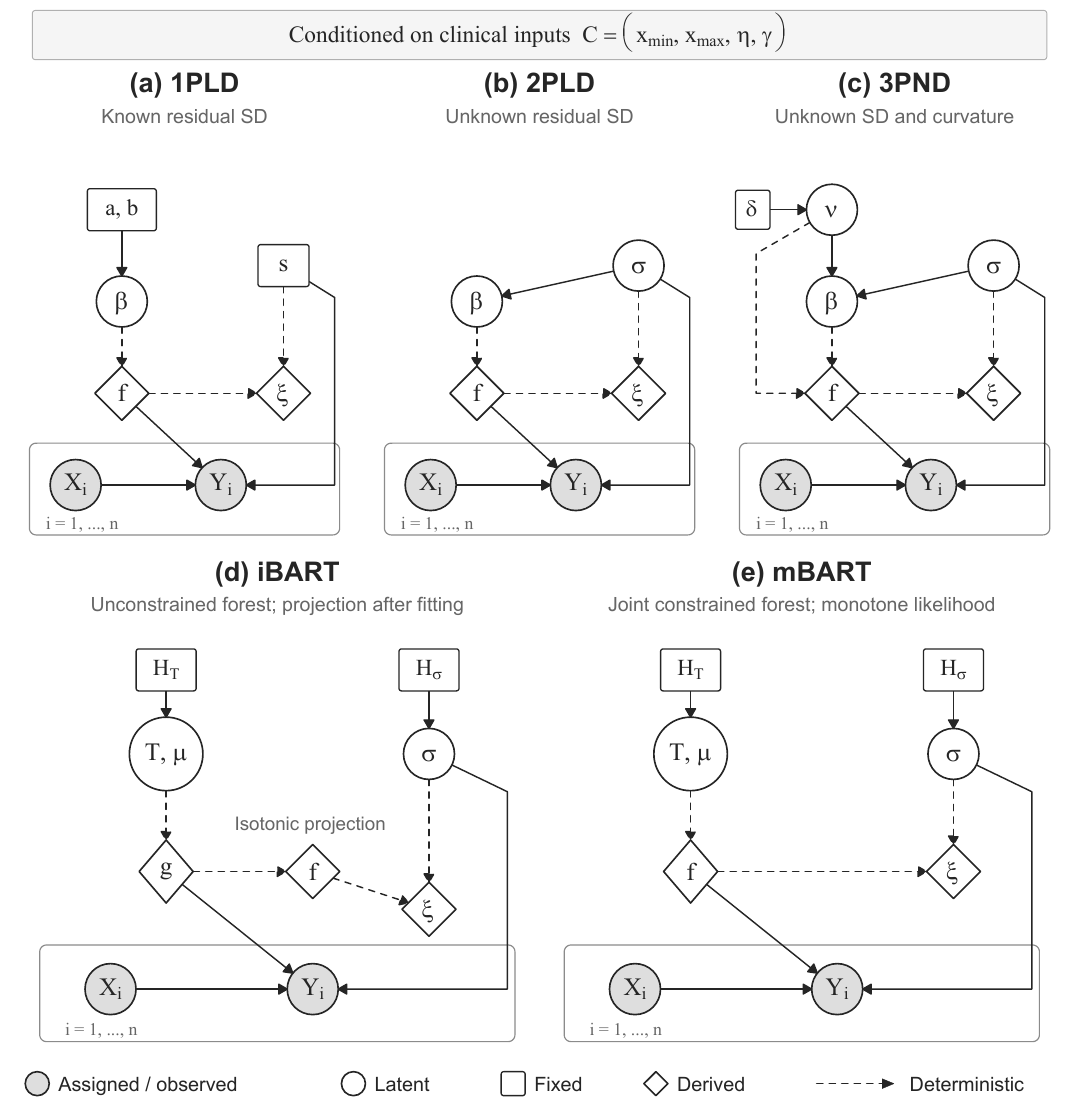}
 \caption{\textbf{Graphical models of the five procedures.}
 Top: 1PLD, 2PLD, and 3PND; bottom: iBART and mBART.
 $X_i$ and $Y_i$ are patient $i$'s assigned dose and response.
 $\beta$ is the mean-curve coefficient; $a,b$ are its 1PLD prior mean and SD;
 $s$ is the known residual SD and $\sigma$ an unknown residual SD.
 The 3PND exponent is $\nu$, with prior SD $\delta$ for $\log\nu$.
 $T$ collects tree structures and $\mu$ their leaf heights.
 $H_T$ fixes the tree count, depth probabilities, cut set, and leaf scale;
 $H_\sigma$ fixes the residual-variance prior's degrees of freedom and scale.
 $g$ is the unconstrained forest mean, $f$ the curve defining the MTD $\xi$.
 The fixed clinical inputs $C=(x_{\min},x_{\max},\eta,\gamma)$ are the dose
 bounds, DLT threshold, and required non-DLT probability.
 The arrows $\sigma\to\beta$ and $\nu\to\beta$ encode conditional priors.
 In iBART, $g$ enters the likelihood while its isotonic projection
 $f=\mathcal P g$ determines the MTD; in mBART, the joint constrained forest gives the monotone mean
 $f$ used by both. Every MTD retains its paired residual SD.
 Allocation history and arrows from $C$ are omitted.}
 \label{fig:method-dags}
\end{figure}

\begin{figure}[!t]
 \centering
 \includegraphics[width=0.95\textwidth]{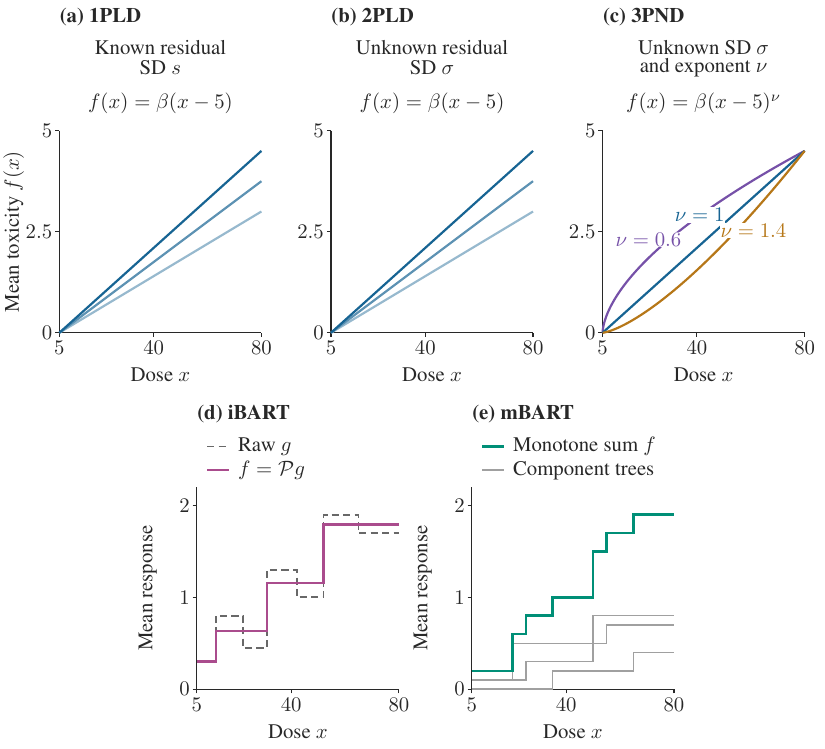}
 \caption{\textbf{Illustrative mean dose--toxicity curves for the five procedures.}
 In (a) and (b), 1PLD and 2PLD share a linear mean family, differing in whether residual
 SD is known or estimated; the examples use $\beta\in\{0.04,0.05,0.06\}$.
 (c) 3PND also estimates the exponent $\nu$; the examples use
 $\nu\in\{0.6,1,1.4\}$ and $\beta=4.5/75^\nu$, giving a common mean
 at the upper dose. Only positive
 slopes are illustrated; the 1PLD prior also permits negative slopes.
 (d) iBART applies cell-length-weighted isotonic projection to the
 unconstrained step curve $g$ (dashed gray), giving $f=\mathcal P g$
 (purple). The likelihood retains $g$, while $f$ and its paired residual SD
 determine the MTD. (e) Three nondecreasing component trees (thin gray)
 sum to a monotone curve $f$ (green), used in both the likelihood and MTD
 calculation. These schematic curves are not fitted estimates.}
 \label{fig:dose-toxicity-curves}
 \label{fig:parametric-dose-curves}
 \label{fig:tree-dose-curves}
\end{figure}

\paragraph{Parametric procedures.}
For 1PLD, the only unknown is the slope $\beta$; the residual standard
deviation $s$ is known. Its prior is the untruncated
$\mathcal N(a,b^2)$ distribution. For 2PLD, $\sigma$ is unknown and the support
of the slope prior depends on it. For 3PND, the exponent $\nu$ also
changes the slope-prior support. Conditional on the fixed inputs,
the respective prior factorizations are
\begin{align*}
 \pi_{\rm 1PLD}(\beta)&=\phi_b(\beta-a),\\
 \pi_{\rm 2PLD}(\beta,\sigma)
   &=\pi_\sigma(\sigma)\,\pi_\beta(\beta\mid\sigma,\nu=1),\\
 \pi_{\rm 3PND}(\beta,\sigma,\nu)
   &=\pi_\sigma(\sigma)\,\pi_\nu(\nu\mid\delta)\,
     \pi_\beta(\beta\mid\sigma,\nu).
\end{align*}
Here $\pi_\sigma$ and $\pi_\beta$ are exactly the truncated half-Cauchy
and conditional uniform densities in
Eq.~\eqref{eq:comparator-prior}, and $\log\nu\sim \mathcal N(0,\delta^2)$.
Thus $\sigma$ and $\nu$ are independent before observing data in 3PND,
but $\beta$ is independent of neither. The functions represented by
the $f$ nodes are $\beta(x-x_{\min})$ for 1PLD and 2PLD, and
$\beta(x-x_{\min})^\nu$ for 3PND. All three use a Gaussian response
likelihood. The MTD node uses Eq.~\eqref{eq:onepld-root} for 1PLD,
including its upper-endpoint convention for small or negative slopes,
and Eq.~\eqref{eq:comparator-root} for 2PLD and 3PND.
Figure~\ref{fig:dose-toxicity-curves}(a)--(c) illustrates their mean-curve families.

\paragraph{Tree procedures.}
In the lower row of Figure~\ref{fig:method-dags}, $T$ denotes the collection of
tree structures and $\mu$ their leaf heights. The fixed inputs $H_T$
specify the tree count, depth probabilities, cut set, and leaf scale;
$H_\sigma$ specifies the residual-variance prior. The fixed mean
offset is also conditioned on. These settings are given in
Section~\ref{sec:bart-priors}. The forest and residual variance are
independent under the prior; observing the responses generally makes
them dependent in the posterior.

For iBART, the forest produces the unconstrained mean $g$. This
$g$ enters the Gaussian likelihood. The separate deterministic
branch applies the weighted isotonic projection $f=\mathcal P g$
and then computes $\xi(f,\sigma)$. There is no arrow from $f$ to
$Y_i$: projection changes the posterior quantity used for dose
selection, not the likelihood that produced the draw. The
corresponding $\sigma$ draw is retained through both calculations.

For mBART, the forest node represents the \emph{joint}
constrained-product prior for structures and leaf heights in
Eq.~\eqref{eq:bart-main-working-prior}. Each tree is nondecreasing,
so the sum $f$ enters both the likelihood and the
MTD calculation. Grouping structures and heights in one node is
deliberate: using the unconstrained topology marginal followed by
a normalized ordered-leaf conditional would instead give the
different prior in Eq.~\eqref{eq:conference-priors}. The numerical
implementation approximates the constrained-product formulation
on the stated leaf grid; the diagram does not identify that finite
approximation with the continuous prior.
Figure~\ref{fig:dose-toxicity-curves}(d)--(e) illustrates projection and monotone summation.

\paragraph{The common EWOC decision.}
The graphs end at the MTD because $\alpha$ is a decision parameter,
not a response-model parameter. After conditioning on the available
patient data, each method computes an MTD posterior and uses its
lower $\alpha$-quantile to propose the next dose. The tree procedures
also check the posterior risk event at the proposed dose, as in
Eq.~\eqref{eq:bart-feasible}, to handle step boundaries. The initial
dose, cohort size, dose restrictions, and stopping rules belong to
this allocation stage. Placing that stage after the posterior
update keeps the diagrams acyclic while preserving the shared
EWOC structure.

%\FloatBarrier
\section{MTD posterior consistency}
\label{sec:doob-consistency}
For the fixed-cell model below, we establish posterior concentration at
the true MTD under prior support, persistent cell sampling, and a strict
risk-threshold margin. The proof uses a Schwartz-type numerator--denominator
argument \citep{schwartz1965}. The margin makes the MTD locally unchanged,
allowing likelihood comparisons to exclude parameters giving a different MTD.

Fix positive-length dose cells $B_1,\ldots,B_J$, with boundaries
$c_0=x_{\min}<c_1<\cdots<c_J=x_{\max}$ and a fixed convention at each
boundary. Write $\vartheta=(u,v)$ for the cell means
$u=(u_1,\ldots,u_J)$ and residual variance $v>0$.
The parameter space $\Theta$ is $\mathbb R^J\times(0,\infty)$ for the
unconstrained iBART likelihood, or its subset with nondecreasing means
for mBART. The prior may have smaller support. Throughout, $\|\cdot\|$
is the Euclidean norm on these identifiable parameters.

For MTD inference, let $f(u)=u$ for mBART and $f(u)=\mathcal P u$
for iBART, where $\mathcal P$ is the fixed cell-length-weighted
isotonic projection. The cell risks and MTD functional are
\[
 \begin{aligned}
 r_j(\vartheta)&=\Phi\{(f_j(u)-\eta)/\sqrt v\},\\
 \xi(\vartheta)&=\sup\bigl(\{x_{\min}\}\cup
              \{x\in\mathcal X:r_{j(x)}(\vartheta)\le\theta\}\bigr),
 \end{aligned}
\]
where $j(x)$ identifies the dose cell and $\theta=1-\gamma$.
The MTD takes values among the boundaries $\{c_0,\ldots,c_J\}$,
with the endpoint conventions of Definition~\ref{def:clinical-mtd}.
Write $\Pi_n^\xi$ for its induced posterior distribution.

Fix a true parameter $\vartheta_0=(u_0,v_0)\in\Theta$ and its MTD
$\xi_0=\xi(\vartheta_0)$.
For iBART, assume additionally that $u_0$ is nondecreasing, so
$\mathcal P u_0=u_0$ and $\xi_0$ is the MTD of the data-generating
response law.
Suppose an indefinitely continued, predictable allocation satisfies
\[
 \begin{aligned}
 Y_i\mid\mathcal F_{i-1},X_i&\sim\mathcal N(u_{0j},v_0)
       &&\text{when }X_i\in B_j,\\
 \liminf_{n\to\infty}\frac{N_{j,n}}n&\ge\kappa>0
       &&(j=1,\ldots,J)
 \end{aligned}
\]
almost surely under $P_{\vartheta_0}$, where
$N_{j,n}=\sum_{i=1}^n\mathbf1\{X_i\in B_j\}$ and $\kappa$ is fixed.
The allocation rule is the same function of observed history under
all parameters. Let $\Pi$ be a fixed proper prior with
$\Pi\{\|\vartheta-\vartheta_0\|<\delta\}>0$ for every $\delta>0$.
Assume also that the true projected or monotone cell risks have a
strict margin from the clinical target:
\[
 \min_{1\le j\le J}|r_j(\vartheta_0)-\theta|>0.
\]

\begin{theorem}[Posterior consistency of the MTD]
\label{prop:doob-grid}
Under the model and assumptions stated above, the posterior distribution
of the MTD $\xi$ is consistent at its true value $\xi_0$: for every
$\epsilon>0$,
\[
 \Pi_n^\xi\{z:|z-\xi_0|\ge\epsilon\}
 \longrightarrow0
 \qquad P_{\vartheta_0}\text{-almost surely}.
\]
\end{theorem}

\noindent\emph{Proof.}
The weighted isotonic projection and Gaussian tail are continuous.
There are finitely many cells, and none of their true risks equals
$\theta$. Hence there is a radius $r>0$ such that
$\|\vartheta-\vartheta_0\|<r$ preserves every acceptable/unacceptable
cell classification and therefore gives $\xi(\vartheta)=\xi_0$.
Define the set of parameters giving the wrong MTD by
$A=\{\vartheta\in\Theta:\xi(\vartheta)\ne\xi_0\}$.
Every parameter in $A$ is outside this radius-$r$ neighborhood.
It suffices to show that its posterior probability tends to zero.
We compare the two integrals in
\begin{equation}
 \Pi_n(A)=
 \frac{\int_A L_n(\vartheta)\,\Pi(d\vartheta)}
      {\int_\Theta L_n(\vartheta)\,\Pi(d\vartheta)}.
 \label{eq:cell-posterior-ratio}
\end{equation}
The dose assigned to patient $i$ is determined by the preceding
history. Its allocation factor therefore has no unknown-parameter
term and cancels from the posterior ratio. Multiplying the conditional
Gaussian response densities gives
\[
 L_n(u,v)=(2\pi v)^{-n/2}
 \exp\left\{-\frac1{2v}
       \sum_{j=1}^J\sum_{i:X_i\in B_j}(Y_i-u_j)^2\right\}.
\]
This factorization uses the sequential conditional model; it does not
require the adaptive observations to be iid.

Every cell is visited infinitely often by the sampling assumption.
For all sufficiently large $n$, define the cell proportions, cell
means, and pooled residual variance by
\[
 p_{j,n}=\frac{N_{j,n}}n,\qquad
 \bar Y_{j,n}=\frac1{N_{j,n}}\sum_{i:X_i\in B_j}Y_i,\qquad
 \widehat v_n=\frac1n\sum_j\sum_{i:X_i\in B_j}
                         (Y_i-\bar Y_{j,n})^2.
\]
Let $\varepsilon_i=Y_i-u_{0j}$ when $X_i\in B_j$.
Predictability makes
$\mathbf1\{X_i\in B_j\}\varepsilon_i$ and
$\varepsilon_i^2-v_0$ martingale differences, with conditional
variances at most $v_0$ and $2v_0^2$, respectively. Dividing these
increments by $i$ gives summable variances. The martingale convergence
theorem and Kronecker's lemma therefore give
\[
 \frac1n\sum_{i=1}^n\mathbf1\{X_i\in B_j\}\varepsilon_i\to0,
 \qquad
 \frac1n\sum_{i=1}^n\varepsilon_i^2\to v_0
 \quad\text{almost surely}.
\]
Since $p_{j,n}\ge\kappa/2$ eventually, the first limit implies
$\bar Y_{j,n}\to u_{0j}$. Completing the square within each cell also
shows
\[
 \widehat v_n=
 \frac1n\sum_{i=1}^n\varepsilon_i^2
       -\sum_jp_{j,n}(\bar Y_{j,n}-u_{0j})^2
 \longrightarrow v_0>0.
\]
For the rest of the proof, work on the probability-one event where
these limits and the allocation assumption hold. No convergence of
the proportions themselves is needed.

Let $\widehat L_n=L_n(\bar Y_n,\widehat v_n)$, using the unrestricted
cell estimates even when $\bar Y_n$ is not nondecreasing.
Completing the square in the likelihood gives
\begin{equation*}
 \begin{aligned}
 \frac{L_n(u,v)}{\widehat L_n}
    &=\exp\{-nH_n(u,v)/2\},\\
 H_n(u,v)
    &=h\left(\frac v{\widehat v_n}\right)
       +\frac{\sum_jp_{j,n}(u_j-\bar Y_{j,n})^2}{v},
 \qquad h(t)=\log t+t^{-1}-1.
 \end{aligned}
\end{equation*}
The function $h$ is nonnegative, vanishes only at $t=1$, and tends
to infinity as $t$ tends to zero or infinity. Thus $H_n\ge0$ and,
eventually, $0<L_n(u,v)\le\widehat L_n<\infty$ for every parameter.
A proper prior then gives a finite, positive posterior denominator.
The common factor $\widehat L_n$ can now be canceled from
\eqref{eq:cell-posterior-ratio}.

A small value of $H_n$ forces both the variance and the cell means
to be close to their true values. To see this, suppose $H_n(u,v)\le c$.
Then $h(v/\widehat v_n)\le c$. The properties of $h$ imply that
$v/\widehat v_n$ is as close to one as desired when $c$ is small.
Because $\widehat v_n\to v_0$, we can choose $c>0$ small enough that,
eventually, this entails $v\le2v_0$ and $|v-v_0|<r/2$,
uniformly over all such parameters.

The remaining term of $H_n$, together with $p_{j,n}\ge\kappa/2$,
gives the explicit bound
\[
 \|u-\bar Y_n\|^2
 \le\frac{2v}{\kappa}H_n(u,v)
 \le\frac{4v_0c}{\kappa}.
\]
Choose $c$ smaller if necessary so the last expression is at most
$r^2/16$. Eventually
$\|\bar Y_n-u_0\|<r/4$, so the triangle inequality yields
$\|u-u_0\|<r/2$. Together with the variance bound, this implies
$\|\vartheta-\vartheta_0\|<r$.
Taking the contrapositive, there is a fixed $c_r>0$ such that
$H_n(\vartheta)\ge c_r$ for every
$\vartheta\in A$ and all sufficiently large $n$.
Hence the normalized numerator satisfies
\[
 \int_{A}e^{-nH_n(\vartheta)/2}\,\Pi(d\vartheta)
 \le e^{-nc_r/2}.
\]
The bound applies to every parameter giving the wrong MTD, including
arbitrarily large means and variances.

Choose a small fixed ball
$B=\{\vartheta\in\Theta:\|\vartheta-\vartheta_0\|<\delta\}$.
For small enough $\delta$, and eventually in $n$, all its variances
are bounded below by $v_0/2$ and its cell means are close to $\bar Y_n$.
More explicitly, since the proportions sum to one,
\[
 \frac{\sum_jp_{j,n}(u_j-\bar Y_{j,n})^2}{v}
 \le \frac{2}{v_0}\|u-\bar Y_n\|^2
 \quad(\vartheta\in B).
\]
This bound can be made uniformly smaller than $c_r/4$ by
choosing $\delta$ small and using $\bar Y_n\to u_0$.
Continuity of $h$ at one and $\widehat v_n\to v_0$ similarly make
the variance term smaller than $c_r/4$. Thus
$\sup_{\vartheta\in B}H_n(\vartheta)\le c_r/2$ eventually.
The support assumption gives $\Pi(B)>0$, so
\[
 \int_\Theta e^{-nH_n(\vartheta)/2}\,\Pi(d\vartheta)
 \ge \Pi(B)e^{-nc_r/4}.
\]
Combining the numerator and denominator bounds in the posterior ratio gives
\[
 \Pi_n(A)
 \le\Pi(B)^{-1}e^{-nc_r/4}\longrightarrow0.
\]
For every $\epsilon>0$, the event $|\xi-\xi_0|\ge\epsilon$ is contained
in the event $\xi\ne\xi_0$. Thus
\[
 \Pi_n^\xi\{z:|z-\xi_0|\ge\epsilon\}\le\Pi_n(A)\longrightarrow0.
\]
This holds on the probability-one event established above and proves
the stated MTD consistency.\hfill$\square$

\section{Posterior computation}
\label{app:bart-numerics}
The five procedures use different posterior calculations. The 1PLD
posterior is analytic. For 2PLD and 3PND, we use coordinate slice
sampling in transformed prior coordinates. For the tree models, we
use Bayesian backfitting: Metropolis--Hastings (MH) updates of tree
structure, followed by conditional draws of leaf values and residual
variance. Thus, \emph{Gibbs sampling} alone would not specify the tree
algorithm. This appendix describes posterior computation for the
$N=45$ study.

\subsection{Parametric procedures}
\paragraph{1PLD: analytic updating.}
Let $d_i=X_i-x_{\min}$. With the untruncated prior
$\beta\sim\mathcal N(a,b^2)$ and known residual SD $s$, conjugacy gives
$\beta\mid\mathcal F_n\sim\mathcal N(\mu_n,V_n)$, where
\[
 V_n=\left(b^{-2}+s^{-2}\sum_{i=1}^n d_i^2\right)^{-1},\qquad
 \mu_n=V_n\left(ab^{-2}+s^{-2}\sum_{i=1}^n d_iY_i\right).
\]
We use $(a,b,s)=(0.04,0.02,0.1)$.
For $A=\eta-sq_\gamma>0$, where $q_\gamma=\Phi^{-1}(\gamma)$,
the MTD is $x_{\max}$ if $\beta\leq A/L_x$ and
$x_{\min}+A/\beta$ otherwise, with $L_x=x_{\max}-x_{\min}$.
MTD quantiles, upper-endpoint mass and assigned-dose overdose
probabilities follow analytically, including the negative-slope
probability; see Eqs.~\eqref{eq:onepld-root}--\eqref{eq:onepld-posterior}.
No MCMC, burn-in or thinning is needed.

\paragraph{2PLD and 3PND: target in prior coordinates.}
The likelihood and normalized hierarchy are given in
Appendix~\ref{app:bayesian-comparators}. The coefficient bounds in
Eq.~\eqref{eq:comparator-prior} have width
$u(\sigma,\nu)-l(\sigma,\nu)=\sigma q_\gamma(1+L_x^{-\nu})$,
which depends on both $\sigma$ and $\nu$. To retain this normalizer
and the parameter-dependent support, put
$a_\sigma=\arctan(\eta/q_\gamma)$ and use
Eq.~\eqref{eq:comparator-latent-prior}:
\[
 \begin{gathered}
 U,V\sim\operatorname{Unif}(0,1),\quad Z\sim\mathcal N(0,1)
 \quad\text{independently},\\
 \sigma=\tan(a_\sigma V),\qquad \nu=\exp(\delta Z),\qquad
 \beta=l(\sigma,\nu)+U\{u(\sigma,\nu)-l(\sigma,\nu)\}.
 \end{gathered}
\]
For 3PND, $\delta=0.1$; for 2PLD, fix $\nu=1$ and omit $Z$.
The sampler updates $z_1=\operatorname{logit}(U)$,
$z_2=\operatorname{logit}(V)$ and, for 3PND, $z_3=Z$.
Including the transformation Jacobian gives the log target,
up to an additive constant,
\begin{equation*}
 \ell_n(z)=
 -n\log\sigma-\frac{1}{2\sigma^2}
       \sum_{i=1}^n\{Y_i-\beta d_i^\nu\}^2
 +\log\{U(1-U)V(1-V)\}-\frac{Z^2}{2}.
\end{equation*}
The last term is omitted for 2PLD. The physical prior normalizers
cancel with the corresponding change-of-variables factors; no
additional physical-coordinate Jacobian is applied.

\paragraph{Slice transitions and physical parameters.}
Each sweep visits $z_1,z_2,z_3$ in that order, omitting the third for
2PLD. For one coordinate, draw
$\log h=\ell_n(z)+\log U_0$, $U_0\sim\operatorname{Unif}(0,1)$.
Initialize a randomly positioned interval of width 1.5 around its
current value, then step out while its endpoints lie above this
height. A random integer $J_0$ uniform on $\{0,\ldots,39\}$ assigns
$J_0$ additional outward increments to the left and $39-J_0$ to the
right. Uniform proposals from the resulting interval are accepted
when their log target exceeds $\log h$; otherwise shrink the
interval toward the current point and propose again
\citep{neal2003slice}. The width is fixed, and failure to accept
within 1,000 shrinkage proposals is reported as an error.

These are slice-within-Gibbs moves in the \emph{transformed}
coordinates. Updating $z_2$ changes both $\sigma$ and $\beta$;
updating $z_3$ changes both $\nu$ and $\beta$. They do not reproduce
the physical-parameter transitions in Algorithms~3--5 of
\citet{lee2023}, whose Gibbs cycle itself includes slice moves.
Our specification concerns the normalized hierarchy stated here.
For example, a physical update of $\nu$ holding $(\beta,\sigma)$
fixed would retain
\begin{equation*}
 \pi(\nu\mid\beta,\sigma,\mathcal F_n)\ \propto\
 \exp\left[-\frac{\sum_i\{Y_i-\beta d_i^\nu\}^2}{2\sigma^2}\right]
 \frac{\operatorname{logN}(\nu;0,\delta^2)}
      {1+L_x^{-\nu}}\,
 \mathbf1\{l(\sigma,\nu)<\beta<u(\sigma,\nu)\}.
\end{equation*}
The displayed normalizer and support are not constants in $\nu$;
the prior-coordinate target retains both automatically.
The change-of-variables identity does not assert identical finite
chains or adaptive dose paths for different samplers.
Each retained draw gives an MTD via Eq.~\eqref{eq:comparator-root}.

\subsection{Tree backfitting and conditional updates}
\paragraph{One ensemble sweep.}
The likelihood mean is $b_0+\sum_{j=1}^m h_j(x)$, with $m=200$ and
fixed $b_0=2$. For tree $j$, remove its old contribution from the
current ensemble and form the partial residuals
\[
 R_{ij}=Y_i-b_0-\sum_{k\ne j}h_k(X_i),\qquad i=1,\ldots,n.
\]
Update its structure against these residuals, refresh its leaf
values, and restore its contribution. After visiting all trees,
update $\sigma^2$. The priors, response calibration and physical-dose
cuts in Appendix~\ref{sec:bart-priors} remain fixed as data accrue.

Both tree samplers use birth and death proposals, without
separate change-rule or swap moves. Birth selects an eligible leaf
and an available split threshold uniformly; death selects an
internal node with two leaf children uniformly. With one predictor,
the split variable is dose. Birth probability is one for a
splittable stump, zero if no leaf can split, and one half otherwise.
For a proposed birth $T\to T'$, let $p_d$ be the parent's prior split
probability, $p_L,p_R$ the children's split probabilities,
$B(T)$ the eligible leaves and $D(T')$ the removable sibling pairs.
The topology-prior and proposal factor is
\[
 A_T=
 \frac{p_d(1-p_L)(1-p_R)}{1-p_d}\,
 \frac{P_D(T')/|D(T')|}{P_B(T)/|B(T)|}.
\]
The uniform threshold prior and proposal factors cancel.
If $I_{\rm parent}$ and $I_{\rm children}$ are the corresponding
leaf-integrated likelihood-times-prior masses, accept with
probability $\min\{1,A_T I_{\rm children}/I_{\rm parent}\}$.
Death uses the reverse factor. For mBART, these masses also retain
the order constraints imposed by unchanged leaves.

\paragraph{iBART: Gaussian leaves and posterior projection.}
For iBART, the Gaussian leaf update uses unit observation weights,
giving the usual homoskedastic sufficient statistics.
For leaf $\ell$, let $n_\ell$ be its count and
$S_\ell=\sum_{i\in\ell}R_{ij}$. With Gaussian leaf prior variance
$\tau^2$, the conditional update is
\begin{equation*}
 V_\ell=(\tau^{-2}+n_\ell/\sigma^2)^{-1},\qquad
 \mu_\ell\mid\cdots\sim
 \mathcal N(V_\ell S_\ell/\sigma^2,V_\ell).
\end{equation*}
Put $b_\ell=n_\ell/\sigma^2$ and $M_\ell=S_\ell/\sigma^2$.
Dropping likelihood factors common to the compared partitions,
the integrated log factor for a topology move is
\[
 H_\ell=-\tfrac12\log(1+b_\ell\tau^2)
       +\frac{M_\ell^2\tau^2}{2(1+b_\ell\tau^2)}.
\]
Thus the birth likelihood ratio is
$\exp(H_L+H_R-H_{\rm parent})$.
Accepted new leaves receive conditional Gaussian draws, followed by
a refresh of all leaves. Empty children are allowed under the fixed
cut prior: for $n_\ell=S_\ell=0$, the integrated factor is one and
the draw is exactly $\mathcal N(0,\tau^2)$.

After sampling the unconstrained likelihood mean $g$, evaluate
each draw on every cell of the common dose-cut grid and apply
weighted PAVA, Eq.~\eqref{eq:bart-pava}, to obtain $f=\mathcal Pg$
using physical cell lengths as weights \citep{deleeuw2009}.
Retain the original paired $\sigma$ draw. Projection neither
replaces $g$ in the likelihood nor refits variance from projected
residuals.

\paragraph{Residual variance in both tree models.}
With residual sum of squares $\mathrm{SSE}=\sum_i\{Y_i-b_0-\sum_j h_j(X_i)\}^2$, draw
\begin{equation*}
 \sigma^2\mid\cdots\sim
 \operatorname{IG}\left(\frac{n+\nu_\sigma}{2},
              \frac{\nu_\sigma\lambda+\mathrm{SSE}}{2}\right),
 \qquad
 \sigma^2=\frac{\nu_\sigma\lambda+\mathrm{SSE}}
                  {U_\sigma},\quad U_\sigma\sim\chi^2_{n+\nu_\sigma}.
\end{equation*}
Here $\operatorname{IG}(a,b)$ is the inverse-gamma distribution with density proportional to
$v^{-a-1}\exp(-b/v)$, $\nu_\sigma=3$, and $\lambda$ is fixed as in
Appendix~\ref{sec:bart-priors}. The SSE uses $g$ for iBART and
the constrained likelihood mean $f$ for mBART.

\subsection{\textnormal{mBART}: ordered finite-grid leaf updates}
The numerical target uses $Q=50$ fixed leaf values
\[
 u_q=-3\tau+\frac{6\tau q}{Q+1},\qquad
 p_q=\frac{\phi(u_q/\tau)}
           {\sum_{r=1}^{Q}\phi(u_r/\tau)},\qquad q=1,\ldots,Q,
\]
with $\tau=\sqrt{1.467}/\sqrt{200}$. These masses are proportional
to Gaussian density values, not Gaussian-bin probabilities; the
endpoints $\pm3\tau$ are excluded.
Other leaves determine the allowable interval $[L_\ell,U_\ell]$:
the current leaf lies above all lower-dose leaf values and below
all higher-dose leaf values, with ties allowed.
Its finite-grid conditional masses are
\[
 w_{\ell q}
 =p_q\exp\left\{\frac{u_qS_\ell-\tfrac12n_\ell u_q^2}{\sigma^2}\right\}
       \mathbf1\{L_\ell\leq u_q\leq U_\ell\}.
\]
Normalize these masses and draw categorically, visiting leaves
sequentially and recomputing the constraints after each draw.
This is a Gibbs update for the finite-grid model, not a continuous
truncated-normal draw. An empty leaf has constant likelihood and
uses the restricted prior masses.

For a birth, the single-leaf mass and ordered pair mass are
\[
 I_1=\sum_q w_{{\rm parent},q},\qquad
 I_2=\sum_{q\leq r}w_{Lq}w_{Rr},
\]
where each child's weights include constraints from unchanged
leaves. The MH ratio uses $I_2/I_1$. An accepted pair is drawn jointly
from the normalized pair weights; death uses $I_1/I_2$ and draws
the merged leaf from its normalized weights.
The restricted prior masses remain in $I_1,I_2$: there is no
division by a separate order-cone probability for each tree.
Normalizing conditional draw probabilities is a different operation.
This implements Eq.~\eqref{eq:bart-main-working-prior}'s
constrained-product formulation on the finite grid.

The pair calculation is linear in $Q$. For numerically scaled
$a_q=w_{Lq}$ and $b_r=w_{Rr}$, the suffix recurrence is
\[
 B_{Q+1}=0,\qquad B_q=b_q+B_{q+1},\qquad
 I_2=\sum_{q=1}^{Q}a_qB_q.
\]
Choose $q$ with probability proportional to $a_qB_q$, then
$r\geq q$ with probability proportional to $b_r$. The implementation
uses the residual of the same uniform variate to preserve the
lexicographic pair distribution.
Scaled likelihood weights and a log-sum-exp fallback avoid
underflow, with scale factors restored in the MH ratio.
For equally spaced leaf points, successive quadratic-likelihood
ratios change by $\exp(-n_\ell\Delta_u^2/\sigma^2)$, where
$\Delta_u$ is the spacing. These recurrences preserve the fixed
numerical target; they do not refine its support.

\subsection{Refitting, retention and inversion}
At every completed cohort, including the final one, each method
uses all accrued observations. The sampled methods start two
fresh, separately seeded chains; no preceding chain is used as a
warm start and no retained sweeps are thinned
(Table~\ref{tab:posterior-computation-settings}).
For 2PLD and 3PND, each chain starts from the stated prior.
iBART starts from zero-valued stump contributions with offset 2
and residual SD 0.2. mBART starts from stump values summing to the
centered sample mean and the sample residual SD, with positive
fallback $\sqrt{\lambda}$ for a degenerate initial sample.
These starting values do not alter the fixed priors.

\begin{table}[htbp]
\centering
\caption{Full posterior computation at each cohort in the $N=45$ study.
Burn-in and retained counts are complete sweeps per chain. Retained
draws from both chains are pooled, without thinning.}
\label{tab:posterior-computation-settings}
\small
\small
\setlength{\tabcolsep}{3pt}
\renewcommand{\arraystretch}{1.0}
\begin{tabular*}{\textwidth}{@{}@{\extracolsep{\fill}}lrrrl@{}}
\toprule
Method & Chains & Burn-in & Retained & Main transition\\
\midrule
1PLD & 0 & 0 & 0 & Analytic Gaussian posterior\\
2PLD & 2 & 4,000 & 64,000 & Two-coordinate slice sweep\\
3PND & 2 & 4,000 & 64,000 & Three-coordinate slice sweep\\
iBART & 2 & 1,000 & 16,000 & Tree MH, Gaussian leaves, variance\\
mBART & 2 & 1,000 & 4,000 & Tree MH, ordered grid leaves, variance\\
\bottomrule
\end{tabular*}
\end{table}

For a retained tree draw, calculate
$p_k=\Phi\{(f_k-\eta)/\sigma\}$ on the 101 dose cells and count the
acceptable prefix with $p_k\leq\theta=1-\gamma$.
No acceptable cell gives $\xi=x_{\min}$; all acceptable cells give
$\xi=x_{\max}$; otherwise $\xi$ is the cut at the prefix's right
boundary. This implements the supremum in Eq.~\eqref{eq:bart-root},
including jumps and endpoint atoms, without imposing risk equality
at $\xi$. MTD quantiles use lower empirical quantiles (type~1).
Assigned-dose risk is recomputed after the interior adjustment and
dose restrictions, with a dose equal to a cut entering the cell
on its right. The last administered dose remains the reported MTD
estimate; the final posterior is an additional inferential summary.

\section{Bayesian planning under dose restrictions and count-based stopping}
\label{app:rl-formulation}
This appendix develops the planning formulation in
Section~\ref{sec:rl-framework} for the protocol used in this paper.
The unknown response model is static; doses affect both patient outcomes
and what can be learned before subsequent assignments. Joint posterior
states are a standard device for Bayesian planning
\citep{ross2011bayesian,guez2012bayes}, and dynamic programming for
phase I trials has a precedent in \citet{bartroff2010dp}.
Here the two dose controls make the last administered dose consequential:
Only Escalation prevents reversing an upward assignment, and the cap
limits the doses reachable within the remaining cohorts. The results
below specialize the decision problem to those controls, the cumulative
DLT stopping rule, and a terminal estimate equal to the actual last dose.
We first derive exact-policy guarantees and then evaluate finite-$B$ rollout
under the conjugate 1PLD model using a fixed linear response curve
(Section~\ref{sec:rl-seven-curve}).
The main simulation study compares the five EWOC procedures.

\subsection{From the agent--environment loop to dose finding}
\label{app:rl-primer}
The \emph{agent--environment loop} describes sequential decisions:
an agent uses its state and policy to choose an action, receives
feedback from the environment, and updates its knowledge for the next
decision. Reward defines the objective; return accumulates rewards
over an episode \citep{ross2011bayesian}.
Figure~\ref{fig:rl-dose-correspondence} maps this loop to dose finding.

\begin{figure}[!htbp]
\centering
\includegraphics[width=\textwidth]{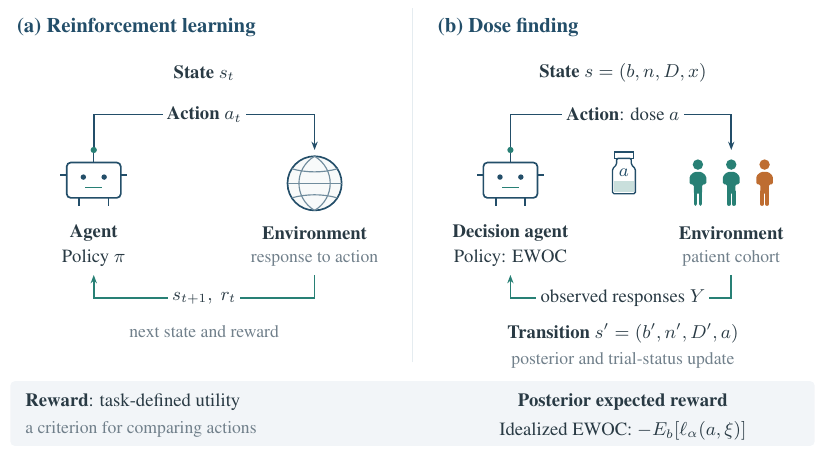}
\caption{\textbf{The RL loop and its dose-finding counterpart.}
An action changes patient exposure and supplies observations for the
next decision. Cohort responses update the posterior and trial state;
posterior expected reward evaluates dose choices. The reward's
MTD-dependent component is latent.}
\label{fig:rl-dose-correspondence}
\end{figure}

\paragraph{Agent, environment, state, and action.}
The \emph{agent} is the dose-selection procedure; the \emph{environment}
is the patient response process with a fixed unknown response model.
After a completed cohort, the \emph{state} $s=(b,n,D,x)$ records the
joint parameter posterior, enrollment, cumulative DLT count, and last
dose. The \emph{action} $a$ is the next cohort's dose. The implemented
\emph{policy} maps this state to a dose using EWOC and the enabled
controls. The joint posterior supports MTD inference and predicts
responses to candidate doses.

\paragraph{Feedback, transition, and episode.}
At dose $a$, the full cohort supplies observed responses $\mathbf Y$.
These update the posterior, enrollment, and DLT count; the new last
dose is $a$. This \emph{transition} links patient exposure to subsequent
learning. The trial \emph{episode} ends when the completed-cohort
stopping check gives $D>K_N$ or $n=N$.

\paragraph{Reward and the EWOC connection.}
The idealized EWOC proposal uses asymmetric dose loss and posterior
expected reward
\[
 \ell_\alpha(a,\xi)=\alpha(\xi-a)_+ +(1-\alpha)(a-\xi)_+,
 \qquad \bar r_b(a)=-\E_b[\ell_\alpha(a,\xi)].
\]
The unadjusted posterior $\alpha$-quantile maximizes this one-step
expected reward; smaller $\alpha$ weights overdosing more heavily
\citep{bartroff2011}. The implemented policy subsequently applies its
numerical adjustment and dose controls. The observed response
$\mathbf Y$ supplies learning feedback, not this reward: the true MTD
$\xi$ is unknown, so dose quality is assessed through its posterior.

\paragraph{Learning and planning roles.}
Value-based RL uses expected future return to compare actions, while
policy-based RL optimizes an action-selection rule; these approaches
can be combined. Here, Bayesian updating learns the response model
within a model-based belief-state formulation \citep{ross2011bayesian}.
EWOC is a prescribed posterior-dependent policy, with no policy-gradient
training. The proposed rollout uses a different objective from the EWOC
one-step loss: it compares expected remaining DLT and above-MTD
counts plus terminal last-dose error. This connects the formulation
to action-value planning
\citep{bartroff2010dp}. The main five-model experiments evaluate EWOC.

\subsection{Joint posterior state and the cohort transition}
Fix the model, prior, dose interval, and protocol. Let $\vartheta$ denote
the static model parameter, $\mu_\vartheta$ its likelihood mean, and
$f_\vartheta$ the curve used to define the MTD. For iBART these are
$\vartheta=(g,\sigma)$, $\mu_\vartheta=g$, and
$f_\vartheta=\mathcal P g$; for mBART they are
$\vartheta=(f,\sigma)$ and $\mu_\vartheta=f_\vartheta=f$.
The parametric models use their stated mean, risk, and endpoint
conventions. Write $\xi(\vartheta)=\xi(f_\vartheta,\sigma)$ and
$p_\vartheta(a)=\Phi\{(f_\vartheta(a)-\eta)/\sigma\}$.
The physical bounds satisfy $0<x_{\min}<x_{\max}<\infty$, so
$\xi(\vartheta)\in[x_{\min},x_{\max}]$ is positive even at an endpoint.

After a completed cohort the state is $s=(b,n,D,x)$, where
$b=\Pi(d\vartheta\mid\mathcal F_n)$, $n$ is accrued enrollment, $D$
is the cumulative DLT count, and $x$ is the last administered dose.
At dose $a$, the next cohort contains $m=\min(c,N-n)$ patients.
For $\mathbf y=(y_1,\ldots,y_m)$ and
$\phi_\sigma(v)=\sigma^{-1}\phi(v/\sigma)$, its likelihood,
predictive density, and posterior are
\begin{align}
 L_a(\vartheta;\mathbf y)&=\prod_{j=1}^m
       \phi_\sigma\{y_j-\mu_\vartheta(a)\},&
 q_b(\mathbf y\mid a)&=\int L_a(\vartheta;\mathbf y)b(d\vartheta),
 \nonumber\\
 b^{a,\mathbf y}(d\vartheta)&=
 \frac{L_a(\vartheta;\mathbf y)b(d\vartheta)}{q_b(\mathbf y\mid a)},&
 T(s,a,\mathbf y)&=(b^{a,\mathbf y},n+m,D+d(\mathbf y),a),
 \label{eq:rl-app-update}
\end{align}
where $d(\mathbf y)=\sum_j\mathbf1\{y_j\geq\eta\}$.
The same parameter generates the entire cohort. The integral therefore
encloses the product; a product of marginal predictive densities would
instead redraw the latent parameter for each patient. The raw $g$
appears in the iBART likelihood and update, even though its projected
curve defines the target and feasibility event.

\begin{proposition}[Sufficiency of the joint belief state]
\label{prop:rl-markov}
Suppose the parameter space is standard Borel, the likelihood and
parameter functionals are measurable, and within each cohort responses
are conditionally independent given the preceding history, assigned
dose, and static parameter, with the Gaussian densities above. Suppose
policy randomization is independent of that parameter conditional on
the observed history. Then $s=(b,n,D,x)$ is a controlled Markov state
for continuation decisions under the specified protocol, with transition
\eqref{eq:rl-app-update}. A state is absorbing when $D>K_N$ or $n=N$.
\end{proposition}
\begin{proof}
Conditional on the observed history and chosen action, the parameter
law is still $b$: prescribed or independently randomized actions add
no parameter information. Conditional independence gives $L_a$, and
integration gives $q_b$. Bayes' formula gives $b^{a,\mathbf y}$ wherever
the denominator is positive and finite, which holds predictive-almost
surely; choose an arbitrary measurable version on the remaining null
set. Probability measures on a standard Borel space admit the required
regular conditional distributions, and this dominated likelihood gives
a measurable update. The cohort size, count increment, and new last dose
depend only on $(s,a,\mathbf y)$. The full cohort, including its posterior
update, precedes the terminal check. Hence no further history is needed
for the future transition or permitted actions.
\end{proof}

\paragraph{Why the MTD marginal is insufficient.}
For an interior candidate cut $z$, fixed $\sigma>0$, and distinct
$v_1,v_2>0$, consider the nondecreasing step means
\[
 f_k(a)=\eta-\sigma\Phi^{-1}(\gamma)
           +v_k\{2\mathbf1(a\geq z)-1\},\qquad k=1,2.
\]
Both have MTD $z$: doses below the cut are acceptable and those at or
above it are not. Degenerate beliefs at these two parameters have the
same MTD marginal, but different predictive response distributions.
Their DLT probabilities also differ. This counterexample to a general
state reduction applies to the tree construction as well as to monotone
Gaussian mean models. Keeping only the MTD posterior discards information
needed to simulate future observations and update the model.

\subsection{Two independent dose controls and the EWOC baseline}
Let $\ell_0=x_{\min}+\delta_{\mathrm{num}}$ and
$u_0=x_{\max}-\delta_{\mathrm{num}}$. At a post-initial, nonterminal
state define
\[
 \ell(s)=\begin{cases}x,&\text{Only Escalation on},\\
                       \ell_0,&\text{Only Escalation off},\end{cases}
 \qquad
 u(s)=\begin{cases}\min(x+\Delta_{\max},u_0),&\text{cap on},\\
                       u_0,&\text{cap off},\end{cases}
\]
Set $I(s)=[\ell(s),u(s)]$.
Every reached last dose is in $[\ell_0,u_0]$, so $I(s)$ is nonempty.
Define three different posterior quantities:
\begin{equation*}
 \begin{split}
 \rho_b(a)&=\int\Phi\{(\mu_\vartheta(a)-\eta)/\sigma\}b(d\vartheta),\\
 r_b(a)&=b\{p_\vartheta(a)>1-\gamma\},\qquad
 o_b(a)=b\{\xi(\vartheta)<a\}.
 \end{split}
\end{equation*}
These are, respectively, predictive DLT probability, posterior probability
of exceeding the toxicity-risk target, and posterior probability of
assigning above the MTD. In particular $r_b$ is not $\rho_b$; at tree
jumps it need not equal $o_b$ either. For iBART, $r_b$ uses the projected
curve whereas $\rho_b$ uses the raw likelihood mean.

The nominal set is $F(s)=\{a\in I(s):r_b(a)\leq\alpha\}$.
Use $F(s)$ if nonempty and the singleton $\{\ell(s)\}$ otherwise.
Thus the default fallback holds $x$; disabling Only Escalation instead
permits the lower floor. The fallback's actual risk is recorded as an
exception, not converted into a stopping condition. Initialization is
the separate prescribed action $x_1=6$. Termination occurs only after
an observed cohort when $D>K_N$ or $n=N$, regardless of the dose options.

For the exact posterior lower $\alpha$-quantile $q_\alpha(b)$, write the
complete EWOC baseline as
\begin{equation*}
 \widetilde a(b)=\max\{\ell_0,\min(x_{\max},q_\alpha(b))-
                    \delta_{\mathrm{num}}\},\qquad
 \pi_0(s)=\max\{\ell(s),\min(\widetilde a(b),u(s))\}.
\end{equation*}
This composition applies the inward margin, floor, cap, and hold in that
order. The admissibility result below is stated for the exact posterior
quantile.

For nondecreasing curves, excessive risk at $a$ implies
$\xi(\vartheta)\leq a$, including an unacceptable boundary point of a
step. Consequently, for $a<q_\alpha(b)$,
\[
 r_b(a)\leq b\{\xi\leq a\}<\alpha.
\]
The same implication holds under the stated 1PLD endpoint convention:
its negative slopes are all acceptable under the fixed calibration.
The inward proposal is below $q_\alpha$ unless the floor binds, and a
cap only lowers it. If $\pi_0(s)$ fails the nominal bound, it must
therefore equal $\ell(s)$. Monotonicity of the exceedance event then
makes all of $I(s)$ infeasible, so this is exactly the required fallback.
Otherwise it belongs to $F(s)$. Thus $\pi_0(s)\in\mathcal A(s)$ for all
four option combinations, with the lower-endpoint fallback when
$F(s)=\varnothing$.

\subsection{Last-dose loss and protocol-imposed reachability}
Fix nonnegative planning weights $w_D,w_O,w_R$ before comparing policies.
For a fixed parameter $\vartheta$ and current state $s$, define the
continuation risk of policy $\pi$ by
\begin{equation}
 \mathcal R^\pi_\vartheta(s)=
 \mathbb E^\pi_\vartheta\!\left[
 w_D\sum_{i=n+1}^{n_{\mathrm{end}}}\mathbf1\{Y_i\geq\eta\}
 +w_O\sum_{i=n+1}^{n_{\mathrm{end}}}\mathbf1\{X_i>\xi(\vartheta)\}
 +w_R\frac{|X_{n_{\mathrm{end}}}-\xi(\vartheta)|}{\xi(\vartheta)}
 \right].
 \label{eq:rl-app-fixed-risk}
\end{equation}
The policy continues to update its belief and obey the count rule while
this expectation holds the generating parameter fixed. The Bayes
continuation value is $V^\pi(s)=\int\mathcal R^\pi_\vartheta(s)b(d\vartheta)$.
At initialization the sums are the full NPD and NPO. At later states,
past loss contributions are fixed additive constants and may be omitted
from the continuation value without changing future decisions.
Signed BTM remains a reporting criterion for the direction of error;
the scalar planning loss uses its relative absolute magnitude.

Let
\begin{equation*}
 c_b(s,a)=m\{w_D\rho_b(a)+w_Oo_b(a)\},\qquad
 R(b,x)=w_R\int\frac{|x-\xi(\vartheta)|}{\xi(\vartheta)}b(d\vartheta).
\end{equation*}
Let $s_t=(b_t,n_t,D_t,x_t)$, with $s_0=s$, and let $\tau$ count the
remaining cohorts until termination. Averaging
\eqref{eq:rl-app-fixed-risk} over the current posterior and applying
iterated expectation gives
\[
 V^\pi(s)=\mathbb E_b^\pi\!\left[
 \sum_{t=0}^{\tau-1}c_{b_t}(s_t,a_t)+R(b_\tau,x_\tau)\right].
\]
Thus the count terms accrue to the realized stopping time, and the
terminal term scores last-dose accuracy; the objective assigns no
separate value to enrollment itself. The terminal posterior includes
the final responses, but $x_\tau$ remains the dose selected beforehand,
with no replacement estimate or terminal dose optimization.

\begin{proof}[Proof of Proposition~\ref{prop:rl-reachability}]
Write $H=H(s)$, equal to zero at terminal states and to
$\lceil(N-n)/c\rceil$ otherwise, and
$u_H=\min\{x+H\Delta_{\max},u_0\}$.
There are at most $H$ remaining cohort assignments. Each is at least
the preceding dose and at most that dose plus $\Delta_{\max}$, with
upper limit $u_0$. Induction gives the asserted interval. Early DLT
stopping can only reduce the number of changes. For every fixed
$\vartheta$, distance from $\xi(\vartheta)$ to this interval is
$(x-\xi(\vartheta))_++(\xi(\vartheta)-u_H)_+$. Divide the pathwise
distance bound by the positive $\xi(\vartheta)$ and integrate over the
fixed initial belief and future outcomes. Later adaptive posterior
updates do not change that joint-expectation argument.
\end{proof}

The bound isolates irreversible overshoot and a target too far above
the current dose to reach in time. Feasibility and count-based stopping
can further restrict a path. For the default
$N=45$, $c=3$, and $\Delta_{\max}=3.5$, the first completed cohort leaves
$n=3$, $x=6$, and $H=14$, giving $u_H=55$. Cohort $j$ cannot exceed
$6+(j-1)3.5$; the first 39 patients cannot exceed 48. For the evaluated
truths with MTD 50, NPO is therefore at most six irrespective of the
model. This bound follows directly from the dose schedule. Disabling
Only Escalation removes the irreversible lower bound $x$; disabling
the cap removes its incremental upper bound.

\subsection{A finite-candidate Bellman problem}
For existence and implementation, take a fixed finite list of measurable
dose functions $a_1(s),\ldots,a_M(s)$ in $I(s)$, including both
$\ell(s)$ and the complete baseline action $\pi_0(s)$. At nonterminal
post-initial states retain candidates with $r_b(a)\leq\alpha$; if none
pass, retain only $\ell(s)$. Denote the resulting nonempty set by
$C(s)$. Because the lower endpoint is included and $r_b$ is
nondecreasing, this fallback occurs exactly when $F(s)$ is empty.
The baseline proof above gives $\pi_0(s)\in C(s)$ everywhere.
For the initial cohort the permitted set is simply $\{x_1\}$, without
nominal feasibility certification, and $\pi_0(b_0,0,0,x_1)=x_1$.
Its last-dose coordinate only becomes an administered dose after that
forced cohort.

\begin{proposition}[Bellman recursion and an optimal finite-candidate policy]
\label{prop:rl-bellman}
Under Proposition~\ref{prop:rl-markov}'s measurability assumptions and
the finite candidate construction, a measurable deterministic Markov
policy attains the optimal Bayes continuation value $V_C$. It obeys
\begin{equation*}
 \begin{gathered}
 V_C(s)=R(b,x),\qquad D>K_N\ \text{or}\ n=N,\\
 V_C(s)=\min_{a\in C(s)}Q_C(s,a),\qquad\text{otherwise},\\
 Q_C(s,a)=c_b(s,a)+\int V_C\{T(s,a,\mathbf y)\}
                           q_b(\mathbf y\mid a)d\mathbf y.
 \end{gathered}
\end{equation*}
The same statement uses the singleton prescribed action at initialization.
\end{proposition}
\begin{proof}
All continuation costs are nonnegative and bounded by
$(w_D+w_O)(N-n)+w_R(x_{\max}-x_{\min})/x_{\min}$.
The forced terminal value is measurable. Induct on $N-n$; each action
observes at least one patient, so successor values are already defined.
Integration against the measurable transition kernel preserves
measurability. A finite minimum with measurable eligibility indicators
is measurable; choose the first minimizer in the fixed enumeration.
Conditional expectation bounds the value of any randomized or
history-dependent continuation below by this minimum. The selected
action followed by the induction policies attains it. This proves both
the recursion and optimality. Count-based terminal states encountered
before $N$ use the same forced terminal value.
\end{proof}

The finite formulation admits the measurable optimal policy in
Proposition~\ref{prop:rl-bellman}; continuous-action variants follow
under suitable measurable-selection conditions.

\subsection{Exact rollout improvement}
Let $V^{\pi_0}$ be the continuation cost of the complete EWOC baseline
under the same candidate and stopping conventions. Define
\[
 Q^{\pi_0}(s,a)=c_b(s,a)+\int V^{\pi_0}\{T(s,a,\mathbf y)\}
                                     q_b(\mathbf y\mid a)d\mathbf y.
\]
This evaluates one candidate cohort followed by baseline continuation
to termination, rather than only the next response.

\begin{proof}[Proof of Proposition~\ref{prop:rl-rollout}]
The values agree at terminal states. Assume the inequality holds for
states with fewer remaining patients. At $a=\pi_+(s)$ it gives
\[
 \begin{aligned}
 V^{\pi_+}(s)
 &\leq c_b(s,a)+\int V^{\pi_0}\{T(s,a,\mathbf y)\}q_b(\mathbf y\mid a)d\mathbf y\\
 &=\min_{v\in C(s)}Q^{\pi_0}(s,v)
 \leq Q^{\pi_0}\{s,\pi_0(s)\}=V^{\pi_0}(s).
 \end{aligned}
\]
The second inequality uses inclusion of the complete baseline action,
including its forced fallback. Induction finishes the proof.
\end{proof}

The next section implements finite-$B$ rollout under the same transition
model, candidate sets, stopping rule, and weighted loss.

\subsection{Monte Carlo rollout specification and approximation error}
\label{sec:rl-mc-rollout}
Algorithm~\ref{alg:rl-app-rollout} gives one repeated-rollout policy
for all five response models. Model-specific routines supply joint
posterior draws, likelihood predictions, posterior updates, and MTD
inversion; $T_{\mathcal M}$ is the whole-cohort update in
\eqref{eq:rl-app-update} for the chosen model. The algorithm starts after
the prescribed initial cohort. It samples one latent model per simulated
continuation and retains it throughout that path. The latent draw
generates responses and scores the path; subsequent EWOC decisions use
only the updated posterior and observed trial state. Each hypothetical
update conditions on all actual and earlier simulated observations.
For iBART, the raw
$g$ supplies the likelihood and update, while $\mathcal P g$ with the
same $\sigma$ determines the MTD and feasibility. For mBART, the
constrained $f$ supplies both. The candidate list below uses the two
enabled dose controls of the numerical study. Posterior updates and EWOC
quantiles are analytic in the 1PLD evaluation in
Section~\ref{sec:rl-seven-curve};
the other models require numerical posterior procedures.

\begin{algorithm}[H]
\caption{Repeated rollout with EWOC as the base policy}
\label{alg:rl-app-rollout}
\label{alg:rl-all-models}
\label{alg:rl-1pld-implemented}
\small
\begin{algorithmic}[1]
\Require Model $\mathcal M\in\{\mathrm{1PLD},\mathrm{2PLD},\mathrm{3PND},\mathrm{iBART},\mathrm{mBART}\}$; joint posterior state $s=(b,n,D,x)$ after the prescribed initial cohort.
\Require Shared protocol and EWOC baseline $\pi_0$, fixed weights $(w_D,w_O,w_R)$, and $B$ simulated continuations per candidate.
\While{$n<N$ and $D\leq K_N$}
 \State Set $a_0=\pi_0(s)$, including inward adjustment, dose bounds, upward cap, and Only Escalation.
 \State Form $(a_0,x,x+\tfrac14(a_0-x),x+\tfrac12(a_0-x),x+\tfrac34(a_0-x))$ and remove duplicates.
 \State Retain candidates with $r_b(a)\leq\alpha$; if none pass, retain the permitted hold $x$. Call this ordered list $C(s)$.
 \If{$|C(s)|=1$}
  \State Set $a$ to its sole entry.
 \Else
  \State Draw independent joint parameters $\vartheta_r\sim b$, retaining their paired $\sigma_r$, and innovations $e_{r,n+1:N}\sim\mathcal N(0,I)$, for $r=1,\ldots,B$; share them across candidates.
  \For{each $a\in C(s)$}
   \For{$r=1,\ldots,B$}
    \State Copy $s_r=s$; set $a_r=a$ and $Z_r=0$.
    \While{$n_r<N$ and $D_r\leq K_N$}
     \State Set $m_r=\min(c,N-n_r)$; generate $y_{r,j}=\mu_{\vartheta_r}(a_r)+\sigma_r e_{r,n_r+j}$ for $j=1,\ldots,m_r$.
     \State Add $w_Dd(\mathbf y_r)+w_Om_r\mathbf1\{a_r>\xi(\vartheta_r)\}$ to $Z_r$.
     \State Update $s_r=T_{\mathcal M}(s_r,a_r,\mathbf y_r)$ using the complete simulated cohort.
     \State If $s_r$ is nonterminal, set $a_r=\pi_0(s_r)$ using only its updated posterior and trial state.
    \EndWhile
    \State Add $w_R|x_r-\xi(\vartheta_r)|/\xi(\vartheta_r)$ to $Z_r$.
   \EndFor
   \State Set $\widehat Q(s,a)=B^{-1}\sum_{r=1}^B Z_r$.
  \EndFor
  \State Set $a$ to the first minimizer of $\widehat Q(s,a)$ in $C(s)$.
 \EndIf
 \State Administer $\min(c,N-n)$ patients at $a$, observe the complete actual cohort, and update $s=T_{\mathcal M}(s,a,\mathbf y)$.
\EndWhile
\Ensure Terminal estimate $x$, the final administered dose.
\end{algorithmic}
\end{algorithm}

\begin{figure}[!p]
\centering
\includegraphics[width=\textwidth]{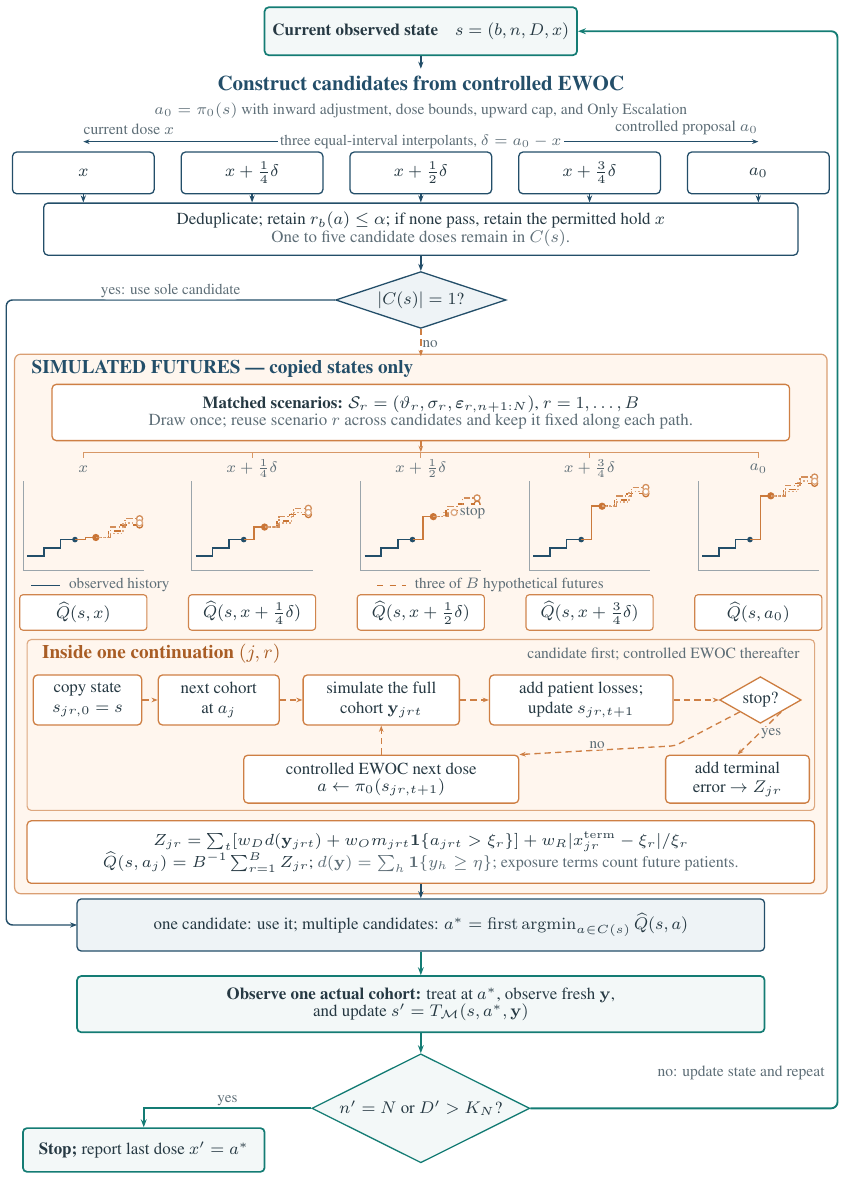}
\caption{\textbf{Repeated EWOC-based rollout at one observed trial state.}
Controlled EWOC generates up to five feasible candidates. A sole candidate is
used directly; otherwise, $B$ matched posterior-predictive continuations
estimate candidate losses, using each candidate first and controlled EWOC
thereafter. The minimum-loss candidate treats the next cohort, the state is
updated, and the comparison repeats until stopping.}
\label{fig:rl-repeated-rollout}
\end{figure}

With exact posterior primitives and independent continuations,
$\widehat Q(s,a)$ is unbiased for $Q^{\pi_0}(s,a)$ at a fixed state and
candidate. Selecting the smallest Monte Carlo estimate introduces
selection error, motivating a conditional bound under a uniform
approximation assumption.

\begin{proposition}[Rollout under uniform approximation error]
\label{prop:rl-approx-rollout}
Let $\mathcal S_{\mathrm{nt}}$ denote the nonterminal states in the
continuation problem. Suppose a fixed measurable approximation, using
the same legal candidate sets, satisfies
\begin{equation*}
 \sup_{s\in\mathcal S_{\mathrm{nt}}}\ \max_{a\in C(s)}
       |\widehat Q(s,a)-Q^{\pi_0}(s,a)|\leq\varepsilon
\end{equation*}
Its first-minimizing policy $\widehat\pi$ satisfies
\[
 V^{\widehat\pi}(s)\leq V^{\pi_0}(s)+2H(s)\varepsilon.
\]
Here $H(s)$ is zero at terminal states and
$\lceil(N-n)/c\rceil$ otherwise, the maximum remaining cohort count.
\end{proposition}
\begin{proof}
The baseline candidate and the two approximation errors give
$Q^{\pi_0}(s,\widehat\pi(s))\leq V^{\pi_0}(s)+2\varepsilon$.
Induct on $H(s)$, with zero error at compulsory terminal states.
The successor continuation contributes at most
$2\{H(s)-1\}\varepsilon$, and the current candidate comparison
contributes at most $2\varepsilon$. Early stopping only reduces these
remaining stages, giving the stated bound.
\end{proof}
The bound requires uniform approximation error over the continuation
domain with the same legal candidate sets. {Neither a finite number of continuations nor MCMC diagnostics alone verifies this uniform condition. The result is a conditional approximation bound, not a finite-budget guarantee for the numerical planner.} The numerical comparison
below reports finite-Monte-Carlo rollout performance.

\FloatBarrier
\section{Rollout under a linear 1PLD model}
\label{sec:rl-seven-curve}

We use the fixed linear scenario in Table~\ref{tab:n45-truth-definitions},
for which the 1PLD working model is correctly specified and analytic
posterior updates isolate rollout planning from numerical posterior
approximation. The mean response
has true MTD 50 and Gaussian noise with standard deviation 0.1.
All policies fit the conjugate 1PLD working model
$Y_i=\beta(X_i-5)+\epsilon_i$, with
$\epsilon_i\sim\mathcal N(0,0.1^2)$ and the untruncated prior
$\beta\sim\mathcal N(0.04,0.02^2)$.

The study uses 1,000 paired trial replications. For each loss profile,
rollout is evaluated with $B=1{,}000$ and $B=2{,}000$ simulated
continuations per candidate, with the latter assessing sensitivity to
the planning budget.
All policies use initial dose 6, dose bounds 5 and 80, inward adjustment
$7.5\times10^{-7}$, feasibility level $\alpha=0.05$, target probability
$\gamma=0.9$, and toxicity threshold $\eta=3$. We compare upward caps
$\Delta_{\max}\in\{3.5,7,10.5\}$, with Only Escalation enabled. With $N=45$
and cohorts of three patients,
the whole cohort is observed and the posterior updated before stopping
at cumulative DLT count above four or enrollment 45.
The terminal estimate is the final administered dose. Write
$d=\mathrm{NPD}/N$, $o=\mathrm{NPO}/N$, and $e=\mathrm{RAE}$, where
planned $N$ remains the denominator if a trial stops early. We use the
equal-component objective $L_B=d+o+e$, corresponding to
$(w_D,w_O,w_R)=(1/N,1/N,1)$, and two sensitivity profiles:
\[
 L_S=2d+2o+e,
 \qquad
 L_P=d+o+2e.
\]
Thus $L_S$ emphasizes safety, whereas $L_P$ increases the relative
weight on MTD precision. Each weighted loss is compared only within
its own profile.

Rollout follows Algorithm~\ref{alg:rl-app-rollout} and is recomputed after
each observed cohort. The same patient-indexed innovations pair all
policies and all cap settings within each replication, while planning
draws are generated independently of the observed outcomes. The policy
calculation uses only the current posterior and protocol state.

Table~\ref{tab:rl-seven-curve} reports the mean and Monte Carlo standard
error (sample standard deviation divided by $\sqrt{1{,}000}$) for each
policy. Weighted loss is the objective in
Proposition~\ref{prop:rl-rollout}; the component metrics show how its
weights affect safety and precision.

\paragraph{Results and sensitivity analyses.}
Across all three cap settings, rollout attained lower observed mean
weighted loss than EWOC under every loss profile and planning budget.
Reductions relative to EWOC ranged from 4.1\% to 25.9\% with
$\Delta_{\max}=3.5$, from 55.8\% to 80.9\% with $\Delta_{\max}=7$,
and from 62.6\% to 84.3\% with $\Delta_{\max}=10.5$. Relative to balanced rollout, safety emphasis reduced NPD and NPO at caps 3.5 and 7.
Precision emphasis reduced RAE at all three caps. Results changed little when $B$ increased from 1,000
to 2,000.

\begin{table}[!htbp]
\centering
\caption{EWOC and rollout under the fitted 1PLD model for the fixed linear response curve, with 1,000 paired replications under each upward cap. Entries are mean (MCSE); RAE and early stopping are displayed as percentages, and RAE enters the losses as a proportion. Loss reduction is $100(1-\overline L_{\mathrm{policy}}/\overline L_{\mathrm{EWOC}})$ within each cap and profile. Only Escalation is enabled in all three panels.}
\label{tab:rl-seven-curve}
\begingroup
\scriptsize
\setlength{\tabcolsep}{1.2pt}
\renewcommand{\arraystretch}{1.0}
\begin{tabular*}{\textwidth}{@{}@{\extracolsep{\fill}}llrrrrrrrr@{}}
\toprule
 & & & & \multicolumn{2}{c}{\textbf{Safety performance}} & \multicolumn{2}{c}{\textbf{Precision performance}} & & \\
\cmidrule(lr){5-6}\cmidrule(lr){7-8}
Loss profile & Policy & \shortstack{Weighted\\loss} & \shortstack{Loss reduction\\vs. EWOC\\(\%)} & \multicolumn{1}{c}{NPD} & \multicolumn{1}{c}{NPO} & \multicolumn{1}{c}{BTM} & \multicolumn{1}{c}{RAE (\%)} & \shortstack{Mean\\enrollment} & \shortstack{Early stop\\(\%)}\\
\midrule
\multicolumn{10}{@{}l}{\textbf{(a) Upward cap $\Delta_{\max}=3.5$}}\\[2pt]
\shortstack[l]{$L_B=d+o+e$\\[0.25ex]\textit{(balanced)}} & EWOC & \shortstack[c]{0.0281\\(0.0011)} & 0.0 & \shortstack[c]{0.3180\\(0.0175)} & \shortstack[c]{0.3840\\(0.0449)} & \shortstack[c]{-0.593\\(0.013)} & \shortstack[c]{1.250\\(0.022)} & \shortstack[c]{45.00\\(0.00)} & \shortstack[c]{0.0\\(0.0)}\\*
 & Rollout, \(B=1{,}000\) & \shortstack[c]{0.0249\\(0.0006)} & 11.5 & \shortstack[c]{0.1320\\(0.0113)} & \shortstack[c]{0.1500\\(0.0207)} & \shortstack[c]{-0.907\\(0.021)} & \shortstack[c]{1.860\\(0.040)} & \shortstack[c]{45.00\\(0.00)} & \shortstack[c]{0.0\\(0.0)}\\*
 & Rollout, \(B=2{,}000\) & \shortstack[c]{0.0247\\(0.0006)} & 12.0 & \shortstack[c]{0.1320\\(0.0113)} & \shortstack[c]{0.1440\\(0.0203)} & \shortstack[c]{-0.906\\(0.021)} & \shortstack[c]{1.858\\(0.040)} & \shortstack[c]{45.00\\(0.00)} & \shortstack[c]{0.0\\(0.0)}\\
\addlinespace[4pt]
\shortstack[l]{$L_S=2d+2o+e$\\[0.25ex]\textit{(safety-emphasized)}} & EWOC & \shortstack[c]{0.0437\\(0.0023)} & 0.0 & \shortstack[c]{0.3180\\(0.0175)} & \shortstack[c]{0.3840\\(0.0449)} & \shortstack[c]{-0.593\\(0.013)} & \shortstack[c]{1.250\\(0.022)} & \shortstack[c]{45.00\\(0.00)} & \shortstack[c]{0.0\\(0.0)}\\*
 & Rollout, \(B=1{,}000\) & \shortstack[c]{0.0328\\(0.0009)} & 24.8 & \shortstack[c]{0.1000\\(0.0101)} & \shortstack[c]{0.0690\\(0.0142)} & \shortstack[c]{-1.254\\(0.029)} & \shortstack[c]{2.534\\(0.057)} & \shortstack[c]{45.00\\(0.00)} & \shortstack[c]{0.0\\(0.0)}\\*
 & Rollout, \(B=2{,}000\) & \shortstack[c]{0.0324\\(0.0009)} & 25.9 & \shortstack[c]{0.1010\\(0.0102)} & \shortstack[c]{0.0600\\(0.0133)} & \shortstack[c]{-1.248\\(0.029)} & \shortstack[c]{2.522\\(0.057)} & \shortstack[c]{45.00\\(0.00)} & \shortstack[c]{0.0\\(0.0)}\\
\addlinespace[4pt]
\shortstack[l]{$L_P=d+o+2e$\\[0.25ex]\textit{(precision-emphasized)}} & EWOC & \shortstack[c]{0.0406\\(0.0011)} & 0.0 & \shortstack[c]{0.3180\\(0.0175)} & \shortstack[c]{0.3840\\(0.0449)} & \shortstack[c]{-0.593\\(0.013)} & \shortstack[c]{1.250\\(0.022)} & \shortstack[c]{45.00\\(0.00)} & \shortstack[c]{0.0\\(0.0)}\\*
 & Rollout, \(B=1{,}000\) & \shortstack[c]{0.0388\\(0.0007)} & 4.4 & \shortstack[c]{0.1440\\(0.0118)} & \shortstack[c]{0.1590\\(0.0213)} & \shortstack[c]{-0.778\\(0.016)} & \shortstack[c]{1.603\\(0.030)} & \shortstack[c]{45.00\\(0.00)} & \shortstack[c]{0.0\\(0.0)}\\*
 & Rollout, \(B=2{,}000\) & \shortstack[c]{0.0389\\(0.0007)} & 4.1 & \shortstack[c]{0.1420\\(0.0117)} & \shortstack[c]{0.1590\\(0.0213)} & \shortstack[c]{-0.782\\(0.016)} & \shortstack[c]{1.611\\(0.030)} & \shortstack[c]{45.00\\(0.00)} & \shortstack[c]{0.0\\(0.0)}\\
\addlinespace[7pt]
\multicolumn{10}{@{}l}{\textbf{(b) Upward cap $\Delta_{\max}=7$}}\\[2pt]
\shortstack[l]{$L_B=d+o+e$\\[0.25ex]\textit{(balanced)}} & EWOC & \shortstack[c]{0.0935\\(0.0050)} & 0.0 & \shortstack[c]{1.3640\\(0.0388)} & \shortstack[c]{2.4600\\(0.2104)} & \shortstack[c]{-0.367\\(0.011)} & \shortstack[c]{0.855\\(0.017)} & \shortstack[c]{44.92\\(0.03)} & \shortstack[c]{1.1\\(0.3)}\\*
 & Rollout, \(B=1{,}000\) & \shortstack[c]{0.0266\\(0.0006)} & 71.5 & \shortstack[c]{0.1160\\(0.0108)} & \shortstack[c]{0.1590\\(0.0213)} & \shortstack[c]{-0.994\\(0.020)} & \shortstack[c]{2.053\\(0.037)} & \shortstack[c]{45.00\\(0.00)} & \shortstack[c]{0.0\\(0.0)}\\*
 & Rollout, \(B=2{,}000\) & \shortstack[c]{0.0268\\(0.0006)} & 71.3 & \shortstack[c]{0.1140\\(0.0107)} & \shortstack[c]{0.1650\\(0.0216)} & \shortstack[c]{-0.999\\(0.020)} & \shortstack[c]{2.062\\(0.037)} & \shortstack[c]{45.00\\(0.00)} & \shortstack[c]{0.0\\(0.0)}\\
\addlinespace[4pt]
\shortstack[l]{$L_S=2d+2o+e$\\[0.25ex]\textit{(safety-emphasized)}} & EWOC & \shortstack[c]{0.1785\\(0.0100)} & 0.0 & \shortstack[c]{1.3640\\(0.0388)} & \shortstack[c]{2.4600\\(0.2104)} & \shortstack[c]{-0.367\\(0.011)} & \shortstack[c]{0.855\\(0.017)} & \shortstack[c]{44.92\\(0.03)} & \shortstack[c]{1.1\\(0.3)}\\*
 & Rollout, \(B=1{,}000\) & \shortstack[c]{0.0341\\(0.0011)} & 80.9 & \shortstack[c]{0.1060\\(0.0104)} & \shortstack[c]{0.1560\\(0.0211)} & \shortstack[c]{-1.090\\(0.024)} & \shortstack[c]{2.242\\(0.045)} & \shortstack[c]{45.00\\(0.00)} & \shortstack[c]{0.0\\(0.0)}\\*
 & Rollout, \(B=2{,}000\) & \shortstack[c]{0.0342\\(0.0011)} & 80.9 & \shortstack[c]{0.1050\\(0.0104)} & \shortstack[c]{0.1560\\(0.0211)} & \shortstack[c]{-1.097\\(0.025)} & \shortstack[c]{2.256\\(0.046)} & \shortstack[c]{45.00\\(0.00)} & \shortstack[c]{0.0\\(0.0)}\\
\addlinespace[4pt]
\shortstack[l]{$L_P=d+o+2e$\\[0.25ex]\textit{(precision-emphasized)}} & EWOC & \shortstack[c]{0.1021\\(0.0049)} & 0.0 & \shortstack[c]{1.3640\\(0.0388)} & \shortstack[c]{2.4600\\(0.2104)} & \shortstack[c]{-0.367\\(0.011)} & \shortstack[c]{0.855\\(0.017)} & \shortstack[c]{44.92\\(0.03)} & \shortstack[c]{1.1\\(0.3)}\\*
 & Rollout, \(B=1{,}000\) & \shortstack[c]{0.0451\\(0.0007)} & 55.8 & \shortstack[c]{0.1210\\(0.0110)} & \shortstack[c]{0.1710\\(0.0220)} & \shortstack[c]{-0.934\\(0.019)} & \shortstack[c]{1.931\\(0.034)} & \shortstack[c]{45.00\\(0.00)} & \shortstack[c]{0.0\\(0.0)}\\*
 & Rollout, \(B=2{,}000\) & \shortstack[c]{0.0450\\(0.0007)} & 55.9 & \shortstack[c]{0.1210\\(0.0110)} & \shortstack[c]{0.1680\\(0.0218)} & \shortstack[c]{-0.933\\(0.019)} & \shortstack[c]{1.930\\(0.034)} & \shortstack[c]{45.00\\(0.00)} & \shortstack[c]{0.0\\(0.0)}\\
\addlinespace[7pt]
\multicolumn{10}{@{}l}{\textbf{(c) Upward cap $\Delta_{\max}=10.5$}}\\[2pt]
\shortstack[l]{$L_B=d+o+e$\\[0.25ex]\textit{(balanced)}} & EWOC & \shortstack[c]{0.1206\\(0.0063)} & 0.0 & \shortstack[c]{1.7590\\(0.0444)} & \shortstack[c]{3.3090\\(0.2674)} & \shortstack[c]{-0.323\\(0.011)} & \shortstack[c]{0.795\\(0.016)} & \shortstack[c]{44.66\\(0.06)} & \shortstack[c]{3.8\\(0.6)}\\*
 & Rollout, \(B=1{,}000\) & \shortstack[c]{0.0293\\(0.0006)} & 75.7 & \shortstack[c]{0.0980\\(0.0101)} & \shortstack[c]{0.1890\\(0.0231)} & \shortstack[c]{-1.106\\(0.022)} & \shortstack[c]{2.292\\(0.041)} & \shortstack[c]{45.00\\(0.00)} & \shortstack[c]{0.0\\(0.0)}\\*
 & Rollout, \(B=2{,}000\) & \shortstack[c]{0.0292\\(0.0006)} & 75.8 & \shortstack[c]{0.0960\\(0.0100)} & \shortstack[c]{0.1920\\(0.0232)} & \shortstack[c]{-1.099\\(0.022)} & \shortstack[c]{2.275\\(0.040)} & \shortstack[c]{45.00\\(0.00)} & \shortstack[c]{0.0\\(0.0)}\\
\addlinespace[4pt]
\shortstack[l]{$L_S=2d+2o+e$\\[0.25ex]\textit{(safety-emphasized)}} & EWOC & \shortstack[c]{0.2332\\(0.0127)} & 0.0 & \shortstack[c]{1.7590\\(0.0444)} & \shortstack[c]{3.3090\\(0.2674)} & \shortstack[c]{-0.323\\(0.011)} & \shortstack[c]{0.795\\(0.016)} & \shortstack[c]{44.66\\(0.06)} & \shortstack[c]{3.8\\(0.6)}\\*
 & Rollout, \(B=1{,}000\) & \shortstack[c]{0.0367\\(0.0012)} & 84.3 & \shortstack[c]{0.0970\\(0.0101)} & \shortstack[c]{0.1890\\(0.0231)} & \shortstack[c]{-1.160\\(0.023)} & \shortstack[c]{2.401\\(0.042)} & \shortstack[c]{45.00\\(0.00)} & \shortstack[c]{0.0\\(0.0)}\\*
 & Rollout, \(B=2{,}000\) & \shortstack[c]{0.0369\\(0.0012)} & 84.2 & \shortstack[c]{0.0960\\(0.0100)} & \shortstack[c]{0.1980\\(0.0236)} & \shortstack[c]{-1.149\\(0.023)} & \shortstack[c]{2.381\\(0.042)} & \shortstack[c]{45.00\\(0.00)} & \shortstack[c]{0.0\\(0.0)}\\
\addlinespace[4pt]
\shortstack[l]{$L_P=d+o+2e$\\[0.25ex]\textit{(precision-emphasized)}} & EWOC & \shortstack[c]{0.1285\\(0.0063)} & 0.0 & \shortstack[c]{1.7590\\(0.0444)} & \shortstack[c]{3.3090\\(0.2674)} & \shortstack[c]{-0.323\\(0.011)} & \shortstack[c]{0.795\\(0.016)} & \shortstack[c]{44.66\\(0.06)} & \shortstack[c]{3.8\\(0.6)}\\*
 & Rollout, \(B=1{,}000\) & \shortstack[c]{0.0481\\(0.0008)} & 62.6 & \shortstack[c]{0.1170\\(0.0108)} & \shortstack[c]{0.1890\\(0.0231)} & \shortstack[c]{-0.996\\(0.020)} & \shortstack[c]{2.063\\(0.036)} & \shortstack[c]{45.00\\(0.00)} & \shortstack[c]{0.0\\(0.0)}\\*
 & Rollout, \(B=2{,}000\) & \shortstack[c]{0.0478\\(0.0008)} & 62.8 & \shortstack[c]{0.1180\\(0.0109)} & \shortstack[c]{0.1890\\(0.0231)} & \shortstack[c]{-0.988\\(0.020)} & \shortstack[c]{2.049\\(0.036)} & \shortstack[c]{45.00\\(0.00)} & \shortstack[c]{0.0\\(0.0)}\\
\bottomrule
\end{tabular*}
\endgroup
\end{table}

Relaxing the cap exposed the safety--precision trade-off. Under EWOC,
mean NPD increased from 0.318 to 1.364 and 1.759, and mean NPO from
0.384 to 2.460 and 3.309, as the cap increased from 3.5 to 7 and 10.5;
mean RAE decreased from 1.250\% to 0.855\% and 0.795\%. At the 10.5
cap, rollout kept mean NPD between 0.096 and 0.118 and mean NPO between
0.189 and 0.198 across the three profiles and two planning budgets.
All rollout policies completed enrollment. EWOC stopped early in 1.1\%
of trials under the 7 cap and 3.8\% under the 10.5 cap.

\paragraph{Computational scope and software.}
Let $K$ bound the decision stages, $A_k\leq A$ the candidates at stage
$k$, $H_k\leq K-k$ the cohorts after the candidate cohort, and $c_m$ a bound on the cost of one
simulated cohort for model $m$, including posterior updating, baseline-dose
selection, response generation, and loss accounting. Full rollout planning
per trial has work
\[
 T_{\mathrm{plan}}
 = O\!\left(B\sum_{k=1}^{K} A_k(1+H_k)c_m\right)
 = O(BAK^2c_m).
\]
After the prespecified first cohort, $K=14$ decisions remain. With up to five
candidate doses per decision, the number of simulated cohorts per trial is
at most $5B\sum_{h=1}^{K}h=5BK(K+1)/2$. A sole candidate bypasses simulation,
and early stopping shortens the paths. At $B=2{,}000$ the ceiling is
1.05 million simulated cohorts per trial (Figure~\ref{fig:rollout-work-bound}(a)).

{Rollout favors an inexpensive base policy because it repeats the posterior update at each simulated decision. Paths can run in parallel, but within a path each EWOC dose depends on the responses from the preceding cohort. To reduce rollout cost, \citet[Sec.~2.2]{tesauro1996online} shortened paths, \citet{bartroff2010dp} fitted an inexpensive rule to EWOC rollouts, and \citet{lee2020efficientrollout} used variance reduction in Bayesian optimization. Even with these savings, each posterior refit for 2PLD, 3PND, iBART, and mBART still requires MCMC in our implementation.}

{Only 1PLD has a closed-form posterior update and EWOC quantile, so we used it for the experiment with 1,000 and 2,000 paths per candidate (Table~\ref{tab:rl-seven-curve}). With full refits (Table~\ref{tab:posterior-computation-settings}), the per-trial ceilings are 142.8 billion slice sweeps for either 2PLD or 3PND, 35.7 billion tree sweeps for iBART, and 10.5 billion for mBART (Figure~\ref{fig:rollout-work-bound}(b)). Each tree sweep visits 200 trees. The desktop simulator offers rollout for all five models at a user-selected budget. A replicated comparison of all four MCMC models at 2,000 paths per candidate with full posterior refits may require supercomputing resources.}

\begin{figure}[!htbp]
\centering
\includegraphics[width=\textwidth]{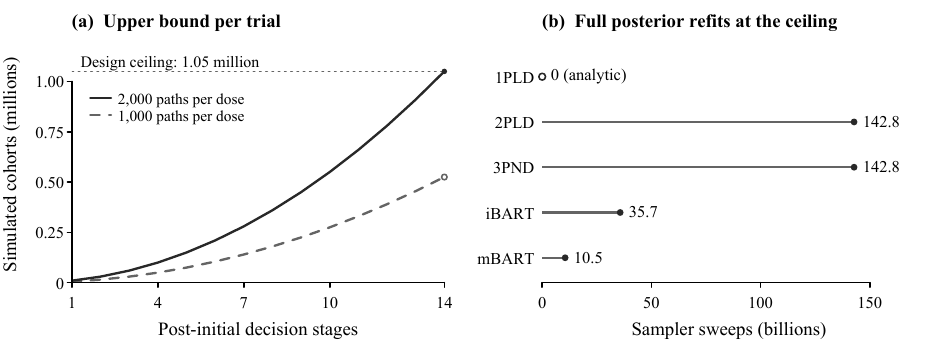}
\caption{{\textbf{Computational upper bounds for repeated EWOC rollout.}
Panel (a) counts simulated cohorts with five candidates at every decision and
continuation to full enrollment. Panel (b) uses 14 decisions, 2,000 paths per
candidate, and the full burn-in and retained sweeps for both chains at
every numerical update. 1PLD uses a closed-form update.}}
\label{fig:rollout-work-bound}
\end{figure}

\FloatBarrier

\section{Simulation design}
\label{sec:bart-experiment}
\subsection{Enrollment, allocation, and stopping}
The study compares all five procedures with $N=45$, cohort size
$c=3$, initial dose 6, interval $(5,80)$, $\eta=3$, $\gamma=0.9$, and
$\alpha=0.05$. Both the upward cap $\Delta_{\max}=3.5$ and
Only Escalation are enabled. A cohort receives a common dose, and all
its responses are observed before posterior updating and the next
allocation. The procedure stops at a completed cohort if
$D_n>\lfloor0.1N\rfloor=4$, or if $n=45$. The fifth observed DLT triggers
the count rule; a complete cohort may overshoot it. Posterior feasibility
exceptions do not stop the trial. The full dose calculation is
\eqref{eq:bart-restricted-dose}.

Starting at 6, a complete trial has 14 adaptive cohort changes and therefore a
maximum reachable dose of 55. Cohort $j$ receives at most
$6+(j-1)3.5$, giving a maximum of 48 at cohort 13. Only the final
six patients can be assigned above the true MTD 50, so $\mathrm{NPO}\leq6$
for every path, including early stops. This reachability constraint is part of
the study design, not information supplied to a fitted model about the
true MTD.

{\textbf{Enrollment limited to 30 patients.} We reanalyzed all 3,500 original EWOC trajectories (residual SD 0.1), retaining observations through cohort 10 or the first completed cohort with more than three DLTs, whichever came first. EWOC uses only the observed history before stopping, so these prefixes reproduce the shorter trial on the same outcome streams. No DLT occurred by patient 30. Both trees reached the maximum attainable dose of 37.5 in all 700 trials per method; the corresponding counts for 1PLD, 2PLD, and 3PND were 514, 498, and 600. On the square-root curve, mean last-dose RAE was 29.21\%, 42.33\%, and 27.00\% for the three parametric methods, versus 25.00\% for each tree. These last doses describe allocation. To examine inference separately, we also computed the absolute percentage error of the posterior median MTD after patient 30. On the linear curve it averaged 60.00\% for iBART and 29.69\% for mBART, despite identical last doses. The iBART median lay at the upper dose boundary of 80 in all 100 linear trials. The data end at or below 37.5, so estimating the boundary at 50 requires extrapolation.}

\subsection{Seven true dose--toxicity curves}
\label{sec:bart-seven-curves}
Let $t=(x-5)/45$ and $A=3-0.1\Phi^{-1}(0.9)$. Responses satisfy
\[
 Y_i=f_0(X_i)+\epsilon_i,\qquad f_0(x)=A h(x),\qquad
 \epsilon_i\stackrel{\mathrm{iid}}{\sim}\mathcal N(0,0.1^2).
\]
The seven shapes are defined in Table~\ref{tab:n45-truth-definitions} and
displayed in Figure~\ref{fig:bart-truth-scenarios}. Each is nondecreasing
with $h(5)=0$ and $h(50)=1$, giving true risk $0.1$ and MTD 50 at dose 50.
The mean there is below the response threshold $\eta=3$ because residual
variation is positive. All models receive dose inputs in physical units;
$t$ only defines the data-generating curves.

\begin{table}[!htbp]\centering\small
\caption{The seven data-generating shapes, exactly matching the software.}
\label{tab:n45-truth-definitions}
\small
\setlength{\tabcolsep}{3pt}
\renewcommand{\arraystretch}{1.0}
\begin{tabular*}{\textwidth}{@{}@{\extracolsep{\fill}}ll@{}}\toprule
Curve & $h(x)$\\\midrule
Linear & $t$\\
Piecewise linear & Linear interpolation of the knots below\\
Quadratic (convex) & $t^2$\\
Square root (concave) & $\sqrt{t}$\\
Sigmoid & $\{L(x)-L(5)\}/\{L(50)-L(5)\}$,
  $L(x)=\{1+\exp[-(x-42.5)/8]\}^{-1}$\\
Exponential (convex) & $(e^{2t}-1)/(e^2-1)$\\
Logarithmic (concave) & $\log(1+4t)/\log(5)$\\\bottomrule
\end{tabular*}
\end{table}
The piecewise linear knots are
\[
\begin{array}{c|rrrrrrr}
x&5&20&35&45&50&60&80\\\hline
h(x)&0&0.04&0.08&0.5&1&1.2&1.25.
\end{array}
\]
\begin{figure}[!htbp]
\centering\includegraphics[width=\textwidth]{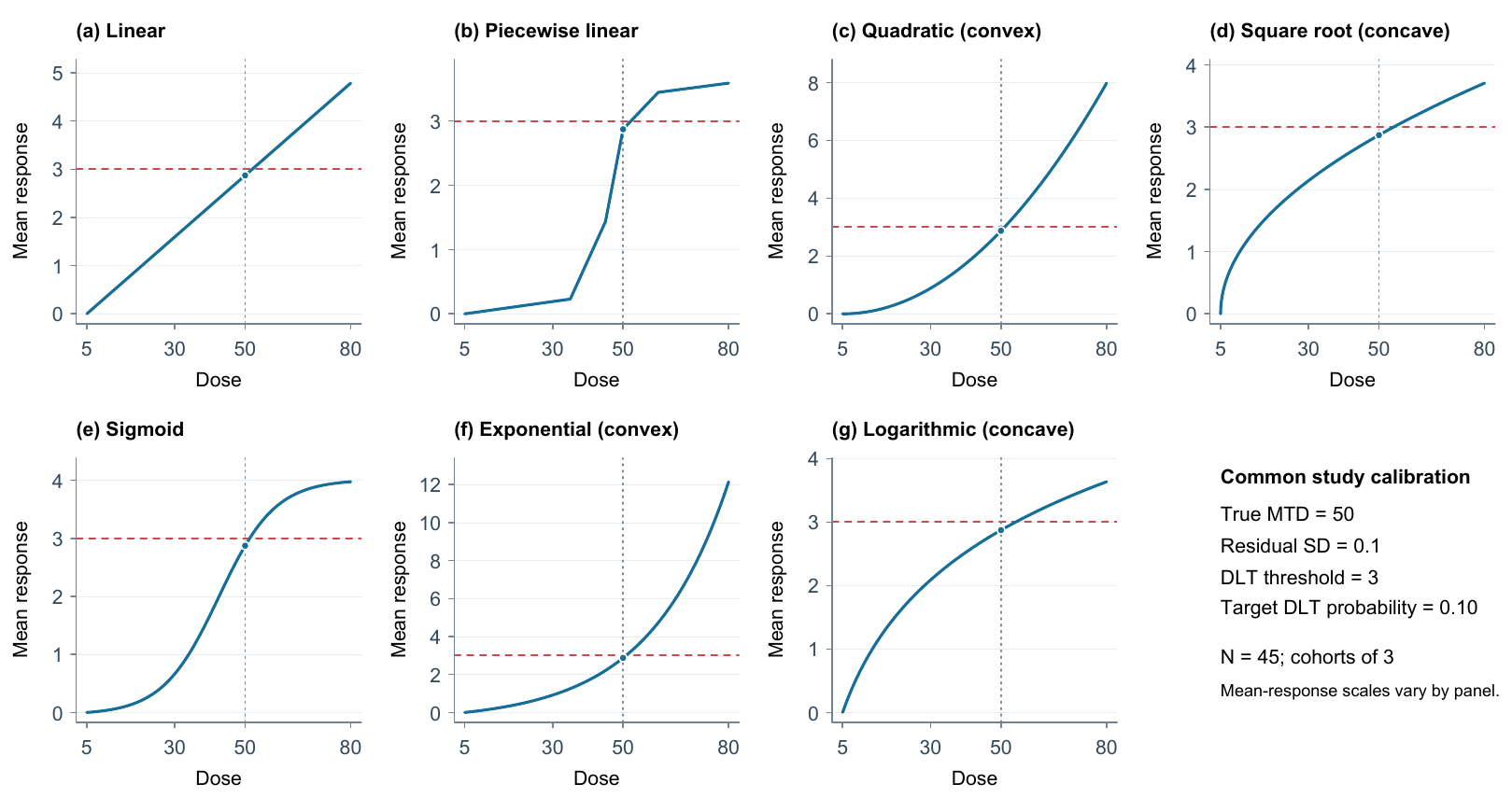}
\caption{Seven true mean dose--toxicity curves used in the $N=45$ study.
All panels share the physical dose axis $(5,80)$; mean-response axes are
scaled separately to show each shape. The red horizontal line is the
toxicity-response threshold $\eta=3$, and the dotted vertical line is
the true MTD 50. At that dose the mean is
$3-0.1\Phi^{-1}(0.9)<3$, so the Gaussian response tail has DLT probability
0.1. The curve is a mean, not a posterior interval or an individual outcome.}
\label{fig:bart-truth-scenarios}
\end{figure}

\subsection{Priors, replication, and computation}
There are 100 replications for each of the seven curves and five methods,
totaling 3,500 trials. Patient innovations are paired
across methods within curve and replication; posterior random streams
are separate. No truth or target location is supplied to allocation.
The 1PLD slope prior is $\mathcal N(0.04,0.02^2)$ with known $\sigma=0.1$.
The 2PLD and 3PND priors follow Appendix~\ref{app:bayesian-comparators},
with $\delta=0.1$ for 3PND. Tree models use 200 trees, shrinkage $k=2$,
response reference range $[0,4]$, 100 physical-dose cuts, and the
inverse-$\chi^2_3$ variance prior with $\Pr(\sigma<0.2)=0.9$.
The reference range does not truncate observed responses. Full sampling
settings appear in Appendix~\ref{app:bart-numerics}.

\paragraph{Exact random-stream schedule.}
The study uses R 4.5.1 with the Mersenne--Twister generator,
Inversion normal generation, and Rejection sampling. Number the curves
$s=1,\ldots,7$ in Table~\ref{tab:n45-truth-definitions} order, the
replications $r=1,\ldots,100$, and methods $m=1,\ldots,5$ as 1PLD,
2PLD, 3PND, iBART, and mBART. The job index is
\[
 j=5\{100(s-1)+(r-1)\}+m.
\]
The patient-innovation seed is $801000000+10000s+r$, shared by all
five methods within a curve and replication. The posterior seed for
cohort $k$ and chain $\ell\in\{1,2\}$ is
$810000000+100000j+100k+\ell$. The analytic 1PLD posterior requires no
MCMC stream.

\subsection{Four performance criteria}
\label{sec:bart-metrics}
At the realized endpoint $n=n_{\mathrm{end}}$, define
\begin{equation}
 \begin{aligned}
 \mathrm{NPD}(n)&=\sum_{i=1}^{n}\mathbf1\{Y_i\geq\eta\},&
 \mathrm{NPO}(n)&=\sum_{i=1}^{n}\mathbf1\{X_i>\xi_0\},\\
 \mathrm{BTM}(n)&=X_n-\xi_0,&
 \mathrm{RAE}(n)&=|X_n-\xi_0|/\xi_0.
 \end{aligned}
 \label{eq:bart-metrics}
\end{equation}
These follow the last-dose convention of \citet{lee2023}. Reported
RAE is multiplied by 100. Counts include every patient in the final
cohort. Trials stopped early remain in all summaries with their actual
endpoint and final dose. Each model--curve cell reports the replicate mean, median, and MCSE,
where MCSE is the sample standard deviation divided by $\sqrt{100}$.
Mean enrollment and early-stop frequency provide context for the four
criteria.

All 3,500 trials enrolled 45 patients, with no early stops. Of these,
2,942 reached the enrollment limit without crossing the DLT-count limit;
558 also crossed that limit in the final cohort. Thus, lower NPD or NPO
cannot be attributed to shorter trials.
\subsection{Relation to the other analyses}
\label{app:experiment-guide}
Section~\ref{sec:rl-seven-curve} reports a separate 1PLD comparison of
EWOC and rollout at two planning budgets under the linear scenario;
its replications are not pooled with the five-model comparison in
Table~\ref{tab:primary-learning}. The pharmacodynamic analysis in
Appendix~\ref{app:bart-case} is a separate 30-patient adaptive
application and contributes no replications to either simulation
comparison.
\section{Complete results for the original seven-curve experiment}
\label{app:paired-reanalysis}
{Table~\ref{tab:primary-learning-full} summarizes all four outcomes for the five methods in the original experiment, with 100 paired trials per curve.}

\begingroup
\fontsize{8}{9}\selectfont
\setlength{\tabcolsep}{2.75pt}
\renewcommand{\arraystretch}{1.0}
% Fixed numeric fields align both separators down each metric column.
\ifdefined\MainTableCountWidth\else\newlength{\MainTableCountWidth}\fi
\ifdefined\MainTableSignedWidth\else\newlength{\MainTableSignedWidth}\fi
\ifdefined\MainTableRelativeWidth\else\newlength{\MainTableRelativeWidth}\fi
\ifdefined\MainTableMCSEWidth\else\newlength{\MainTableMCSEWidth}\fi
\settowidth{\MainTableCountWidth}{0.00}
\settowidth{\MainTableSignedWidth}{-00.00}
\settowidth{\MainTableRelativeWidth}{00.00}
\settowidth{\MainTableMCSEWidth}{0.00}
\providecommand{\MainTableTriple}[5]{%
 \makebox[#1][r]{#3}\,\textbar\,%
 \makebox[#2][r]{#4}\,\textbar\,%
 \makebox[\MainTableMCSEWidth][r]{#5}}
\providecommand{\NPDStats}[3]{\MainTableTriple{\MainTableCountWidth}{\MainTableCountWidth}{#1}{#2}{#3}}
\providecommand{\NPOStats}[3]{\MainTableTriple{\MainTableCountWidth}{\MainTableCountWidth}{#1}{#2}{#3}}
\providecommand{\BTMStats}[3]{\MainTableTriple{\MainTableSignedWidth}{\MainTableSignedWidth}{#1}{#2}{#3}}
\providecommand{\RAEStats}[3]{\MainTableTriple{\MainTableRelativeWidth}{\MainTableRelativeWidth}{#1}{#2}{#3}}

\setlength{\LTpre}{6pt}
\setlength{\LTpost}{6pt}
\setlength{\LTcapwidth}{\textwidth}
\setlength{\LTleft}{0pt}
\setlength{\LTright}{0pt}
\begin{longtable}{@{}@{\extracolsep{\fill}}llrrrr@{}}
\caption{\small Entries are \emph{mean \textbar\ median \textbar\ MCSE}; RAE is a
percentage. Bold, underlined methods minimize mean RAE ($p$),
both mean NPD and NPO ($s$), or both categories ($s,p$) within each curve,
using unrounded means and exact ties. A zero MCSE may reflect rounding.}\label{tab:primary-learning-full}\\
\toprule
 & & \multicolumn{2}{c}{\textbf{Safety performance}} & \multicolumn{2}{c}{\textbf{Precision performance}} \\
\cmidrule(lr){3-4}\cmidrule(lr){5-6}
Curve & Method & \multicolumn{1}{c}{NPD} & \multicolumn{1}{c}{NPO} & \multicolumn{1}{c}{BTM} & \multicolumn{1}{c}{RAE (\%)}\\
\midrule
\endfirsthead
\multicolumn{6}{@{}l}{\small Table \thetable\ (continued)}\\[2pt]
\toprule
 & & \multicolumn{2}{c}{\textbf{Safety performance}} & \multicolumn{2}{c}{\textbf{Precision performance}} \\
\cmidrule(lr){3-4}\cmidrule(lr){5-6}
Curve & Method & \multicolumn{1}{c}{NPD} & \multicolumn{1}{c}{NPO} & \multicolumn{1}{c}{BTM} & \multicolumn{1}{c}{RAE (\%)}\\
\midrule
\endhead
\bottomrule
\endfoot
\bottomrule
\endlastfoot

Linear & \textbf{\underline{1PLD}}\textsuperscript{p} & \NPDStats{0.34}{0.00}{0.06} & \NPOStats{0.33}{0.00}{0.13} & \BTMStats{-0.57}{-0.58}{0.04} & \RAEStats{1.20}{1.16}{0.07} \\*
 & 2PLD & \NPDStats{0.26}{0.00}{0.05} & \NPOStats{0.24}{0.00}{0.11} & \BTMStats{-0.75}{-0.75}{0.05} & \RAEStats{1.53}{1.49}{0.09} \\*
 & 3PND & \NPDStats{0.20}{0.00}{0.04} & \NPOStats{0.27}{0.00}{0.11} & \BTMStats{-0.89}{-0.87}{0.06} & \RAEStats{1.83}{1.74}{0.11} \\*
 & \textbf{\underline{iBART}}\textsuperscript{s} & \NPDStats{0.03}{0.00}{0.02} & \NPOStats{0.00}{0.00}{0.00} & \BTMStats{-1.99}{-1.93}{0.06} & \RAEStats{3.98}{3.86}{0.12} \\*
 & mBART & \NPDStats{0.06}{0.00}{0.02} & \NPOStats{0.06}{0.00}{0.06} & \BTMStats{-1.61}{-1.93}{0.06} & \RAEStats{3.24}{3.86}{0.11} \\
\addlinespace[1.5pt]
Piecewise linear & 1PLD & \NPDStats{3.86}{4.00}{0.09} & \NPOStats{6.00}{6.00}{0.00} & \BTMStats{5.00}{5.00}{0.00} & \RAEStats{10.00}{10.00}{0.00} \\*
 & 2PLD & \NPDStats{3.86}{4.00}{0.09} & \NPOStats{6.00}{6.00}{0.00} & \BTMStats{5.00}{5.00}{0.00} & \RAEStats{10.00}{10.00}{0.00} \\*
 & 3PND & \NPDStats{3.86}{4.00}{0.09} & \NPOStats{6.00}{6.00}{0.00} & \BTMStats{4.98}{5.00}{0.01} & \RAEStats{9.97}{10.00}{0.01} \\*
 & \textbf{\underline{iBART}}\textsuperscript{s} & \NPDStats{2.02}{2.00}{0.12} & \NPOStats{6.00}{6.00}{0.00} & \BTMStats{1.51}{1.50}{0.00} & \RAEStats{3.02}{3.00}{0.01} \\*
 & \textbf{\underline{mBART}}\textsuperscript{s,p} & \NPDStats{2.02}{2.00}{0.12} & \NPOStats{6.00}{6.00}{0.00} & \BTMStats{1.50}{1.50}{0.00} & \RAEStats{3.00}{3.00}{0.00} \\
\addlinespace[1.5pt]
Quadratic & 1PLD & \NPDStats{5.21}{5.00}{0.07} & \NPOStats{6.00}{6.00}{0.00} & \BTMStats{5.00}{5.00}{0.00} & \RAEStats{10.00}{10.00}{0.00} \\*
 & 2PLD & \NPDStats{4.36}{4.00}{0.10} & \NPOStats{6.00}{6.00}{0.00} & \BTMStats{1.51}{1.50}{0.01} & \RAEStats{3.02}{3.00}{0.01} \\*
 & \textbf{\underline{3PND}}\textsuperscript{p} & \NPDStats{0.93}{1.00}{0.10} & \NPOStats{3.60}{6.00}{0.29} & \BTMStats{0.12}{0.09}{0.04} & \RAEStats{0.55}{0.42}{0.05} \\*
 & \textbf{\underline{iBART}}\textsuperscript{s} & \NPDStats{0.39}{0.00}{0.09} & \NPOStats{0.99}{0.00}{0.22} & \BTMStats{-0.46}{-0.45}{0.05} & \RAEStats{1.19}{0.89}{0.07} \\*
 & mBART & \NPDStats{1.05}{0.00}{0.15} & \NPOStats{2.73}{0.00}{0.30} & \BTMStats{0.04}{-0.45}{0.07} & \RAEStats{1.20}{0.89}{0.07} \\
\addlinespace[1.5pt]
Square root & \textbf{\underline{1PLD}}\textsuperscript{s} & \NPDStats{0.00}{0.00}{0.00} & \NPOStats{0.00}{0.00}{0.00} & \BTMStats{-10.02}{-10.03}{0.04} & \RAEStats{20.05}{20.07}{0.08} \\*
 & \textbf{\underline{2PLD}}\textsuperscript{s} & \NPDStats{0.00}{0.00}{0.00} & \NPOStats{0.00}{0.00}{0.00} & \BTMStats{-17.78}{-17.81}{0.06} & \RAEStats{35.56}{35.62}{0.13} \\*
 & \textbf{\underline{3PND}}\textsuperscript{p} & \NPDStats{0.11}{0.00}{0.03} & \NPOStats{0.00}{0.00}{0.00} & \BTMStats{-3.45}{-3.40}{0.12} & \RAEStats{6.90}{6.80}{0.24} \\*
 & iBART & \NPDStats{0.03}{0.00}{0.02} & \NPOStats{0.00}{0.00}{0.00} & \BTMStats{-4.75}{-4.90}{0.09} & \RAEStats{9.49}{9.80}{0.18} \\*
 & mBART & \NPDStats{0.06}{0.00}{0.02} & \NPOStats{0.00}{0.00}{0.00} & \BTMStats{-4.37}{-4.16}{0.09} & \RAEStats{8.73}{8.32}{0.19} \\
\addlinespace[1.5pt]
Sigmoid & 1PLD & \NPDStats{4.73}{5.00}{0.09} & \NPOStats{6.00}{6.00}{0.00} & \BTMStats{5.00}{5.00}{0.00} & \RAEStats{10.00}{10.00}{0.00} \\*
 & 2PLD & \NPDStats{3.52}{4.00}{0.12} & \NPOStats{6.00}{6.00}{0.00} & \BTMStats{1.50}{1.50}{0.00} & \RAEStats{2.99}{3.00}{0.01} \\*
 & \textbf{\underline{3PND}}\textsuperscript{s} & \NPDStats{0.12}{0.00}{0.04} & \NPOStats{0.00}{0.00}{0.00} & \BTMStats{-0.78}{-0.76}{0.03} & \RAEStats{1.55}{1.52}{0.07} \\*
 & iBART & \NPDStats{0.14}{0.00}{0.04} & \NPOStats{0.27}{0.00}{0.12} & \BTMStats{-0.84}{-1.19}{0.05} & \RAEStats{1.74}{2.38}{0.09} \\*
 & \textbf{\underline{mBART}}\textsuperscript{p} & \NPDStats{0.53}{0.00}{0.09} & \NPOStats{1.71}{0.00}{0.27} & \BTMStats{-0.40}{-0.45}{0.07} & \RAEStats{1.43}{0.89}{0.09} \\
\addlinespace[1.5pt]
Exponential & 1PLD & \NPDStats{5.67}{6.00}{0.06} & \NPOStats{6.00}{6.00}{0.00} & \BTMStats{5.00}{5.00}{0.00} & \RAEStats{10.00}{10.00}{0.00} \\*
 & 2PLD & \NPDStats{5.28}{5.00}{0.07} & \NPOStats{6.00}{6.00}{0.00} & \BTMStats{1.59}{1.50}{0.02} & \RAEStats{3.18}{3.00}{0.04} \\*
 & 3PND & \NPDStats{3.97}{4.00}{0.17} & \NPOStats{5.94}{6.00}{0.06} & \BTMStats{1.11}{1.21}{0.04} & \RAEStats{2.23}{2.43}{0.08} \\*
 & \textbf{\underline{iBART}}\textsuperscript{s,p} & \NPDStats{0.72}{0.00}{0.13} & \NPOStats{2.25}{0.00}{0.29} & \BTMStats{-0.18}{-0.45}{0.06} & \RAEStats{1.08}{0.89}{0.07} \\*
 & mBART & \NPDStats{2.24}{2.00}{0.21} & \NPOStats{4.47}{6.00}{0.26} & \BTMStats{0.47}{0.30}{0.08} & \RAEStats{1.49}{0.89}{0.10} \\
\addlinespace[1.5pt]
Logarithmic & \textbf{\underline{1PLD}}\textsuperscript{s} & \NPDStats{0.00}{0.00}{0.00} & \NPOStats{0.00}{0.00}{0.00} & \BTMStats{-8.24}{-8.22}{0.04} & \RAEStats{16.49}{16.45}{0.08} \\*
 & \textbf{\underline{2PLD}}\textsuperscript{s} & \NPDStats{0.00}{0.00}{0.00} & \NPOStats{0.00}{0.00}{0.00} & \BTMStats{-12.63}{-12.65}{0.05} & \RAEStats{25.26}{25.30}{0.10} \\*
 & 3PND & \NPDStats{0.03}{0.00}{0.02} & \NPOStats{0.00}{0.00}{0.00} & \BTMStats{-4.72}{-4.83}{0.09} & \RAEStats{9.44}{9.66}{0.18} \\*
 & iBART & \NPDStats{0.02}{0.00}{0.01} & \NPOStats{0.00}{0.00}{0.00} & \BTMStats{-4.46}{-4.16}{0.08} & \RAEStats{8.93}{8.32}{0.16} \\*
 & \textbf{\underline{mBART}}\textsuperscript{p} & \NPDStats{0.04}{0.00}{0.02} & \NPOStats{0.00}{0.00}{0.00} & \BTMStats{-3.95}{-4.16}{0.09} & \RAEStats{7.90}{8.32}{0.19} \\
\end{longtable}
\endgroup

\section{\texorpdfstring{{Higher outcome variability across seven dose--toxicity curves}}{Higher outcome variability across seven dose-toxicity curves}}
\label{app:focused-noise}
{The additional experiment evaluates all seven dose--response curves and all five procedures with 50 paired replications per curve, giving 1,750 trials. The generating residual SD is $0.2$, compared with $0.1$ in the primary experiment. We retain $N=45$, cohorts of three, initial dose 6, dose interval $(5,80)$, upward cap $3.5$, Only Escalation, EWOC level $\alpha=0.05$, and the completed-cohort DLT stopping rule with $K=4$. The prior settings and actual-history posterior sampling budgets remain those of the primary study.}

{For each original mean curve $f_0$, the higher-noise mean is
\[
f_{0.2}(x)=f_0(x)
\frac{\eta-0.2\Phi^{-1}(\gamma)}
     {\eta-0.1\Phi^{-1}(\gamma)},
\qquad \eta=3,\quad\gamma=0.9.
\]
The curve family is unchanged and the true MTD remains 50, since $f_{0.2}(50)+0.2\Phi^{-1}(\gamma)=\eta$. As in the original experiment, 1PLD uses the known generating residual SD and the other four methods estimate it.}

{The fixed replication IDs 1--50 are used for every method and curve. Each curve--replication pair shares one stream of patient-indexed innovations across methods, while posterior sampling uses separate model streams. These streams are distinct from those of the primary experiment. The pairing therefore supports within-setting comparisons between methods. All stopped trials remain in the summaries, with the last administered dose retained as the terminal estimator. We report the four original criteria, enrollment and stopping, and numerical diagnostics. Paired differences use pointwise Student-$t$ Monte Carlo intervals with 49 degrees of freedom.}
\paragraph{{Where precision improves.}}
{Table~\ref{tab:focused-noise-means} summarizes the four outcomes for all 35 curve--model settings. On the piecewise-linear, square-root, exponential and logarithmic curves, both trees had smaller mean last-dose RAE than 3PND, with both paired intervals below zero for each curve (Figure~\ref{fig:higher-noise-tree-outcomes}). On the piecewise-linear curve, mean RAE was $9.30\%$ for 3PND, $2.96\%$ for iBART and $3.05\%$ for mBART. The trees also averaged fewer DLTs there, while both intervals for the differences in above-MTD counts included zero. On the exponential curve, both trees improved precision and reduced both safety counts. On the square-root and logarithmic curves, 3PND tended to end farther below the MTD. Both trees estimated the MTD more closely but incurred more DLTs. All five methods stayed at or below the true MTD on those two curves.}

{On the linear curve, iBART increased mean RAE by $1.24$ percentage points relative to 3PND, with a pointwise 95\% Monte Carlo interval of $[0.88,1.61]$. For mBART the difference was $0.30$ points ($[-0.07,0.67]$). On the quadratic and sigmoid curves both trees reduced DLT counts and above-MTD assignments. All four intervals for the precision differences on these two curves included zero. These comparisons concern the complete procedures, each with its own prior and a common controlled EWOC policy.}

\paragraph{{Stopping, posterior risk and computation.}}
{All trials reached 45 patients and entered the paired comparisons. Each method assigned 4,900 post-initial cohorts of three patients. At the assigned dose, posterior overdose probability exceeded $\alpha=0.05$ in 21 cohorts for 1PLD, 65 for 2PLD, 166 for 3PND, 144 for iBART and 220 for mBART. NPO counts patients dosed above the true MTD.}

{Of 147,000 parameter-by-cohort diagnostic records for the four models with numerical posteriors, 22,671 were constant. Among the remaining records, the largest split $\widehat R$ was $1.0049$ and the smallest effective sample size was $747.8$.}\par
\begingroup

\fontsize{8}{9}\selectfont
\setlength{\tabcolsep}{2.75pt}
\renewcommand{\arraystretch}{1.0}
% Fixed numeric fields align both separators down each metric column.
\ifdefined\MainTableCountWidth\else\newlength{\MainTableCountWidth}\fi
\ifdefined\MainTableSignedWidth\else\newlength{\MainTableSignedWidth}\fi
\ifdefined\MainTableRelativeWidth\else\newlength{\MainTableRelativeWidth}\fi
\ifdefined\MainTableMCSEWidth\else\newlength{\MainTableMCSEWidth}\fi
\settowidth{\MainTableCountWidth}{0.00}
\settowidth{\MainTableSignedWidth}{-00.00}
\settowidth{\MainTableRelativeWidth}{00.00}
\settowidth{\MainTableMCSEWidth}{0.00}
\providecommand{\MainTableTriple}[5]{%
 \makebox[#1][r]{#3}\,\textbar\,%
 \makebox[#2][r]{#4}\,\textbar\,%
 \makebox[\MainTableMCSEWidth][r]{#5}}
\providecommand{\NPDStats}[3]{\MainTableTriple{\MainTableCountWidth}{\MainTableCountWidth}{#1}{#2}{#3}}
\providecommand{\NPOStats}[3]{\MainTableTriple{\MainTableCountWidth}{\MainTableCountWidth}{#1}{#2}{#3}}
\providecommand{\BTMStats}[3]{\MainTableTriple{\MainTableSignedWidth}{\MainTableSignedWidth}{#1}{#2}{#3}}
\providecommand{\RAEStats}[3]{\MainTableTriple{\MainTableRelativeWidth}{\MainTableRelativeWidth}{#1}{#2}{#3}}

\setlength{\LTpre}{6pt}
\setlength{\LTpost}{6pt}
\setlength{\LTcapwidth}{\textwidth}
\setlength{\LTleft}{0pt}
\setlength{\LTright}{0pt}
\begin{longtable}{@{}@{\extracolsep{\fill}}llrrrr@{}}
\caption{\small {Higher outcome variability (residual SD 0.2). All 50 paired replications per curve and model reached 45 patients. Entries are mean \textbar\ median \textbar\ MCSE. RAE is a percentage. Bold, underlined methods minimize mean RAE ($p$), both mean NPD and NPO ($s$), or both ($s,p$) within a curve. Marks use unrounded means and include exact ties. A zero MCSE may reflect rounding.}}\label{tab:focused-noise-means}\\
\toprule
 & & \multicolumn{2}{c}{\textbf{Safety performance}} & \multicolumn{2}{c}{\textbf{Precision performance}} \\
\cmidrule(lr){3-4}\cmidrule(lr){5-6}
Curve & Method & \multicolumn{1}{c}{NPD} & \multicolumn{1}{c}{NPO} & \multicolumn{1}{c}{BTM} & \multicolumn{1}{c}{RAE (\%)}\\
\midrule
\endfirsthead
\multicolumn{6}{@{}l}{\small Table \thetable\ (continued)}\\[2pt]
\toprule
 & & \multicolumn{2}{c}{\textbf{Safety performance}} & \multicolumn{2}{c}{\textbf{Precision performance}} \\
\cmidrule(lr){3-4}\cmidrule(lr){5-6}
Curve & Method & \multicolumn{1}{c}{NPD} & \multicolumn{1}{c}{NPO} & \multicolumn{1}{c}{BTM} & \multicolumn{1}{c}{RAE (\%)}\\
\midrule
\endhead
\bottomrule
\endfoot
\bottomrule
\endlastfoot
Linear & \textbf{\underline{1PLD}}\textsuperscript{p} & \NPDStats{0.48}{0.00}{0.10} & \NPOStats{0.72}{0.00}{0.27} & \BTMStats{-1.01}{-1.03}{0.12} & \RAEStats{2.28}{2.11}{0.19} \\*
 & 2PLD & \NPDStats{0.34}{0.00}{0.07} & \NPOStats{0.66}{0.00}{0.26} & \BTMStats{-1.32}{-1.37}{0.16} & \RAEStats{2.99}{2.88}{0.25} \\*
 & 3PND & \NPDStats{0.32}{0.00}{0.07} & \NPOStats{0.60}{0.00}{0.24} & \BTMStats{-1.40}{-1.36}{0.16} & \RAEStats{3.05}{2.71}{0.28} \\*
 & \textbf{\underline{iBART}}\textsuperscript{s} & \NPDStats{0.22}{0.00}{0.06} & \NPOStats{0.30}{0.00}{0.18} & \BTMStats{-2.11}{-1.93}{0.16} & \RAEStats{4.29}{3.86}{0.29} \\*
 & mBART & \NPDStats{0.38}{0.00}{0.07} & \NPOStats{0.54}{0.00}{0.22} & \BTMStats{-1.50}{-1.93}{0.17} & \RAEStats{3.35}{3.86}{0.25} \\
\addlinespace[1.5pt]
Piecewise linear & 1PLD & \NPDStats{2.26}{2.00}{0.15} & \NPOStats{6.00}{6.00}{0.00} & \BTMStats{5.00}{5.00}{0.00} & \RAEStats{10.00}{10.00}{0.00} \\*
 & 2PLD & \NPDStats{2.02}{2.00}{0.16} & \NPOStats{5.88}{6.00}{0.08} & \BTMStats{4.70}{5.00}{0.07} & \RAEStats{9.41}{10.00}{0.14} \\*
 & 3PND & \NPDStats{1.98}{2.00}{0.16} & \NPOStats{5.88}{6.00}{0.08} & \BTMStats{4.65}{5.00}{0.08} & \RAEStats{9.30}{10.00}{0.15} \\*
 & \textbf{\underline{iBART}}\textsuperscript{s,p} & \NPDStats{1.12}{1.00}{0.14} & \NPOStats{5.88}{6.00}{0.12} & \BTMStats{1.46}{1.50}{0.05} & \RAEStats{2.96}{3.00}{0.06} \\*
 & mBART & \NPDStats{1.14}{1.00}{0.14} & \NPOStats{6.00}{6.00}{0.00} & \BTMStats{1.52}{1.50}{0.02} & \RAEStats{3.05}{3.00}{0.03} \\
\addlinespace[1.5pt]
Quadratic & 1PLD & \NPDStats{4.04}{4.00}{0.13} & \NPOStats{6.00}{6.00}{0.00} & \BTMStats{5.00}{5.00}{0.00} & \RAEStats{10.00}{10.00}{0.00} \\*
 & 2PLD & \NPDStats{2.40}{2.50}{0.16} & \NPOStats{5.94}{6.00}{0.06} & \BTMStats{2.12}{1.60}{0.15} & \RAEStats{4.24}{3.20}{0.30} \\*
 & 3PND & \NPDStats{1.60}{1.00}{0.17} & \NPOStats{5.70}{6.00}{0.18} & \BTMStats{0.98}{1.00}{0.09} & \RAEStats{1.98}{2.00}{0.17} \\*
 & \textbf{\underline{iBART}}\textsuperscript{s,p} & \NPDStats{0.74}{0.00}{0.14} & \NPOStats{2.28}{0.00}{0.41} & \BTMStats{-0.28}{-0.45}{0.15} & \RAEStats{1.81}{2.08}{0.16} \\*
 & mBART & \NPDStats{1.18}{1.00}{0.18} & \NPOStats{3.84}{6.00}{0.40} & \BTMStats{0.37}{0.30}{0.15} & \RAEStats{2.00}{2.38}{0.15} \\
\addlinespace[1.5pt]
Square root & \textbf{\underline{1PLD}}\textsuperscript{s} & \NPDStats{0.00}{0.00}{0.00} & \NPOStats{0.00}{0.00}{0.00} & \BTMStats{-10.80}{-10.90}{0.14} & \RAEStats{21.59}{21.80}{0.29} \\*
 & \textbf{\underline{2PLD}}\textsuperscript{s} & \NPDStats{0.00}{0.00}{0.00} & \NPOStats{0.00}{0.00}{0.00} & \BTMStats{-17.35}{-17.26}{0.25} & \RAEStats{34.70}{34.52}{0.50} \\*
 & 3PND & \NPDStats{0.02}{0.00}{0.02} & \NPOStats{0.00}{0.00}{0.00} & \BTMStats{-9.50}{-9.47}{0.37} & \RAEStats{18.99}{18.95}{0.75} \\*
 & iBART & \NPDStats{0.18}{0.00}{0.08} & \NPOStats{0.00}{0.00}{0.00} & \BTMStats{-4.99}{-4.90}{0.30} & \RAEStats{9.99}{9.80}{0.59} \\*
 & \textbf{\underline{mBART}}\textsuperscript{p} & \NPDStats{0.24}{0.00}{0.08} & \NPOStats{0.00}{0.00}{0.00} & \BTMStats{-4.74}{-4.90}{0.30} & \RAEStats{9.49}{9.80}{0.59} \\
\addlinespace[1.5pt]
Sigmoid & 1PLD & \NPDStats{3.22}{3.00}{0.13} & \NPOStats{6.00}{6.00}{0.00} & \BTMStats{5.00}{5.00}{0.00} & \RAEStats{10.00}{10.00}{0.00} \\*
 & 2PLD & \NPDStats{1.88}{2.00}{0.11} & \NPOStats{5.88}{6.00}{0.08} & \BTMStats{2.13}{1.57}{0.14} & \RAEStats{4.25}{3.14}{0.28} \\*
 & 3PND & \NPDStats{1.32}{1.00}{0.12} & \NPOStats{5.46}{6.00}{0.24} & \BTMStats{1.00}{1.19}{0.09} & \RAEStats{2.15}{2.37}{0.14} \\*
 & \textbf{\underline{iBART}}\textsuperscript{s} & \NPDStats{0.42}{0.00}{0.10} & \NPOStats{1.56}{0.00}{0.35} & \BTMStats{-0.37}{-0.45}{0.15} & \RAEStats{1.89}{2.38}{0.16} \\*
 & \textbf{\underline{mBART}}\textsuperscript{p} & \NPDStats{0.74}{1.00}{0.12} & \NPOStats{3.06}{3.00}{0.41} & \BTMStats{0.19}{0.30}{0.16} & \RAEStats{1.88}{1.49}{0.19} \\
\addlinespace[1.5pt]
Exponential & 1PLD & \NPDStats{4.26}{4.00}{0.12} & \NPOStats{6.00}{6.00}{0.00} & \BTMStats{5.00}{5.00}{0.00} & \RAEStats{10.00}{10.00}{0.00} \\*
 & 2PLD & \NPDStats{3.08}{3.00}{0.16} & \NPOStats{5.94}{6.00}{0.06} & \BTMStats{2.73}{2.73}{0.14} & \RAEStats{5.45}{5.46}{0.28} \\*
 & 3PND & \NPDStats{2.02}{2.00}{0.19} & \NPOStats{5.70}{6.00}{0.18} & \BTMStats{1.28}{1.50}{0.06} & \RAEStats{2.57}{3.00}{0.13} \\*
 & \textbf{\underline{iBART}}\textsuperscript{s,p} & \NPDStats{0.50}{0.00}{0.13} & \NPOStats{2.70}{1.50}{0.40} & \BTMStats{-0.20}{-0.07}{0.13} & \RAEStats{1.51}{0.89}{0.15} \\*
 & mBART & \NPDStats{1.48}{1.00}{0.19} & \NPOStats{4.68}{6.00}{0.34} & \BTMStats{0.58}{1.04}{0.14} & \RAEStats{1.99}{2.08}{0.15} \\
\addlinespace[1.5pt]
Logarithmic & \textbf{\underline{1PLD}}\textsuperscript{s} & \NPDStats{0.02}{0.00}{0.02} & \NPOStats{0.00}{0.00}{0.00} & \BTMStats{-9.13}{-9.25}{0.14} & \RAEStats{18.25}{18.49}{0.28} \\*
 & \textbf{\underline{2PLD}}\textsuperscript{s} & \NPDStats{0.02}{0.00}{0.02} & \NPOStats{0.00}{0.00}{0.00} & \BTMStats{-12.25}{-12.08}{0.17} & \RAEStats{24.50}{24.16}{0.35} \\*
 & \textbf{\underline{3PND}}\textsuperscript{s} & \NPDStats{0.02}{0.00}{0.02} & \NPOStats{0.00}{0.00}{0.00} & \BTMStats{-7.35}{-7.04}{0.25} & \RAEStats{14.71}{14.09}{0.50} \\*
 & iBART & \NPDStats{0.14}{0.00}{0.05} & \NPOStats{0.00}{0.00}{0.00} & \BTMStats{-5.14}{-4.90}{0.21} & \RAEStats{10.28}{9.80}{0.42} \\*
 & \textbf{\underline{mBART}}\textsuperscript{p} & \NPDStats{0.26}{0.00}{0.07} & \NPOStats{0.00}{0.00}{0.00} & \BTMStats{-4.69}{-4.16}{0.22} & \RAEStats{9.39}{8.32}{0.45} \\
\end{longtable}
\endgroup
\FloatBarrier

\section{Pharmacodynamic example}
\label{app:bart-case}
\label{sec:bart-case}
\subsection{Published data and common target}
We compare five models using the $O^6$-benzylguanine example of
\citet[pp.~12--15]{lee2023}, which
specifies $(x_{\min},x_{\max})=(20,140)$\,mg/m$^2$ and reports
24 tumor $O^6$-alkylguanine-DNA alkyltransferase (AGT) measurements
at 40, 60, 80, and 100\,mg/m$^2$ (Table~\ref{tab:bart-case-data}).

\begin{table}[ht]
\centering\small
\caption{Published AGT measurements used as dose-specific resampling pools.}
\label{tab:bart-case-data}
\small
\setlength{\tabcolsep}{3pt}
\renewcommand{\arraystretch}{1.0}
\begin{tabular*}{\textwidth}{@{}@{\extracolsep{\fill}}rrl@{}}
\toprule
Dose (mg/m$^2$)&Count&AGT (fmol/mg)\\
\midrule
40&3&26.35, 42.00, 15.00\\
60&3&23.00, 13.50, 11.00\\
80&9&31.67, 8.00, 9.00, 14.50, 11.50, 7.00, 11.70, 9.03, 8.00\\
100&9&4.07, 5.00, 8.70, 2.50, 4.07, 6.13, 3.60, 5.00, 5.00\\
\bottomrule
\end{tabular*}
\end{table}

Following \citet{lee2023}, let the continuous response be
$R=60-\mathrm{AGT}$, with threshold 55. For 1PLD and the two tree
models we use the equivalent score $Y=R/60$ and $\eta=55/60$;
2PLD and 3PND retain the original response units. Doses remain in
mg/m$^2$. With $\gamma=0.6$ and $\theta=1-\gamma=0.4$, the
common pharmacodynamic boundary is
\begin{equation*}
 p_{\mathrm{dep}}(x)=\Phi\left\{\frac{f(x)-55/60}{\sigma}\right\},
 \qquad
 \xi=\sup\bigl(\{20\}\cup\{x\in(20,140):p_{\mathrm{dep}}(x)\le0.4\}\bigr).
\end{equation*}
Here $p_{\mathrm{dep}}(x)=\Pr(\mathrm{AGT}<5\mid x)$; $f$ and
$\sigma$ in the display use the score scale. For iBART,
$f=\mathcal P g$ is the projected mean functional with each draw's
residual SD, whereas its likelihood uses $g$. The other methods
compute the tail from their fitted Gaussian mean and SD. The endpoints
20 and 140 represent an empty and an entirely acceptable range.

Although AGT depletion measures pharmacodynamic activity rather than
clinical toxicity, the analysis retains the sequential dose-finding
task of learning a monotone dose--response boundary while assigning
doses. These published measurements make the application reproducible;
comparable patient-level toxicity data are difficult to access and
publish because of confidentiality restrictions. The inferred AGT
boundary is a pharmacodynamic target, not a toxicity MTD or a
validated clinical dose recommendation.

For each method we generate a separate 30-patient allocation path in ten cohorts
of three, beginning at 40\,mg/m$^2$. At the assigned dose, each cohort
draws three AGT values without replacement from its published pool and
adds independent centered uniform perturbations to $R=60-\mathrm{AGT}$.
The half-widths are 8, 6, 4, and 2\,fmol/mg at 40, 60, 80, and
100\,mg/m$^2$, respectively, following the redesign of
\citet{lee2023}. A source value may reappear in a later cohort. A common precomputed
cohort-by-dose response bank gives two methods identical outcomes if
they assign the same dose in the same cohort. Each posterior update uses
only that method's accrued cohort responses. The 24 source values
provide the resampling pools and 1PLD's fixed SD but are not counted
as an initial cohort.

\subsection{Prior specification and computation}
The five models share the target above but require model-specific
calibration of the response scale and residual variation.
For iBART and mBART we retain the 200-tree architecture,
method-specific tree-depth settings, and $k=2$ described in
Appendices~\ref{sec:bart-priors} and~\ref{app:bart-numerics}. For the
normalized AGT score, we use mean center 0.5,
response span 1, and residual-SD reference 0.20 on that scale;
the leaf scale is thereby recalibrated to the score units.
The mean center and span correspond to the score units rather than
transferring the simulation outcome's numerical 2 and 4 unchanged.
The 100 equally spaced interior split points are placed on the
case's physical interval, $20+120j/101$, $j=1,\ldots,100$.
This is an outcome-scale calibration of the BART prior, not a change
to the clinical threshold or dose domain.

For 2PLD and 3PND, we use the normalized hierarchical priors in
Appendix~\ref{app:bayesian-comparators} on the source's raw response
$60-\mathrm{AGT}$, with threshold 55, and use $\delta=0.5$ for 3PND
as in the source's pharmacodynamic example \citep{lee2023}.
This retains the scale-one truncated half-Cauchy prior in its original
response units. Dividing fitted means and SDs by 60 then gives the
common display scale; applying a scale-one half-Cauchy directly to the
score-scale SD would specify a different prior.
All three parametric mean functions are anchored at zero at dose 20;
the tree means are not constrained to that endpoint value.

For 1PLD, the residual SD is estimated once by pooling within-dose
variation across the four groups and then treated as fixed:
\[
 s_{\mathrm{pool}}^2
 =\frac{\sum_{d\in\{40,60,80,100\}}\sum_{i:X_i=d}(Y_i-\bar Y_d)^2}
        {24-4},
 \qquad s_{\mathrm{pool}}=0.1142754.
\]
This is 6.8565 fmol/mg before division by 60.
The untruncated score-scale slope prior is
$\beta\sim\mathcal N(0.005,0.0025^2)$, using the fixed calibration
$a=0.60/(140-20)$ and $b=0.30/(140-20)$.
These are implementation prior constants, not estimates from the AGT
observations. The analytic 1PLD posterior conditions on the plug-in SD
and does not propagate uncertainty from estimating it.

At each cohort, the 2PLD and 3PND posterior calculations use 4,096
sequential Monte Carlo particles with eight rejuvenation sweeps and a
0.75 effective-sample-size resampling threshold. Each tree posterior
is refitted using 1,000 burn-in iterations and 4,000 retained draws;
the 1PLD update is analytic. Independent posterior random seeds are
fixed by method and cohort. Posterior calculation follows
Appendix~\ref{app:bart-numerics}.

\subsection{Sequential EWOC allocation}
\label{sec:bart-case-allocation}
All five methods use the same decision rule and differ only in their
dose--response posterior. Let $q_{m,t}$ be model $m$'s posterior
$\alpha$-quantile of $\xi$ after cohort $t$, with $\alpha=0.05$.
We adopt the published 3PND application's strict upward grid mapping,
while adding a 20\,mg/m$^2$ upward cap and Only Escalation. We first cap
the \emph{continuous} EWOC proposal, then choose the first eligible dose
strictly above it. Specifically, let
\[
\begin{aligned}
u_{m,t+1}&=\min\{q_{m,t},d_{m,t}+20\},\\
\mathcal D_{m,t}&=\{d\in\mathcal D:d\le d_{m,t}+20\},\\
\mathcal D&=\{40,60,80,100\}\,\mathrm{mg/m^2}.
\end{aligned}
\]
Choose $r_{m,t+1}$ as the smallest dose in $\mathcal D_{m,t}$ strictly
above $u_{m,t+1}$; if none exists, use $\max\mathcal D_{m,t}$. Then
$d_{m,t+1}=\max\{d_{m,t},r_{m,t+1}\}$, so the assignment rises by at
most one level and never decreases. Strict rounding advances an
exact-grid proposal when the cap permits it.
The published pool contains no AGT values at 120\,mg/m$^2$, so 100 is
the highest available action. The first cohort is assigned
40\,mg/m$^2$ by protocol. Every method completes ten cohorts; an
AGT-depletion event does not trigger toxicity stopping. This
application uses no rollout.

The posterior $\alpha$-quantile is a continuous-dose EWOC proposal.
Upward grid rounding or holding a previous dose can give an assigned
dose with $\Pr\{p_{\mathrm{dep}}(d_{m,t+1})>0.4\mid\mathcal F_{m,3t}\}>\alpha$.
Consequently, this discrete rule is \emph{EWOC-guided}
but does not guarantee the nominal posterior feasibility bound for
the administered doses. We record every such post-initial exception.

\subsection{Five-method sequential results}
Figure~\ref{fig:bart-case-paths} shows the administered doses by
cohort, and Table~\ref{tab:bart-case-sequential} reports the
last-cohort dose and terminal posterior summary for each method.

\begin{figure}[H]
 \centering
 \includegraphics[width=\textwidth]{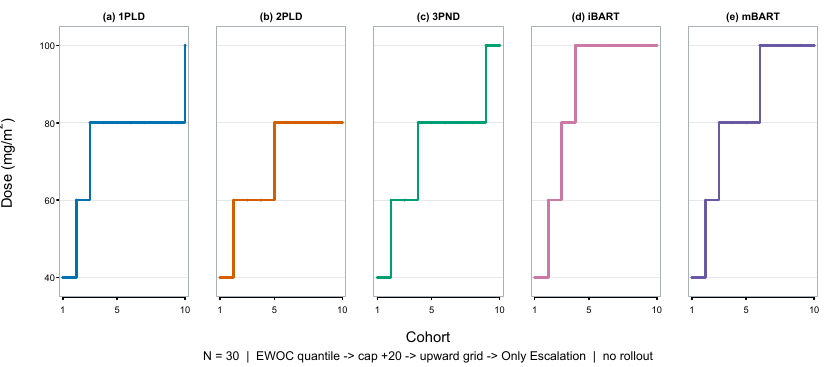}
 \caption{Administered-dose paths for five 30-patient runs under a common
 $\alpha=0.05$ EWOC quantile, a 20\,mg/m$^2$ cap, strict upward-grid
 assignment, and Only Escalation. Table~\ref{tab:bart-case-sequential}
 reports posterior-bound exceptions.}
 \label{fig:bart-case-paths}
\end{figure}

\begin{table}[H]
\centering\small
\caption{Single-path AGT results. Last dose follows 27 outcomes;
terminal $q_{0.05}(\xi)$ uses all 30. ``Bound exceeded'' counts
post-initial assignments with
$\Pr\{p_{\mathrm{dep}}(d_{\mathrm{next}})>0.4\mid\mathcal F\}>0.05$.
AGT $<5$ is pharmacodynamic, not a clinical DLT.}
\label{tab:bart-case-sequential}
\setlength{\tabcolsep}{6pt}
\renewcommand{\arraystretch}{1.12}
\begin{tabular}{@{}lcccc@{}}
\toprule
Method & Last dose & AGT $<5$ & Terminal $q_{0.05}(\xi)$ &
Bound exceeded\\
& (mg/m$^2$) & (of 30) & (mg/m$^2$) & (of 9)\\
\midrule
1PLD  & 100 & 3  & 82.93 & 8\\
2PLD  & 80  & 1  & 76.17 & 9\\
3PND  & 100 & 6  & 85.38 & 9\\
iBART & 100 & 14 & 82.97 & 7\\
mBART & 100 & 10 & 82.97 & 8\\
\bottomrule
\end{tabular}
\end{table}

\FloatBarrier
The methods escalate at different times. 3PND reaches 60, 80, and
100\,mg/m$^2$ in cohorts 2, 4, and 9; iBART, mBART, and 1PLD reach
100 in cohorts 4, 6, and 10; 2PLD ends at 80. These last doses are
allocation outputs, not terminal posterior boundaries. AGT-depletion
counts reflect the distinct paths; a single path per method cannot
establish comparative operating characteristics.

\subsection{Interpretation}
The original glioma trial selected 100\,mg/m$^2$ as a biochemical
modulatory dose using an AGT threshold of 10\,fmol/mg
\citep{friedman1998}. In a 30-patient redesign with a threshold of 5\,fmol/mg,
\citet[pp.~13--15]{lee2023} reported 80\,mg/m$^2$ for standard 3PND
escalation and 100\,mg/m$^2$ for accelerated escalation. Our five
paths use a common resampled response bank, five posteriors, and capped
upward-grid allocation with Only Escalation; they are not numerical
replications of either published trajectory.

The nominal $\alpha=0.05$ posterior bound is exceeded at 7--9 of nine
post-initial assignments per model (Table~\ref{tab:bart-case-sequential}).
For example, 3PND's first continuous proposal is 42.84\,mg/m$^2$;
upward rounding assigns 60, with posterior exceedance probability
0.501. Neither the cap nor Only Escalation restores feasibility at the
\emph{administered} dose. Because AGT depletion is pharmacodynamic and
the true boundary is unknown, we cannot compute NPO, BTM, or RAE
against a known target or infer a validated clinical optimal dose.
Within these limits, the five-model comparison places terminal doses
in the 80--100\,mg/m$^2$ region reported in the earlier AGT analyses,
while showing how the posterior model changes escalation timing and
depletion outcomes.

\FloatBarrier
\section{Software manual}
\label{app:software-manual}

\textsf{\textbf{Dose Trial Lab}} 0.1.4 provides EWOC, rollout, and paired comparisons for
all five models and seven scenarios. The controls illustrate policy selection;
the dashboard and summaries show saved EWOC examples from the linear
scenario with 45 patients. These examples are not simulation-study results.

\begin{figure}[H]
\centering
\begin{minipage}[t]{.36\textwidth}
\centering\textbf{(a) Trial design}\par\smallskip
\includegraphics[width=\linewidth,trim={18bp 374.25bp 825bp 99bp},clip]{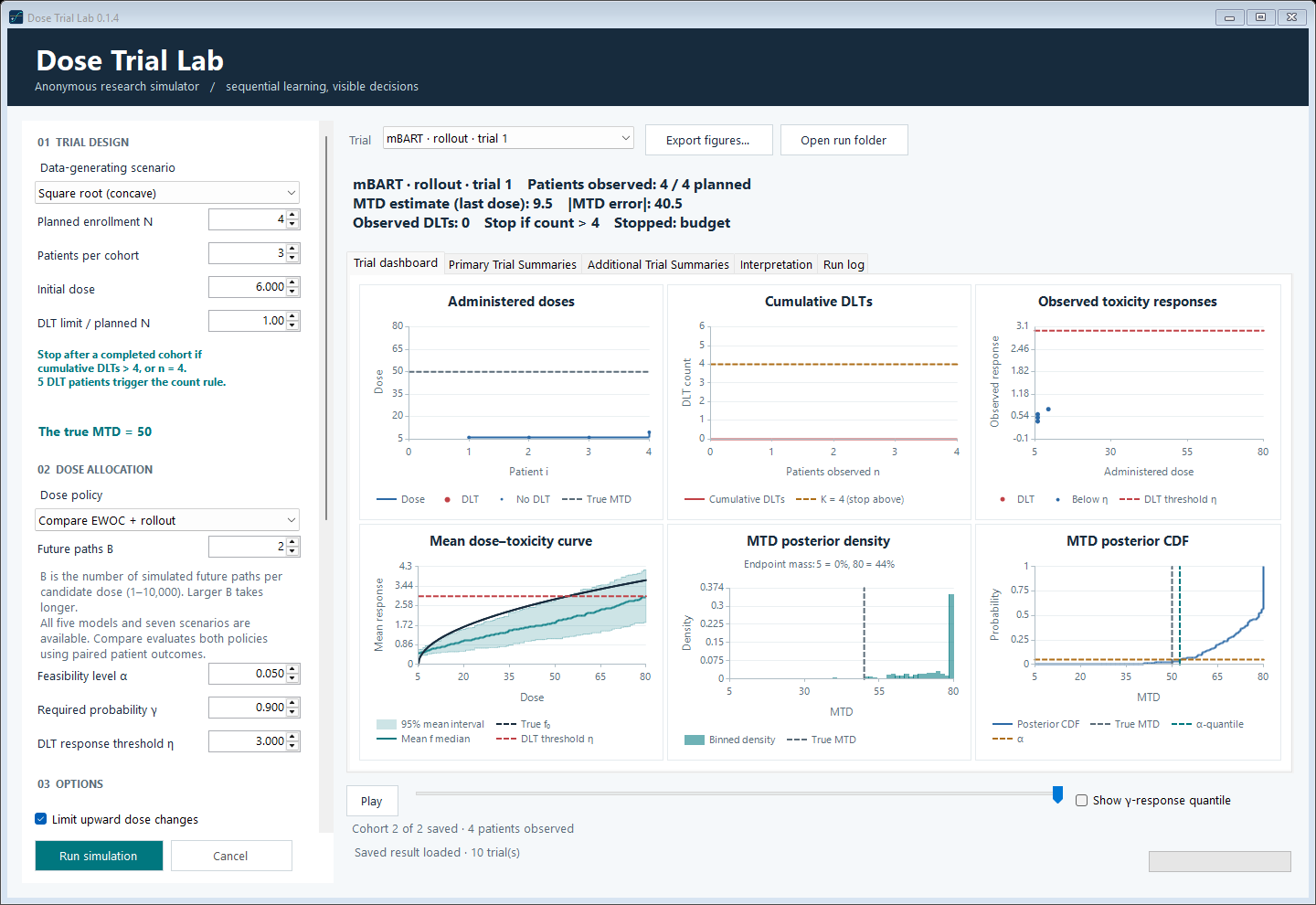}
\end{minipage}\hspace{.08\textwidth}
\begin{minipage}[t]{.36\textwidth}
\centering\textbf{(b) Policy and planning budget}\par\smallskip
\includegraphics[width=\linewidth,trim={18bp 121.5bp 825bp 375bp},clip]{figures/software_014_options.png}
\end{minipage}
\caption{Trial and policy controls. Choose \textbf{EWOC}, \textbf{Rollout},
or \textbf{Compare EWOC + rollout}. \textbf{Future paths B} is the number
of hypothetical continuations per candidate; larger values require more
computation. Scroll down for the dose controls and model choices.}
\label{fig:software-setup}
\end{figure}

\subsection{Install, choose settings, and run}
\label{app:software-installation}
\label{app:software-controls}

\begin{enumerate}
\setlength{\itemsep}{2pt}
\setlength{\parsep}{0pt}
\setlength{\topsep}{4pt}
\item Run \path{DoseTrialLab-0.1.4-Setup.exe} on 64-bit Windows with
.NET Framework 4.8, then open \textsf{\textbf{Dose Trial Lab}}. Other required
components are bundled.
\item Choose the scenario, planned enrollment, cohort size, and initial
dose. Defaults are \(N=45\), cohorts of three, and initial dose 6. Under
these defaults and with DLT-limit fraction 0.10, observe the whole cohort,
then stop at five or more
DLTs or enrollment 45.
\item Choose the \textbf{Dose policy}. For rollout or comparison, set
\textbf{Future paths B} from 1 to 10,000 (default 30). Comparison pairs
patient innovations between policies within each trial. Set the
feasibility level and response thresholds as required.
\item Both dose controls are on by default, with upward cap 3.5.
Turn off \textbf{Only Escalation} to permit dose decreases.
\item Select models and replicates; retain or change the random seed.
Select \textbf{Run simulation}; \textbf{Cancel} interrupts a run.
\end{enumerate}
Standalone EWOC uses manuscript computation settings. Rollout and
comparison use faster settings, matched between comparison policies.
These settings apply to interactive software runs and differ from those
used in the reported experiments.

\subsection{View the trial and replay its updates}
\label{app:software-plots}

Select the model, policy, and replicate from \textbf{Trial}. The
\textbf{Trial dashboard} shows doses, DLTs, responses, the mean curve,
and MTD posterior. Use the cohort slider or \textbf{Play}/\textbf{Pause}
to inspect updates. \textbf{Show \(\gamma\)-response quantile} adds the
optional curve overlay.

\begin{figure}[!ht]
\centering
\includegraphics[width=\linewidth,trim={651bp 213.75bp 51bp 525bp},clip]{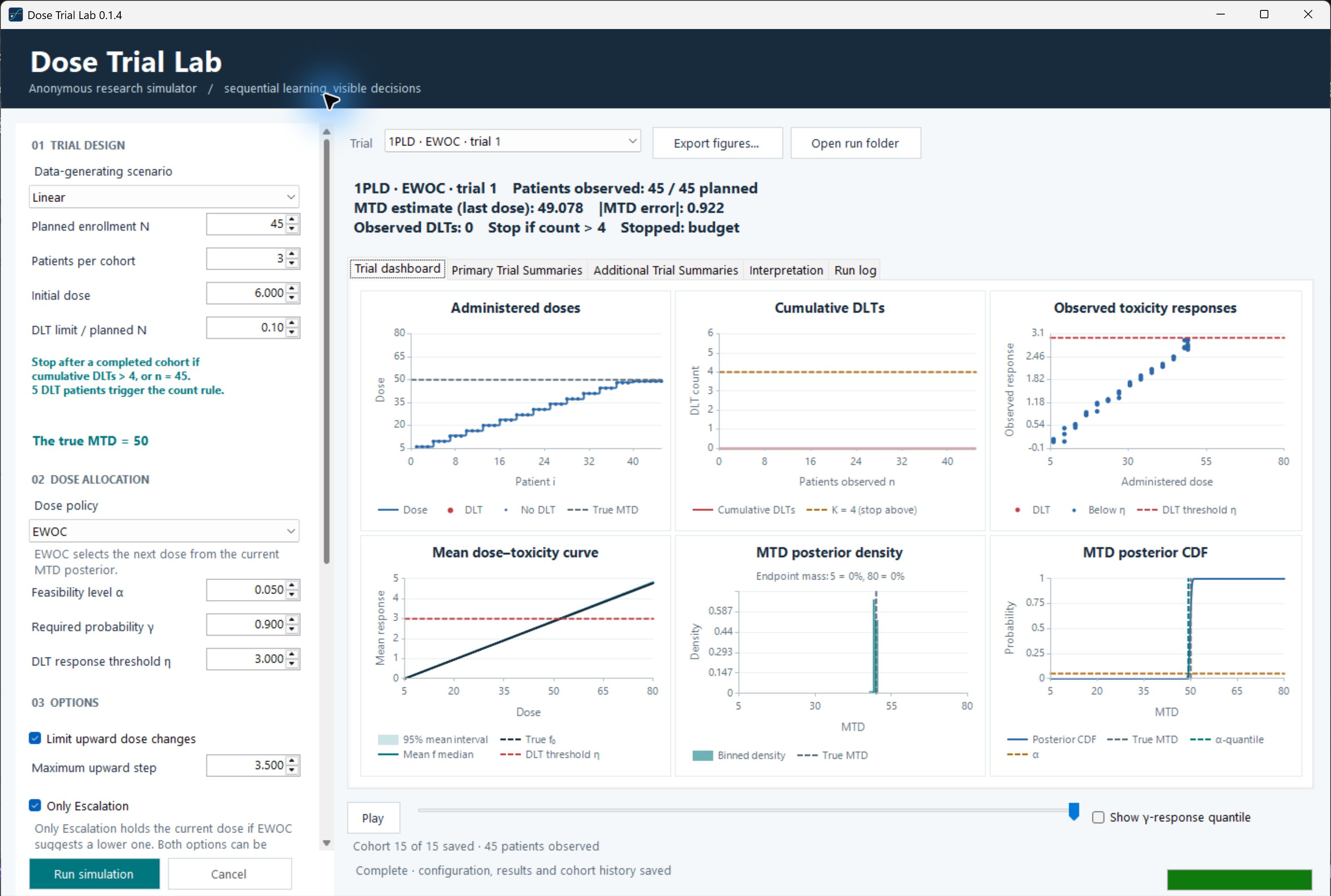}
\caption{Six-panel dashboard for a saved 1PLD EWOC trial under the linear
scenario after 45 patients: administered doses, cumulative DLTs, observed
responses, mean dose--toxicity curve, and MTD posterior density and CDF.
Change \textbf{Trial} or the cohort slider to inspect another saved update.}
\label{fig:software-plots}
\end{figure}

\textbf{Primary Trial Summaries} reports the last-dose estimate, NPD,
NPO, BTM, RAE, weighted loss, enrollment, and stopping reason.
\textbf{Additional Trial Summaries} contains posterior summaries.
The header identifies the selected model, policy, and endpoint.
Paired summaries can coincide when both policies choose the same doses.

\subsection{Export figures and reopen a run}
\label{app:software-output}

Choose \textbf{Export figures} to save all six plots for the selected
trial and cohort. \textbf{Open run folder} locates saved results and
paired exports; \textbf{Load saved result} reopens them.
\textbf{Save configuration} retains the settings.

\begin{figure}[H]
\centering
\includegraphics[width=\textwidth,trim={628.5bp 780bp 159.75bp 224.25bp},clip]{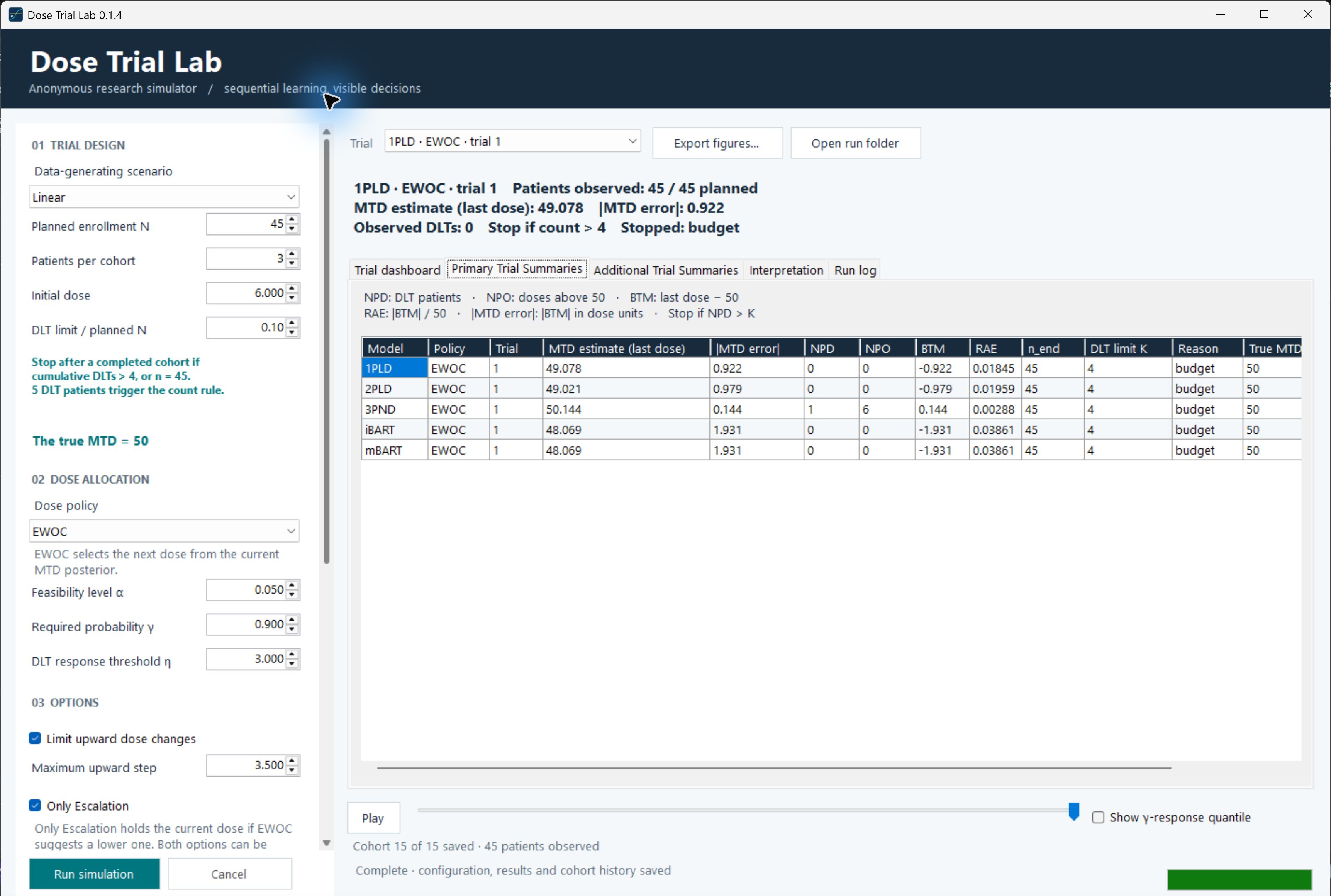}
\caption{EWOC-only \textbf{Primary Trial Summaries} for the five models
in a saved linear-scenario run with 45 patients per trial. The selected
1PLD trial is the dashboard example above. These individual trials
illustrate the interface, not a comparison of average performance.}
\label{fig:software-export}
\end{figure}

\end{document}